\documentclass[12pt]{article}

\usepackage{amsfonts}
\usepackage{amsmath}
\usepackage{amssymb}
\usepackage{amsthm}
\usepackage{algorithm}
\usepackage[noend]{my_algorithmic}
\usepackage[hyphens]{url}
\usepackage{graphicx}
\usepackage{xspace}
\usepackage{units}
\usepackage{caption}

\newtheorem{theorem}{Theorem}[subsection]
\newtheorem{corollary}[theorem]{Corollary}
\newtheorem{proposition}[theorem]{Proposition}
\newtheorem{lemma}[theorem]{Lemma}
\newtheorem{observation}[theorem]{Observation}
\newtheorem{remark}[theorem]{Remark}

\newcommand{\CF}{{\sc Uniform Circle Formation}\xspace}
\newcommand{\UCF}{{\tt UCF}\xspace}

\newcommand{\CO}{{\em Co-radial}\xspace}
\newcommand{\RE}{{\em Regular}\xspace}
\newcommand{\PR}{{\em Pre-regular}\xspace}
\newcommand{\CE}{{\em Central}\xspace}
\newcommand{\EQ}{{\em Equiangular}\xspace}
\newcommand{\BI}{{\em Biangular}\xspace}
\newcommand{\DBI}{{\em Double-biangular}\xspace}
\newcommand{\PER}{{\em Periodic}\xspace}
\newcommand{\UP}{{\em Uni-periodic}\xspace}
\newcommand{\UA}{{\em Uni-aperiodic}\xspace}
\newcommand{\BP}{{\em Bi-periodic}\xspace}
\newcommand{\BA}{{\em Bi-aperiodic}\xspace}
\newcommand{\AP}{{\em Aperiodic}\xspace}
\newcommand{\HD}{{\em Half-disk}\xspace}
\newcommand{\VA}{{\em Valid}\xspace}
\newcommand{\IN}{{\em Invalid}\xspace}
\newcommand{\WA}{{\em Waiting}\xspace}
\newcommand{\RD}{{\em Ready}\xspace}

\newcommand{\ASYNC}{${\cal ASYNC}$\xspace}
\newcommand{\SSYNC}{${\cal SSYNC}$\xspace}
\newcommand{\FSYNC}{${\cal FSYNC}$\xspace}

\begin{document}

\title{Distributed Computing by Mobile Robots:\\
Uniform Circle Formation}

\author{
Paola Flocchini \footnotemark[1]\
\and
Giuseppe Prencipe \footnotemark[2]\
\and
Nicola Santoro \footnotemark[3]\
\and
Giovanni Viglietta \footnotemark[1]\
}

\renewcommand{\thefootnote}{\fnsymbol{footnote}}
\footnotetext[1]{School of Electrical Engineering and Computer Science, University of
Ottawa, {\tt flocchin@site.uottawa.ca}, {\tt viglietta@gmail.com}}
\footnotetext[2]{Dipartimento di Informatica,
Universit\`a di Pisa, {\tt prencipe@di.unipi.it}}
\footnotetext[3]{School of Computer Science, Carleton University,
  {\tt santoro@scs.carleton.ca}}
  
\renewcommand{\thefootnote}{\arabic{footnote}}	

\date{}

\maketitle
\begin{abstract}
\noindent Consider a set of $n$ simple autonomous mobile robots (asynchronous, no common coordinate system, no identities, no central coordination, no direct communication, no memory of the past, non-rigid, deterministic) initially in distinct locations, moving freely in the plane and able to sense the positions of the other robots. We study the primitive task of the robots arranging themselves on the vertices of a regular $n$-gon not fixed in advance ({\sc Uniform Circle Formation}). In the literature, the existing algorithmic contributions are limited to conveniently restricted sets of initial configurations of the robots and to more powerful robots. The question of whether such simple robots could deterministically form a uniform circle has remained open. In this paper, we constructively prove that indeed the {\sc Uniform Circle Formation} problem is solvable for any initial configuration in which the robots are in distinct locations, without any additional assumption (if two robots are in the same location, the problem is easily seen to be unsolvable). In addition to closing a long-standing problem, the result of this paper also implies that, for pattern formation, asynchrony is not a computational handicap, and that additional powers such as chirality and rigidity are computationally irrelevant.
\end{abstract}

\section{Introduction}\label{sec:intro}
Consider a set of  punctiform  computational  entities, called {\em robots}, located  in ${\mathbb R}^2$,
where they can freely move.
Each entity is
provided with a local coordinate system and
 operates in   {\em Look-Compute-Move}  cycles.
During a cycle, a robot obtains a snapshot of the positions of   the other robots, expressed in  its own  coordinate system ({\em Look}); using the snapshot as an input, it
executes a deterministic algorithm (the same for all robots)  
to determine a destination
  ({\em Compute});
  and it
moves towards the computed destination along a straight line ({\em Move}).

To understand the nature of  the distributed universe of these mobile robots
and to discover its computational boundaries, the  research efforts
have focused on the minimal capabilities the robots need to have to be able to solve a problem. 
Thus,  the extensive literature on distributed computing by  mobile robots has 
almost exclusively focused on very simple entities operating in strong adversarial conditions.
The robots   we consider are
    {\em anonymous}  (without ids or distinguishable features), {\em  
autonomous} (without central or external control),
     {\em oblivious} (no recollection of computations and observations  
done in previous cycles),
  {\em disoriented} (no agreement among the individual coordinate  
systems, nor on unit  distance and chirality), and {\em non-rigid} (they may be stopped before reaching the destination they compute at each cycle).
  In particular, the choice of individual coordinate systems, the
  activation schedule, the duration of each operation during a
cycle, and the length traveled by a robot during its
  movement are determined by an adversary; the only constraints on the  
adversary are fairness (i.e., the duration of each cycle of each robot is arbitrary but finite), and minimality (i.e., there  
exists $\delta>0$, unknown to the robots, such that,
  if the destination of a robot is at distance at most $\delta$, the robot will  
reach it; else it will move at least $\delta$ towards the
  destination, and then it may be unpredictably stopped by the adversary).
  For this type of robots, depending on the activation schedule and timing assumptions,  three
  main models  
  have been studied in the literature: the   {\em asynchronous} model, {\ASYNC}, where  no
  assumptions are made on synchronization among the robots' cycles nor their duration, and 
 the  {\em semi-synchronous} and {\em fully synchronous} models,
   denoted by    \SSYNC and  \FSYNC, respectively,
    where  the robots,  while oblivious and disoriented, 
    operate in synchronous rounds, and each round is ``atomic":
    all robots active in that round terminate their cycle by the next round;
 the only difference   is whether  all robots  are activated in every round (\FSYNC),
   or,  subject to some fairness condition, a possibly different  subset is activated in each round (\SSYNC). 
All three models have been intensively studied (e.g., see~\cite{ChaMN04,CieFPS12,CohP05,DaFPS14,DefK02,DefS08,DieLP08,DieLPV08,DieP07a,FloPSW08,FujYKY12,KamLOT11,SuzY99,YamS10};
for a detailed overview refer  to the recent monograph~\cite{FloPS12}). 

The research on the  {\em computability} aspects  has
focused     almost exclusively   on
 the  fundamental class of
{\sc Geometric Pattern Formation} problems.
  A {\em geometric pattern} (or simply {\em pattern}) $P$
is a set of points in the plane;
the robots   {\em form}  the pattern $P$  at   time $t$ if  the  
configuration of the robots (i.e., the set of their positions)
  at time $t$ is similar to $P$ (i.e., coincident with $P$ up to
scaling, rotation, translation, and reflection).  
A pattern $P$ is {\em formable}   if there exists an algorithm  
that  allows the robots to   form $P$   within finite time and  no longer move,
regardless of the activation scheduling and delays (which, recall,   
are decided by the adversary) and of the initial placement of
the robots in distinct points.
Given a model, the research questions are:  to determine if 
 a given  pattern $P$ is  
 formable in that model;
if so, to design an  algorithm that will allow its formation;
and, more in general,
to fully characterize the set  of  patterns formable in that model.
The  research effort has focused  on answering these questions 
for {\ASYNC} and less demanding models 
both
in general  (e.g.,~\cite{DaFPS14,FloPSW08,FujYKY12,SugS96,SuzY99,YamS10})
and
for specific classes of patterns (e.g.,~\cite{ChaMN04,DefS08,DieLP08,DieP07a,DieP08,FloPS08,Kat05,MiyIH09}).

Among  specific patterns, a special research place  is occupied by two classes:
   {\tt Point}  and  {\tt Uniform Circle}.
   The class   {\tt Point} is the set consisting of a single point;
  point formation 
corresponds to the important {\sc Gathering} 
problem  requiring all robots to gather at a same location,
not determined in advance (e.g., see~\cite{BraT15,CieFPS12,CohP05,CohP08,IzuSK+12,OasSY97}).
The other important class of patterns  is {\tt Uniform Circle}: the  
points of the pattern form the vertices of a regular
$n$-gon, where $n$ is the number of robots (e.g.,~\cite{ChaMN04,DefK02,DefS08,DieLP08,DieP07a,DieP08,FloPS08,MiyIH09}).

In addition to their relevance as individual problems, the classes {\tt Point} and {\tt Uniform Circle}
play another  important role.  
A crucial observation, by  Suzuki and Yamashita~\cite{SuzY99}, 
is that formability  of  a pattern $P$ from an initial  configuration $\Gamma$  in model ${\cal M}$
depends on the relationship between $\rho_{\cal M}(P)$ and
$\rho_{\cal M}(\Gamma)$,  where $\rho_{\cal M}(V)$ is a special parameter, called  {\em symmetricity}, of a multiset of 
points $V$, interpreted as robots modeled by ${\cal M}$. Based on this observation, it follows that
the only patterns that {\em might} be formable from any initial configuration in \FSYNC (and thus also in
\SSYNC and \ASYNC) are  single points and 
regular polygons (also called uniform circles). 
It is rather easy to see that  both points and 
uniform circles
 can  be formed in \FSYNC, i.e.,  if the robots are fully
synchronous.
After a  long quest by several researchers, it has been  shown that
{\sc Gathering} is solvable (and thus {\tt Point} is formable)  in \ASYNC (and thus also in \SSYNC)~\cite{CieFPS12}, 
leaving open only the question of whether {\tt Uniform Circle} is formable in these models.
In  \SSYNC, it was known  that the robots can {\em converge} towards a
uniform circle without ever forming it~\cite{DefS08}. 
Other results indicate
that the robots can actually form {\tt Uniform Circle} in   \SSYNC. 
In fact, by concatenating  the algorithm of~\cite{Kat05},  for forming a biangular configuration,
 with  the one  of~\cite{DieP07a}, for circle formation from a
 biangular starting configuration, it is 
possible
to form {\tt Uniform Circle} starting from any initial configuration in \SSYNC (the case with four robots has been solved separately in~\cite{DieP08}). Observe, however, that the
two algorithms can be concatenated only because the robots are semi-synchronous.
 Hence, the outstanding  question is whether  it is possible to form {\tt Uniform Circle} in \ASYNC.
 
 In spite of the simplicity of its formulation and the repeated efforts by several researchers,
  the existing algorithmic contributions are  
limited to restricted sets of initial configurations of the robots
and to more powerful robots. 
In particular, it has been proven that, with the additional property of {\em chirality} (i.e., a common notion of ``clockwise"),  the robots can form {\tt Uniform Circle}~\cite{FloPS08}, and  with a very simple algorithm;
the fact that  {\tt Uniform Circle} is formable in  \ASYNC+{\em chirality} follows also
from the recent general result  of~\cite{FujYKY12}. 
The difficulty of the problem stems from the fact that the  inherent difficulties of asynchrony, obliviousness, and
disorientation  are amplified by their simultaneous presence.

A step toward the solution has been made in~\cite{FloPSV14}, where the authors solved the problem assuming that the robots had the ability to move along circular arcs, as well as straight lines.

In this paper we  show that indeed the
{\sc Uniform Circle Formation} problem is solvable for any initial  
configuration of robots (located in distinct positions) without any additional assumption,
thus closing a problem that has been open for over a decade.
This result also implies that, for {\sc Geometric Pattern Formation} problems, \emph{asynchrony} is not a computational
 handicap, 
 and that additional  powers such as
 {\em chirality} and {\em rigidity}
 are computationally irrelevant.
 
The paper is structured as follows. In the next Section, the model and the terminology are introduced. In Section~\ref{sec:informal}, we describe the ideas behind our  solution in an informal way. We  provide the rigorous and formal presentation of the algorithm  in Section~\ref{sec:formal}.
 We then give the formal proof of correctness in Section~\ref{sec:correctness}.
 
\section{Model and Terminology}\label{sec:model}
The system consists of a \emph{swarm} $\mathcal{R} = \{ r_1 ,\cdots,
r_{n}\}$ of \emph{mobile robots}, which are computational entities moving and operating
 in the Euclidean plane $\mathbb R^2$.
Each robot can move freely and continuously in the plane, and
operates  in \emph{Look-Compute-Move} cycles. 

\paragraph{Look, Compute, and Move phases.} The three phases of each cycle are as follows.
\begin{enumerate}
\item In the Look phase, a robot takes an instantaneous snapshot of the positions of all robots in the swarm. This snapshot is expressed as an $n$-uple of points in the robot's coordinate system, which is an orthogonal Cartesian system whose origin is the robot's current location.
\item In the Compute phase, a robot executes a deterministic algorithm, which is the same for all robots, and computes a destination point in its own coordinate system. The only input to such an algorithm is the snapshot taken in the previous Look phase.
\item In the Move phase, a robot moves toward the destination point that it computed in the previous Compute phase. At each instant, the velocity of the robot is either null or it is directed toward the destination point.
\end{enumerate}
After a Move phase is done, the next cycle begins with a new Look phase, and so on.

The robots are \emph{anonymous}, which means that they are indistinguishable and do not have identifiers. This translates into the fact that the snapshot a robot takes during a Look phase is simply a set of points, with no additional data. Since the origin of a robot's local coordinate system is always the robot's current location, each snapshot will always contain a point with coordinates $(0,0)$, representing the observing robot itself.

Robots are also \emph{oblivious}, meaning that they do not retain any memory of previous cycles. This translates into the fact that the only input to the algorithm executed by a robot in a Compute phase is just the last snapshot that the robot took. Similarly, we can say that the robots are \emph{silent}, in that they have no means of direct communication of information to other robots.

Different robots' coordinate systems may have different units of distance, different orientation, and different handedness. A robot's coordinate system may even change from one cycle to the next, as long as its position stays at the origin.

The operations that can be executed by a robot in the Compute phase are limited to algebraic functions of the points in the input snapshot. We assume that computations of algebraic functions can be performed in finite time with infinite precision.

The robots are \emph{asynchronous}, meaning that the duration of each cycle of each robot is completely arbitrary (but finite) and independent of the cycles of the other robots. In particular, a robot may perform a Look phase while another robot is in the middle of a movement. Also, from the time a robot takes a snapshot to the time it actually moves based on that snapshot, an arbitrarily long time may pass. This means that, when the robot actually moves, it may do so based on a very old and ``obsolete'' observation. The entity that decides the duration of each robot's cycles is the \emph{scheduler}. We may think of the scheduler as an ``adversary'' whose goal is to prevent the robots from performing a certain task.

During a Move phase, a robot moves directly toward the destination point that it computed in the previous Compute phase, along a line segment. In particular, it cannot move backwards on such a line. However, there are no assumptions on the robot's speed, and the speed may also vary arbitrarily during the Move phase. A robot can even occasionally stop and then move again (toward the same destination point) within the same Move phase. Again, the speed of the robot at each time is decided by the scheduler. The scheduler may also prevent a robot from reaching its destination point, by stopping it in the middle of the movement and then ending its Move phase. This model is called \emph{non-rigid} in the literature (as opposed to the \emph{rigid} model, in which a robot is always guaranteed to reach its destination by the end of every Move phase). The only constraint that we pose on the scheduler is that it cannot end a robot's Move phase unless the robot has moved by at least a positive constant $\delta$ during the current cycle, or it has reached its destination point. This $\delta$ is measured in a universal coordinate system (i.e., not in a robot's local coordinate system), and it is an absolute constant that is decided by the scheduler once and for all, and cannot be changed for the entire execution. We stress that the value of $\delta$ is not known to the robots, as it is not part of the input to the algorithm executed in the Compute phase.\footnote{The value of $\delta$ is assumed to be the same for all robots. However, since the robots are finitely many, nothing changes if each robot has a different $\delta$: all the executions in this model are compatible with a ``global'' $\delta$ that is the minimum of all the ``local'' $\delta$'s.}

The scheduler also decides the robots' initial positions in the plane (i.e., at time $t=0$), with the only constraint that they must be $n$ distinct locations (i.e., no two robots can occupy the same location, initially). We assume that initially the robots are not moving, and are waiting to be activated by the scheduler. When the scheduler activates a robot for the first time, it starts with a Look phase, and then proceeds normally. Different robots may perform the first Look phase at different times.

Note that, without loss of generality, we may assume that each cycle's Look and Compute phases are executed at the same time, instantaneously. Indeed, we can ``simulate'' a delay between the two phases by making a robot stay still for a while at the beginning of the next Move phase. Note that some authors also distinguish a \emph{Wait} phase, which occurs just before a Look. Again, this phase can be easily incorporated into the previous Move phase. Hence, in this paper, we will refer to only two phases: an instantaneous Look-Compute phase, and a Move phase, in which the moving robot may also stay still for arbitrarily long (but finite) periods of time.

\paragraph{Executions and properties.}
Let a swarm of $n$ robots operate according to an algorithm $\mathcal A$, starting from an initial configuration $I$, and with minimality constant $\delta$ (as defined above). We call \emph{execution} the sequence of configurations formed by the robots as a function of time, which depends on how the adversary activates the robots, and includes each robot's phase at each time. We denote by $\mathcal E^\delta_{I,\mathcal A}$ the set of all possible executions of such a swarm.
Note that, if $0<\delta'\leqslant\delta$, then $\mathcal E^\delta_{I,\mathcal A}\subseteq \mathcal E^{\delta'}_{I,\mathcal A}$. Since $\delta$ is not known to the robots, it makes sense to consider the set $\mathcal E_{I,\mathcal A}=\bigcup_{\delta>0} \mathcal E^\delta_{I,\mathcal A}$ as the class of all possible executions, regardless of how small the constant $\delta$ is. Similarly, we define $\mathcal E_{\mathcal A}=\bigcup_{I} \mathcal E_{I,\mathcal A}$ as the class of all possible distributed executions of algorithm $\mathcal A$, regardless of the initial position of the $n$ robots (as long as they are in distinct locations).

We call  \emph{property} any Boolean predicate on sequences of configurations.  We say that $\mathcal E^\delta_{I,\mathcal A}$ \emph{enjoys} property $\mathcal P$ if 
$\mathcal P$ is true for all executions in $\mathcal E^\delta_{I,\mathcal A}$. 

\paragraph{Trajectories and frozen configurations.}
For a given execution, we denote by $r(t)$ the position of robot $r\in\mathcal R$, expressed in a global coordinate system, at time $t\geqslant 0$. If $r$ is in a Look-Compute phase (respectively, in a {Move} phase) at time $t$, then the \emph{trajectory} of $r$ at time $t$ is the set consisting of the single point $r(t)$ (respectively, the segment with endpoints $r(t)$ and the destination point of $r$ at time $t$).

A robot is said to be \emph{frozen} at time $t$ if its trajectory at time $t$ is $\{r(t)\}$. The swarm $\mathcal R$ is said to be \emph{frozen} at time $t$ if every robot in $\mathcal R$ is frozen at time $t$. If the robots in the swarm reach a frozen configuration at time $t$, they are said to \emph{freeze} at time $t$. Recall that we assume the swarm to be frozen initially, i.e., at time $t=0$.

\paragraph{The \CF problem.} We may equivalently regard a property of executions as a set of ``behaviors'' that the robots may have. Assigning a \emph{task}, or a \emph{problem}, to a swarm of robots is the same as declaring that some behaviors are ``acceptable'', in that they attain a certain goal, and all other behaviors are ``unacceptable''. Hence, we can define a problem in terms of the  property that the  executions must satisfy.
Now, given a problem, expressed as a property $\mathcal P$ of executions, we say that algorithm $\mathcal A$ \emph{solves} the problem if $\mathcal E_\mathcal A$ enjoys $\mathcal P$. 

In this paper we will consider the \CF problem, defined as the property $\mathcal U$ which is true only for those executions for which there is a time $t^*$ such that the robots are frozen at the vertices of a regular $n$-gon at every time $t\geqslant t^*$. In the following, we will describe the algorithm \UCF, and we will prove that it solves the \CF problem.

Note that we insisted on having only initial configurations with robots in distinct locations because otherwise the \CF problem would be unsolvable. Indeed, if two robots are initially coincident, the scheduler can force them to remain coincident for the entire execution (by giving them the same coordinate system and activating them synchronously). For the same reason, in our \UCF algorithm we never allow two robots to collide, although this is not explicitly imposed by the problem's definition.

\section{The Algorithm: Informal Description}\label{sec:informal}

The general idea of the algorithm, called \UCF,  is  rather simple. 
Its implementation is however complicated 
by many technical details, which make the overall strategy quite 
involved and the correctness proof very complex.

Consider  the case with $n>5$ robots.  Recall that the goal of the robots is to position themselves on the vertices of a regular $n$-gon, and stop moving. We call this type of configuration \RE.
Our general strategy  is to have the robots  move  
 to the smallest enclosing circle (SEC); once there, determine their   final target points, 
and then   move to their target points.
The only exception to this procedure
is when the robots form, either ``intentionally'' or ``accidentally'', a special type of configuration called \PR, in which case they follow a special procedure.

In the following we describe the ideas behind our  solution in an informal way.

\subsection{Special Cases: \BI and \PR Configurations}\label{sec:bipr}
Consider first a very special class of configurations in which the robots may be found: the \BI configurations, exemplified in Figure~\ref{f01}(a). A \BI configuration can be defined as one consisting of an even number $n$ of robots, and having exactly $n/2$ axes of symmetry. Note that a \BI configuration can be partitioned into two \RE configurations of equal size.
In this situation, the robots may all have exactly the same view of the environment, provided that their axes are oriented symmetrically. Hence the scheduler may force all of them to perform the same   computation and 
then move  at the same time, which will force the configuration to remain \BI at all times (or become \RE). 
In this scenario, the algorithm must ensure 
that a common computation  and   simultaneous movements  would
result in the   formation of a \RE configuration.
 On the other hand, because of asynchrony while moving towards this goal  the robots may also
 form  different and possibly asymmetric intermediate configurations. Therefore, it is clearly desirable that the robots preserve some invariant so that any such intermediate
 configuration  is  treated coherently  to the \BI case.
A solution to the problem of forming a regular polygon starting from a \BI configuration is described in~\cite{DieP07a}, 
where the robots can identify a ``supporting regular polygon'' (see Figure~\ref{f01}(b)),
and each robot  moves towards the closest vertex of such a polygon. Any intermediate configuration possibly formed while the robots move asynchronously and independently
 towards  the vertices of the supporting polygon is called \PR (note that all \BI configurations are also \PR). While executing this procedure from a \PR configuration, the supporting polygon remains invariant (e.g., see Figure~\ref{f01}(c)). So, whenever the configuration is perceived as \PR by \emph{all} the robots, moving towards the appropriate vertex  of the supporting polygon results in the formation of a \RE configuration. In Lemma~\ref{l:uniquepr} we will prove that, if $n>4$ and a supporting polygon exists, then it is unique.

\begin{figure}[h!]
\centering
\begin{tabular}{c@{\qquad}c@{\qquad}c}
\includegraphics[scale=0.45]{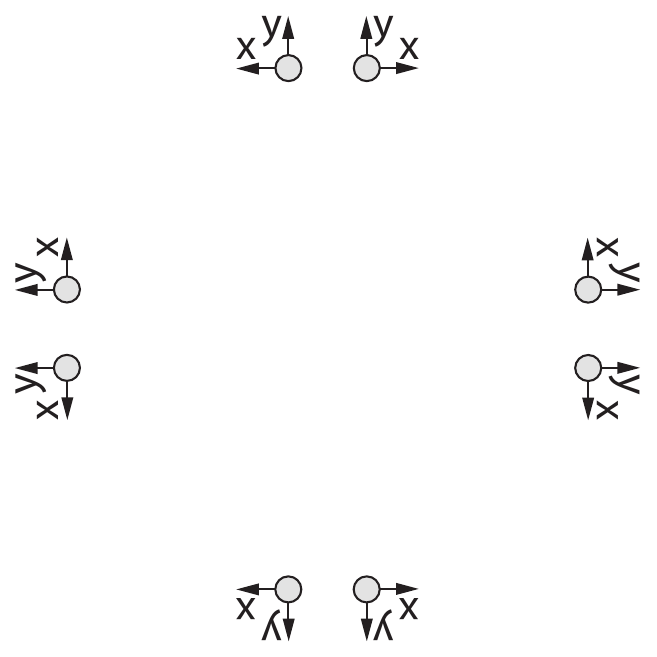} &
\includegraphics[scale=0.45]{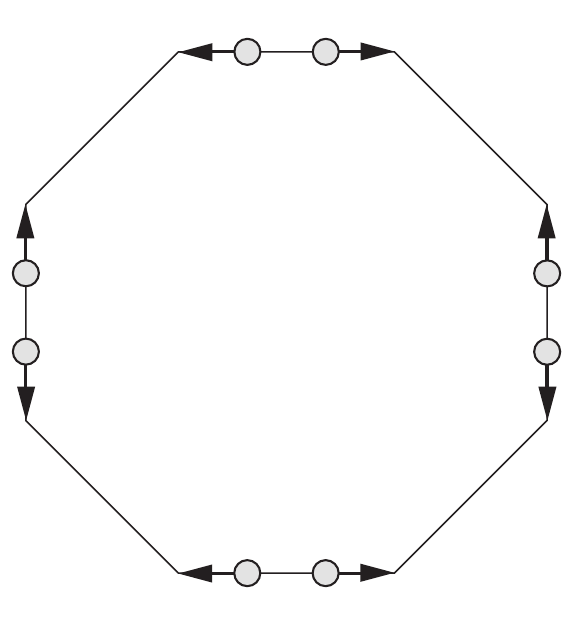} &
\includegraphics[scale=0.45]{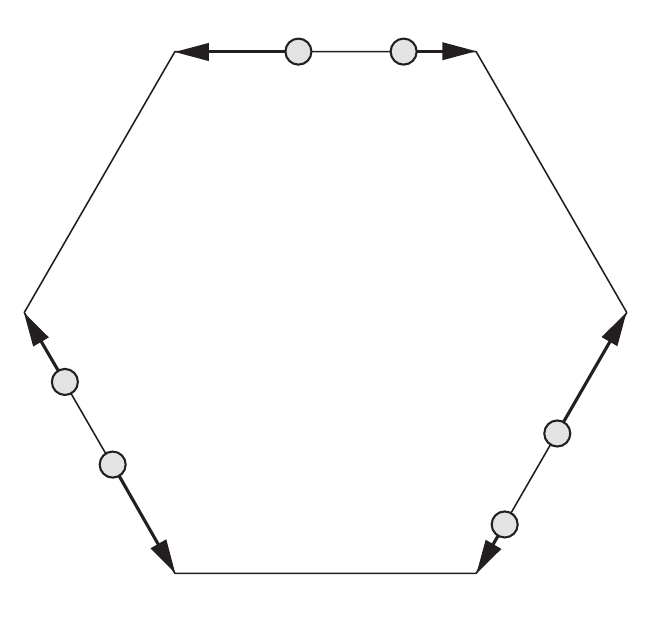} \\
(a) & (b) & (c)
\end{tabular}
\caption{(a) A \BI configuration, with local axes oriented in such a way that all robots have the same view. (b) The correct way to resolve a \BI configuration. (c) A generic \PR configuration with its supporting polygon, which remains invariant as the robots move according to the arrows.}
\label{f01}
\end{figure}

\subsection{General Strategy: SEC and Analogy Classes}

Consider now a starting position of the robots that is not \PR (and hence not \BI). 
Recall that the robots have no common reference frame, and there are no ``environmental'' elements that can be used by the robots to orient themselves. This is a serious difficulty that may prevent the robots from coordinating their movements and act ``consistently'' from one cycle to another. To overcome this difficulty, we identify the smallest enclosing circle (SEC) of the robots' positions (as shown in Figure~\ref{f02}(a)), and we make sure the robots move in such a way as to keep SEC fixed (note that SEC is unique and it is easy to compute). This will hold true as long as the configuration is not \PR. If the configuration happens to become \PR during the execution, then the procedure of Section~\ref{sec:bipr} will be executed, and SEC will no longer be preserved.

The general algorithm will attempt to make all robots reach the perimeter of SEC, as a preliminary step. So, let us consider a configuration that is not \PR and in which all robots lie on the perimeter of SEC. In this situation, we identify pairs of robots that are located in ``symmetric'' positions, i.e., such that there is an isometry of the swarm that maps one of the two robots into the other. We call two such robots \emph{analogous}, and the swarm is thus partitioned into \emph{analogy classes} of analogous robots (see Figure~\ref{f02}(b)). In general, an analogy class has either the shape of a \RE set or of a \BI set (with some degenerate cases, such as a single point or a pair of points).

Similarly to the \BI case (cf.~the discussion in Section~\ref{sec:bipr}), the scheduler may force all the robots in an analogy class to perform the same computation and  move  at the same time, thus occupying symmetric positions again, and potentially forever.  To accommodate this, we may as well incorporate this type of behavior into the algorithm, and make all analogous robots always \emph{deliberately} move together in the same fashion.

\begin{figure}[h!]
\centering
\begin{tabular}{c@{\quad}c@{\quad}c}
\includegraphics[scale=0.5]{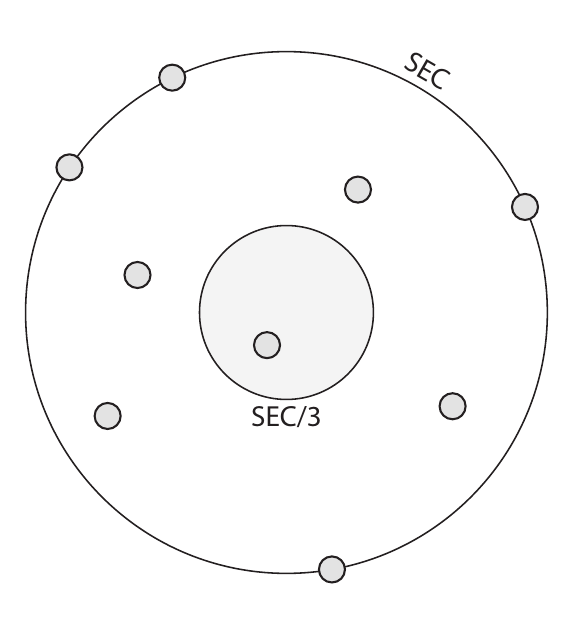} &
\includegraphics[scale=0.5]{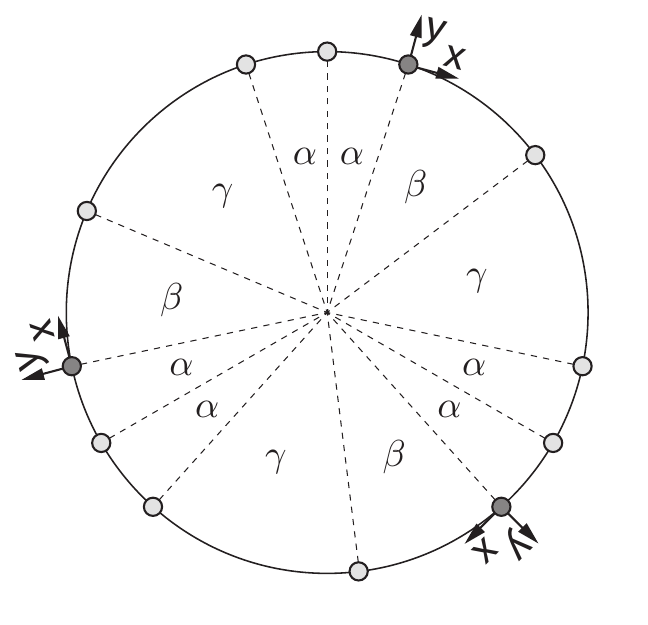} &
\includegraphics[scale=0.5]{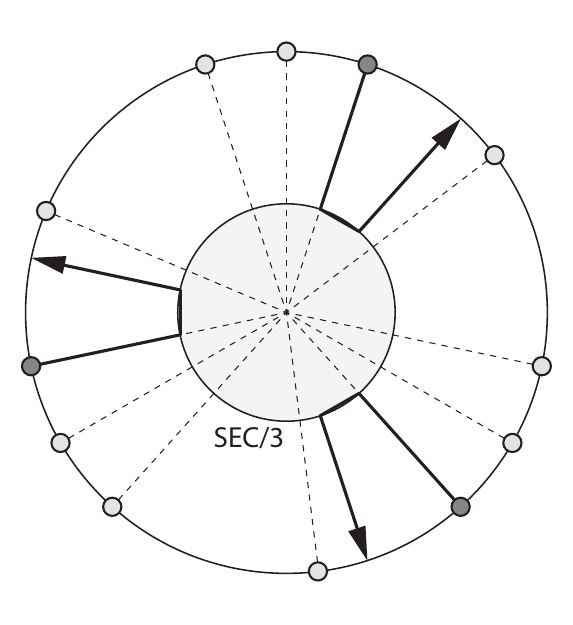} \\
(a) & (b) & (c)
\end{tabular}
\caption{(a) A swarm of robots, with its SEC and SEC/3. (b) The three highlighted robots form an analogy class. If their axes are oriented as indicated, the three robots have the same view. (c) The three dark-shaded robots are selected as walkers, and move according to the arrows. At the end of the move, each walker has an angular distance of $\pi/3$ (which is a multiple of $2\pi/n$) from a non-walker.}
\label{f02}
\end{figure} 

More specifically, we will let only one analogy class move at a time, while all the others wait on SEC (see Figure~\ref{f02}(c)). The robots in the analogy class that is allowed to move are called \emph{walkers}. When the walkers have been chosen, they move radially to SEC/3, which is the circle concentric with SEC and whose radius is $1/3$ of the radius of SEC. Once they are all there, they move to their \emph{finish set}, while staying within SEC/3 (or in its interior). When they are all in their finish set, they move radially to SEC again. Subsequently, a new analogy class of walkers is chosen, and so on. The walkers and the finish set are chosen in such a way that, when the walkers are done moving, some kind of ``progress'' toward a \RE configuration is made. By ``progress'' we mean, for instance, that two analogy classes merge and become one, or that the angular distance between two robots on SEC becomes a multiple of $2\pi/n$ (note that in a \RE configuration all angular distances are multiples of $2\pi/n$).

Of course, as the walkers move to some other location, they all need a strategy to ``wait for each other'', and make sure to reach a configuration in which they are once again analogous. Also, different analogy classes should plan their movements ``coherently'', in such a way that their combined motion eventually results in the formation of a \RE configuration. Note that this is complicated by the fact that, when a class of walkers starts moving, some of the ``reference points'' the robots were using to compute their destinations may be lost. Moreover, it may be impossible to select a class of walkers in such a way that some ``progress'' is made when they reach their destinations, and in such a way that SEC does not change as they move. In this case, the configuration is \emph{locked}, and some special moves have to be made.
Finally, as the robots move according to the general  algorithm we just outlined, they may form a \PR configuration ``by accident''. When this happens, the robots need a mechanism to stop immediately and start executing the procedure of Section~\ref{sec:bipr} (note that some robots may be in the middle of a movement when a \PR configuration is formed accidentally).

All these aspects will be discussed in some detail in this section. Next we will show how the robots can reach SEC from any initial configuration, as a preliminary step.

\subsection{Preliminary Step: Reaching SEC} 

A simple way to make all robots reach SEC without colliding is to make each of them move radially, away from the center, as in Figure~\ref{f03}(a). This works nicely, as long as no two robots are \emph{co-radial}, i.e., collinear with the center of SEC. A special case is the \CE configuration, in which one robot lies at the center of SEC. \CE configurations are easily resolved, by simply making the central robot move to SEC/3, in such a way as not to become co-radial with any other robot.

\begin{figure}[h!]
\centering
\begin{tabular}{c@{\qquad}c@{\qquad}c}
\includegraphics[scale=0.5]{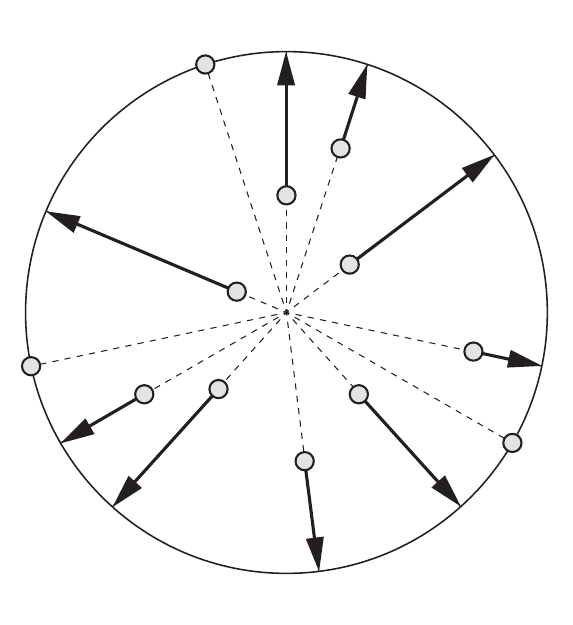} &
\includegraphics[scale=0.5]{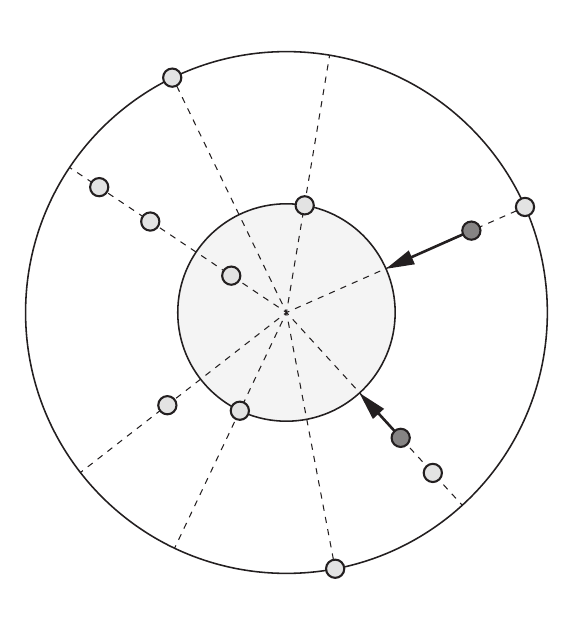} &
\includegraphics[scale=0.5]{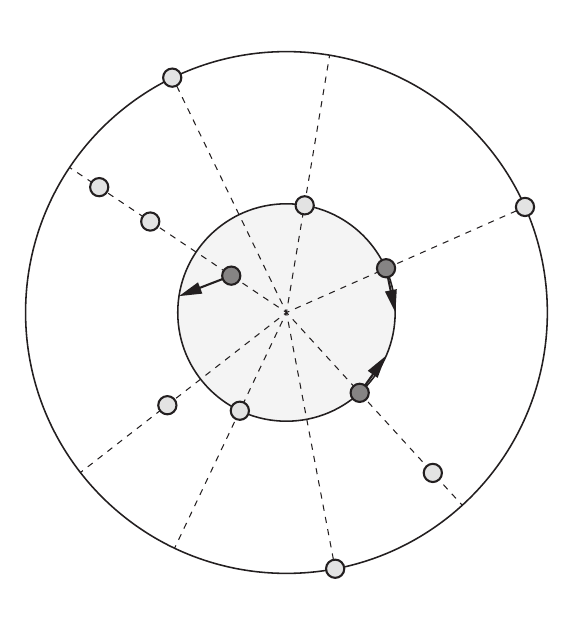} \\
(a) & (b) & (c)
\end{tabular}
\caption{(a) All robots move radially to reach SEC. (b) The most internal co-radial robots move radially to SEC/3. (c) When they are in SEC/3, they make a small lateral move.}
\label{f03}
\end{figure}

The \CO configurations that are not \CE are handled as follows. First of all, if there are non-co-radial robots that are in the interior of SEC/3, they move radially to SEC/3 (note how the evolution of a \CE configuration nicely blends with this). Then, the co-radial robots that are closest to the center of SEC move radially toward the center, until they are in SEC/3 (see Figure~\ref{f03}(b)). Finally, the most internal co-radial robots make a lateral move to become non-co-radial, as in Figure~\ref{f03}(c). The lateral move is within SEC/3 (or its interior) and it is ``sufficiently small'', in order to prevent collisions. A sufficiently small move is, for instance, a move that reduces the angular distance to any other robot by no more than 1/3.

The reason why we make robots reach SEC/3 before performing lateral moves is because we want to prevent the accidental formation of \PR configurations. We will discuss this aspect later, in Section~\ref{introAcc}.

It is easy to see how this strategy makes the robots coordinate their movements and avoid collisions. Indeed, as soon as a robot $r$ makes a lateral move and stops being co-radial, it is seen by the other robots as a non-co-radial robot lying in the interior of SEC/3. Hence, no other robot will take initiatives, and will just wait until $r$ has reached SEC/3 and has stopped there. This guarantees that, when a robot decides to perform a lateral move, no other robot is \emph{in the middle} of a lateral move.
 
Also, no matter how many robots lie on the same line through the center of SEC, the innermost will always move first, and then the others will follow in order, after the first has stabilized on SEC/3. When this procedure is completed, there are no more co-radial robots and no robots in the interior of SEC/3. At this point, the robots can safely move toward SEC, radially.

After this phase of the algorithm has been completed, no two robots will ever become co-radial again. We will achieve this through a careful selection of \emph{walkers} and \emph{target points}, and by making walkers move appropriately.

\subsection{\HD Configurations}

One other special initial case has to be resolved: the \HD case. In this configuration, all the robots lie in one half-disk of SEC, and the diameter of such a half-disk is called \emph{principal line} (see Figure~\ref{f04}(a)). The reason why we want to resolve these configurations immediately and separately from all others will be explained in the following, when discussing locked configurations.

\begin{figure}[h!]
\centering
\begin{tabular}{c@{\qquad}c@{\qquad}c}
\includegraphics[scale=0.5]{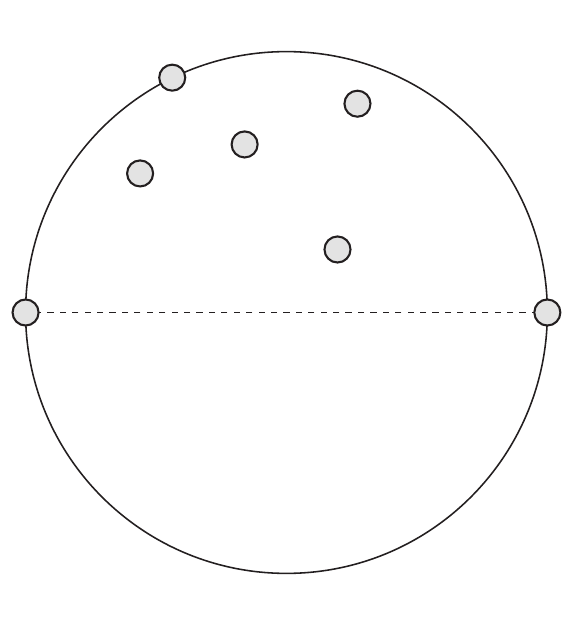} &
\includegraphics[scale=0.5]{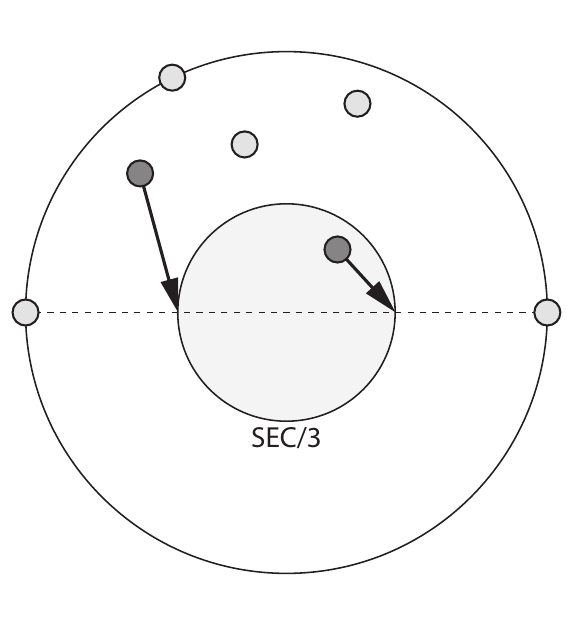} &
\includegraphics[scale=0.5]{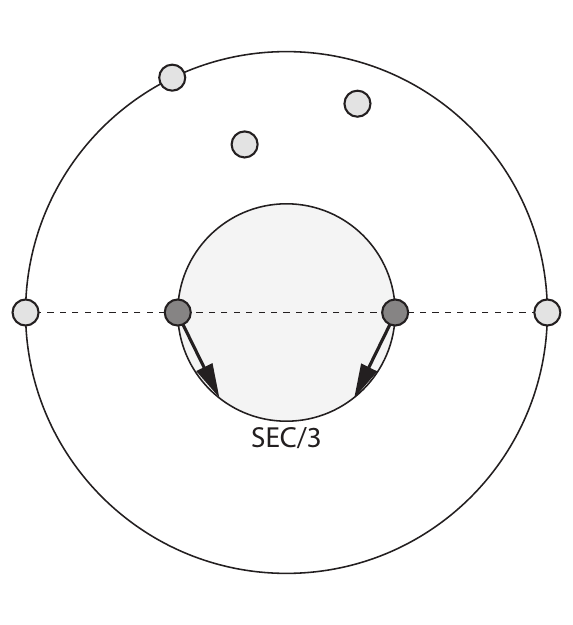} \\
(a) & (b) & (c)
\end{tabular}
\caption{(a) A \HD configuration, with the principal line. (b) Two robots move to the intersection between the principal line and SEC/3. (c) The same two robots move to the non-occupied half-disk.}
\label{f04}
\end{figure}

\HD configurations are resolved by making some robots move from the ``occupied'' half-disk to the ``non-occupied'' one. Note that, while doing so, some robots have to cross the principal line. Also, by definition of SEC, the principal line must contain robots on both endpoints. These two robots, $r_1$ and $r_2$, must stay in place in order to maintain SEC stable. Hence, exactly two other robots, which have smallest angular distances from $r_1$ and $r_2$ respectively, move to the two points in which the principal line intersects SEC/3 (see Figure~\ref{f04}(b)). Once they are both there, they move into the non-occupied half-disk, remaining inside SEC/3, as in Figure~\ref{f04}(c). (More precisely, if the principal line already contains some robots on or inside SEC/3, such robots do not preliminarily move to the perimeter of SEC/3, because it is unnecessary and it may even cause collisions; in this case, they move into the unoccupied half-disk right away.)

A very special \HD case is the one in which all robots lie on the same line. This case is handled like a generic \HD, with two robots first moving on SEC/3 (if they are not already on it or in its interior), and then moving away from the principal line. If they move in opposite directions, the configuration is no longer \HD. If they move in the same direction, they form a generic \HD, which is then resolved normally.

When analyzing the possible evolutions of a \HD configuration, one has to keep in mind that it transitions into a different configuration while one or two robots are still moving. This turns out to be relatively easy, since the moving robots are inside SEC/3 (like the robots that move laterally in the \CO case) and move in a very predictable and controlled way. When the configuration ceases to be \HD, the robots will move on SEC as described before, and they will never form a \HD configuration again.

\subsection{Identifying Targets} 

Suppose now that all robots are on SEC, and the configuration is not \PR and not \HD. In this case we can define a \emph{target set}, which represents the final \RE configuration that the robots are trying to form. Each element of the target set is called a \emph{target}, and corresponds to some robot's intended destination. Hence the target set is a \RE set of $n$ points, arranged on SEC in such a way that it can be computed by all robots, regardless of their local coordinate system (i.e, regardless of the orientation of their local axes, their handedness, and their unit of distance). Next we describe how the target set is defined, depending on the configuration of the robots.

Assume that the configuration has an axis of symmetry $\ell$. Then $\ell$ must also be an axis of symmetry of the target set. If one robot $r$ lies on $\ell$, then the target of $r$ coincides by definition with $r$, and the other targets are defined accordingly (see Figure~\ref{f05}(a)). If no robot lies on $\ell$, then no target lies on $\ell$, either. The correspondences between robots and targets are as in Figure~\ref{f05}(b). Note that the targets are uniquely determined even if the configuration has more than one axis of symmetry, and therefore the same targets are computed by all robots (we will prove this in Proposition~\ref{p:targetcompatible} and Remark~\ref{rem:targets}).

\begin{figure}[h!]
\centering
\begin{tabular}{ccc}
\includegraphics[scale=0.5]{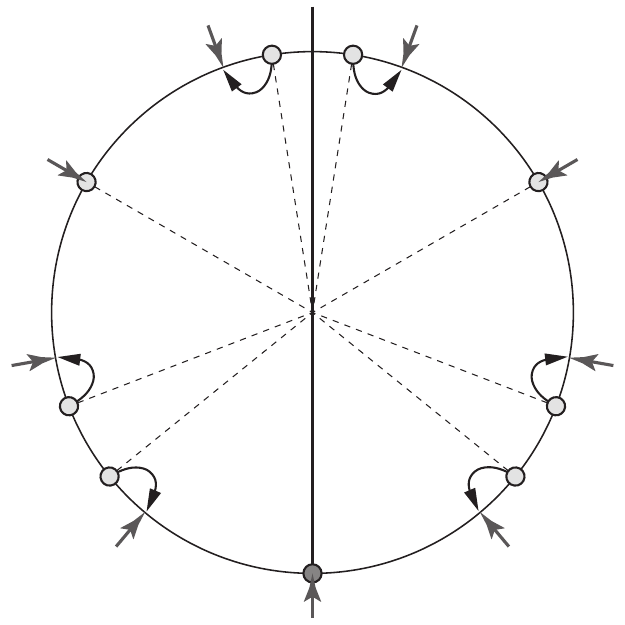} &
\includegraphics[scale=0.5]{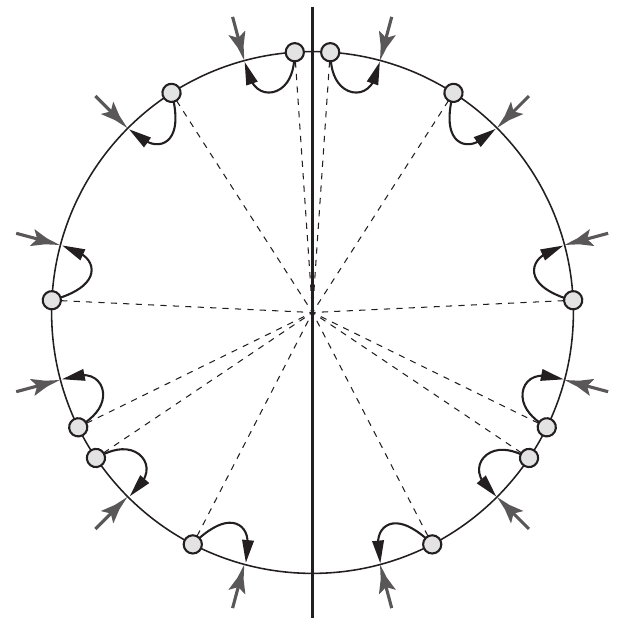} &
\includegraphics[scale=0.5]{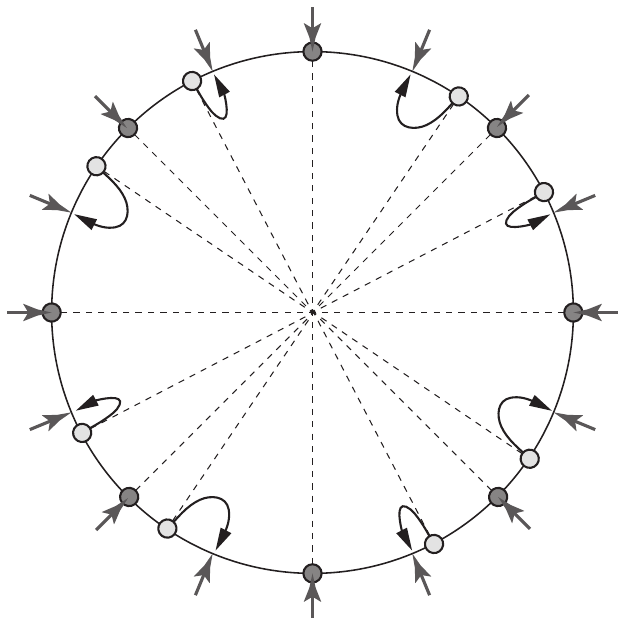} \\
(a) & (b) & (c)
\end{tabular}
\caption{The outer arrows indicate targets, and the inner arrows indicate correspondences between robots and targets. (a) The dark-shaded robot lies on an axis of symmetry. (b) There are some axes of symmetry, none of which contains a robot. (c) There are no axes of symmetry, and the dark-shaded robots form the largest concordance class.}
\label{f05}
\end{figure}

Assume now that the configuration has no axes of symmetry. In this case we say that two robots are \emph{concordant} if their angular distance is of the form $2k\pi/n$, for some integer $k$, and between them there are exactly $k-1$ robots. In other words, two concordant robots have the ``correct'' angular distance, and between them there is the ``correct'' number of robots. This relation partitions the robots into \emph{concordance classes}. The largest concordance class determines the target set: each robot in this class coincides with its own target, by definition.  Even if the largest concordance class is not unique, it turns out that there is always a way to choose one of them unambiguously, in such a way that all robots agree on it. Once some targets have been fixed, the other targets and correspondences are determined accordingly, as Figure~\ref{f05}(c) shows.

\subsection{Identifying Walkers, Locked Configurations} 

When the target set has been identified, then the \emph{walkers} can be defined. The walkers are simply the analogy class of robots that are going to move next.

Typically, the algorithm will attempt to move an analogy class of robots to their corresponding targets. The robots that currently lie on their targets are called \emph{satisfied}, and these robots should not move. Moreover, the walkers should be chosen in such a way that, when they move from their positions into the interior of SEC, they do not cause SEC to change. An analogy class of robots with this property is called \emph{movable}. Finally, no new co-radialities should be formed as the robots move. This means that the walkers should be chosen in such a way that, as they move toward their targets, they do not become co-radial with other robots. The targets of such robots are said to be \emph{reachable}.

Therefore, the walkers are a movable analogy class whose robots are not satisfied and can reach their targets without creating co-radialities. If such a class is not unique, one can always be chosen unambiguously.

There are special cases in which no such an analogy class or robots exists: these configurations are said to be \emph{locked} (see for instance Figure~\ref{f06}(a)). In a locked configuration, the walkers will be an analogy class that is movable and not satisfied, and that is adjacent to some non-movable analogy class. Such an analogy class is called \emph{unlocking}. The goal of these walkers is not to reach their targets (if they could, the configuration would not be locked), but to move in such a way as to ``unlock'' the configuration (as in Figure~\ref{f06}(b)), thus allowing other robots, which were previously non-movable, to reach their targets (as in Figure~\ref{f06}(c)). It can be shown (cf.\ Proposition~\ref{p:locked2}) that, in a locked configuration, the robots that cannot be moved are at most two, and are adjacent on SEC. Also, in a locked configuration, each analogy class consists of at most two robots. Hence there are either one or two walkers in a locked configuration, and they are both adjacent to some non-movable robot.

\begin{figure}[h!]
\centering
\begin{tabular}{ccc}
\includegraphics[scale=0.5]{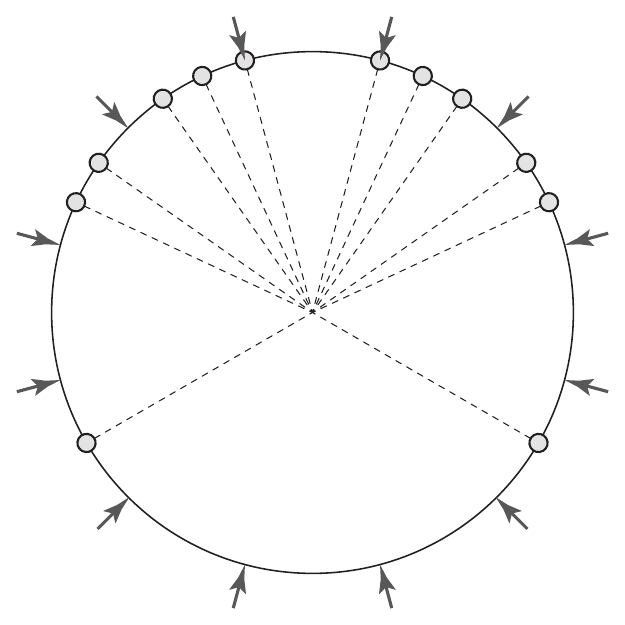} &
\includegraphics[scale=0.5]{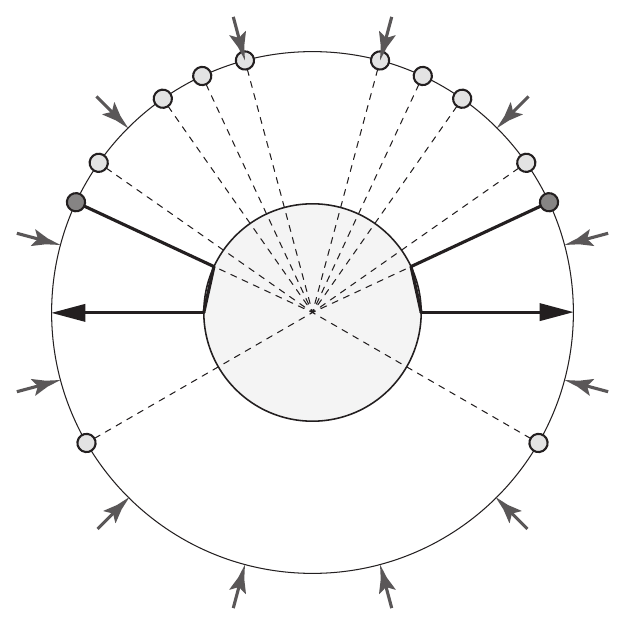} &
\includegraphics[scale=0.5]{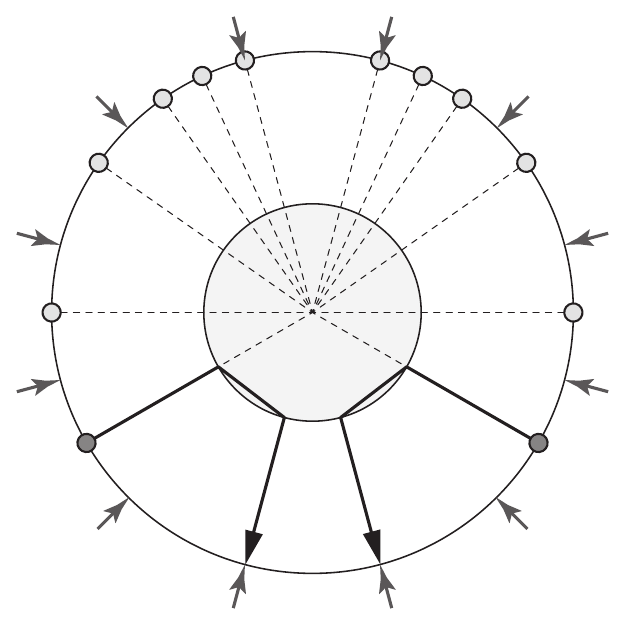} \\
(a) & (b) & (c)
\end{tabular}
\caption{(a) A locked configuration: the topmost robots are satisfied, the bottommost robots are non-movable, and all other robots would create co-radialities in the process of reaching their targets. (b) A preliminary move is made to unlock the configuration. (c) When the configuration is unlocked, the bottommost robots become movable.}
\label{f06}
\end{figure}

\subsection{Identifying \VA Configurations} 

Now we describe the journey that the walkers take to reach their destinations. First they move radially to SEC/3, and they wait for each other there. Once they are all on SEC/3, they start moving laterally, remaining within SEC/3 and its interior, until they reach their \emph{finish set}. Once they are in their finish set, they move back to SEC radially.

The reason why the walkers move to SEC/3 is two-fold. It makes it easier to foresee and prevent the accidental formation of \PR configurations (see Section~\ref{introAcc}), and it clearly separates the robots that should move from the ones that should wait, so that no one gets confused as the configuration changes.

Note that it is easy to recognize a configuration in which the walkers are moving radially to SEC/3 or back to SEC, because analogy classes (and hence the walkers) depend only on angular distances between robots. Hence, if all robots are on SEC, except a few analogous robots that are between SEC and SEC/3, then the configuration is recognized as a ``consistent'', or \VA one, in which the walkers are either moving to SEC/3, or are moving back to SEC (see Figure~\ref{f07}(a)).

\begin{figure}[h!]
\centering
\begin{tabular}{c@{\qquad}c}
\includegraphics[scale=0.5]{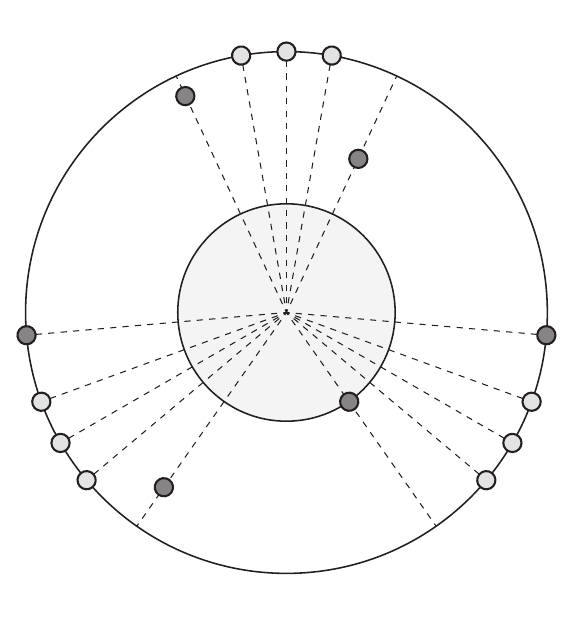} &
\includegraphics[scale=0.5]{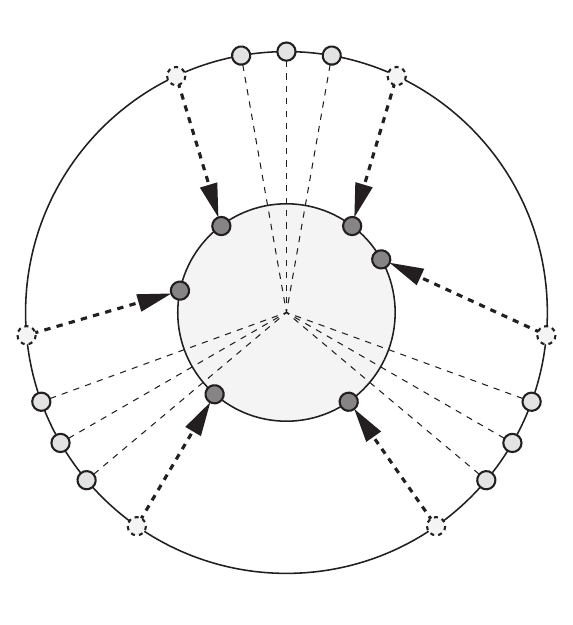} \\
(a) & (b)
\end{tabular}
\caption{Two types of \VA configurations. (a) Some analogous robots lie between SEC and SEC/3, and all other robots are on SEC. (b) All robots are on SEC or on SEC/3, and the distribution of the internal robots is compatible with a possible initial configuration in which they were all on SEC, forming an analogy class.}
\label{f07}
\end{figure}

If the walkers have already started moving laterally in SEC/3, then recognizing the configuration as a \VA one is a little harder. This can be done by ``guessing'' where the internal robots were located when they were still on SEC and they have been selected as walkers. If there is a way to re-position the internal robots within their respective ``sectors'' of SEC in such a way as to make them become a full analogy class, then the configuration is considered \VA, and the internal robots are considered walkers (see Figure~\ref{f07}(b)). Otherwise, it means that the execution is in one of the earlier stages, and the robots still have to make their preliminary move to SEC.

\subsection{Identifying the Finish Set}

Once the configuration has been recognized as \VA and all walkers are on SEC/3, they compute their \emph{finish set}. This is simply the set of their destinations on SEC/3, which they want to reach before moving back to SEC.

In order to understand where they should be going, the walkers have to recompute their targets. Indeed, note that the original targets have been computed when the walkers were on SEC. As they are now on SEC/3 and they will soon be moving laterally inside SEC/3, we need a robust way to define targets. By ``robust'' we mean that different walkers should compute the same target set, and that the target set should not change as the walkers move within SEC/3. Of course it may not be possible to reconstruct the original walkers' positions on SEC and recompute the original targets, and therefore once again the walkers have to ``take a guess''. The guess is that, when they were still on SEC, each walker was equidistant from its two adjacent robots, as in Figure~\ref{f08}(a). This position of the walkers is referred to as the \emph{principal relocation}, and it can be computed unambiguously by all robots.

Now the robots compute the finish set as follows. First of all, if the principal relocation is not a full analogy class, but just a subset of one, then the walkers know that it could not possibly be their initial position on SEC (see Figure~\ref{f08}(b)). In this case, the finish set is the principal relocation itself. The reason is that, by moving to their principal relocation, the walkers all join some bigger analogy class. This is a good thing to do, because it makes progress toward having a unique analogy class.

\begin{figure}[h!]
\centering
\begin{tabular}{c@{\quad}c@{\quad}c}
\includegraphics[scale=0.5]{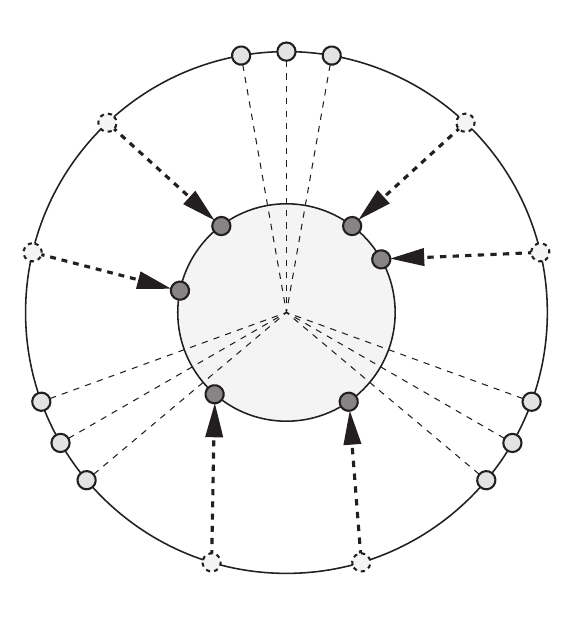} &
\includegraphics[scale=0.5]{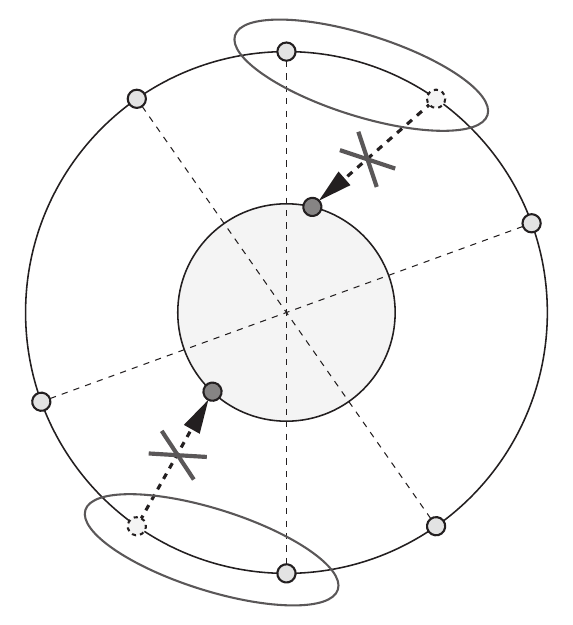} &
\includegraphics[scale=0.5]{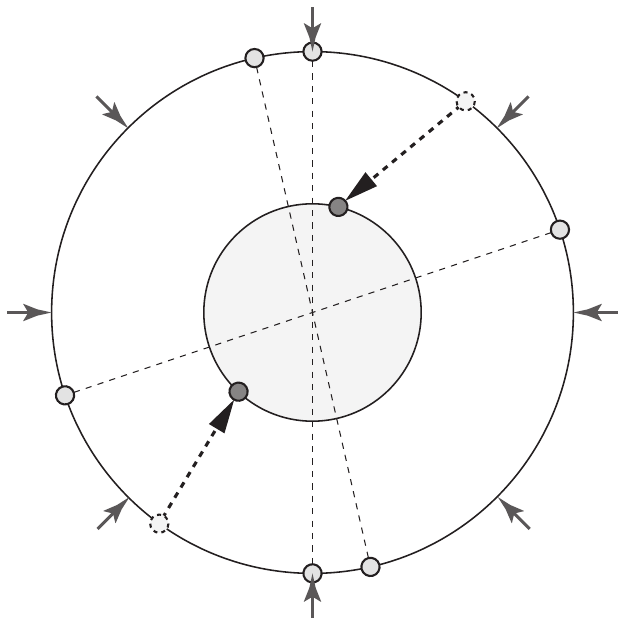} \\
(a) & (b) & (c)
\end{tabular}
\caption{(a) The principal relocation of the internal robots. (b) If the principal relocation is a proper subset of an analogy class, it cannot be the original position of the internal robots, or else a larger set of walkers would have been selected. (c) If the principal relocation forms an analogy class, it is used to determine the target set. Such targets remain fixed as the internal robots move within their respective sectors.}
\label{f08}
\end{figure}

If the principal relocation forms in fact an analogy class, then the walkers assume that to be their original position on SEC. Hence they compute the new targets based on that configuration, with the usual algorithm (see Figure~\ref{f08}(c)). Now, if the walkers can reach their respective targets from inside SEC/3 (that is, without becoming co-radial with other robots), then the finish set is the set of their targets. Otherwise, the walkers are confused, and by default their finish set is the principal relocation.

Now that the finish set has been defined, the robots move there, always remaining within SEC/3, and without becoming co-radial with each other. There is only one exception: suppose that the walkers reach their finish set and move radially to SEC: let $R$ be the set of the final positions of the walkers on SEC. If the new configuration is locked, and the robots in $R$ happen to form an unlocking analogy class, then it was not a good idea for the walkers to go to $R$. Indeed, this would cause them to become walkers again (unless there are two unlocking analogy classes and the other one is chosen), and the execution would enter an infinite loop. In this special case, the walkers have to do something to unlock the configuration, instead of reaching $R$. The strategy is simple: if the walkers are two, they move to two antipodal points (as in Figure~\ref{f06}(b)); if there is a unique walker, it becomes antipodal with some non-movable robot currently located on SEC. In the resulting configuration, all analogy classes will be movable, and the configuration will not be locked (cf.\ Proposition~\ref{p:locked1}). Note that this type of move would not be possible in a \HD configuration: this is precisely why we made sure to resolve \HD configurations early on.

\subsection{Accidental Formation of \PR Configurations}\label{introAcc}

Our algorithm has still one big unresolved issue. Recall that, every time a robot computes a new destination, it first checks if the configuration is \PR. If it is, it executes a special protocol; otherwise it proceeds normally. So, what happens if the swarm is executing the non-\PR protocol, and suddenly a \PR configuration is formed ``by accident''? If a robot happens to perform a Look-Compute phase right at that time, it is going to execute the \PR protocol, while all the other robots are still executing the other one, and maybe they are in the middle of a move (see Figure~\ref{f09}(a)). This leads to an inconsistent behavior that will potentially disrupt the ``flow'' of the entire algorithm.

To resolve this issue, we have to avoid the unintended formation of \PR configurations whenever possible. If in some cases it is not easily avoidable, then we have to make sure that the whole swarm stops moving (or \emph{freezes}, in the terminology of Section~\ref{sec:model}) whenever a \PR configuration is formed. This way, all robots will transition into the new configuration, and all of them will coherently execute the \PR protocol in the next cycle.

In Section~\ref{sec:analysis} we thoroughly discuss this topic, and we show how the robots should behave in every case. Fortunately, certain important configurations are safe: no \CE or \CO or \HD configuration can be \PR. So, in these initial phases, no \PR configuration can be formed accidentally. Also, in a \PR configuration no robot can be in SEC/3: this explains why we make our walkers move radially to SEC/3 first, and we allow them to move laterally only within SEC/3.

Hence, the only ``dangerous'' moves are the radial ones, which are performed by the walkers, or by the robots that are reaching SEC during the preliminary step. We can conveniently simplify the problem if we move only one analogy class of robots at a time. Note that this is already the case when the moving robots are the walkers, and in the other cases there is always a way to totally order the analogy classes unambiguously. If only one analogy class is moving radially (either from SEC to SEC/3 or from SEC/3 to SEC), it is easier to understand what is going to happen, and to keep everything under control.

\begin{figure}[h!]
\centering
\begin{tabular}{ccc}
\includegraphics[scale=0.5]{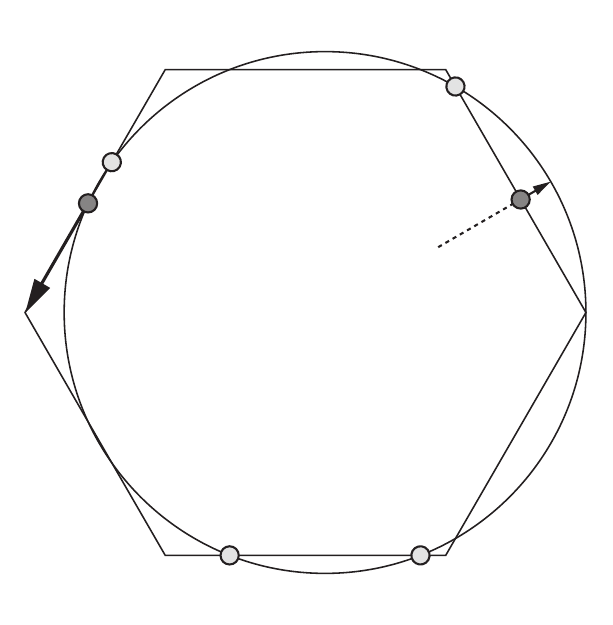} &
\includegraphics[scale=0.5]{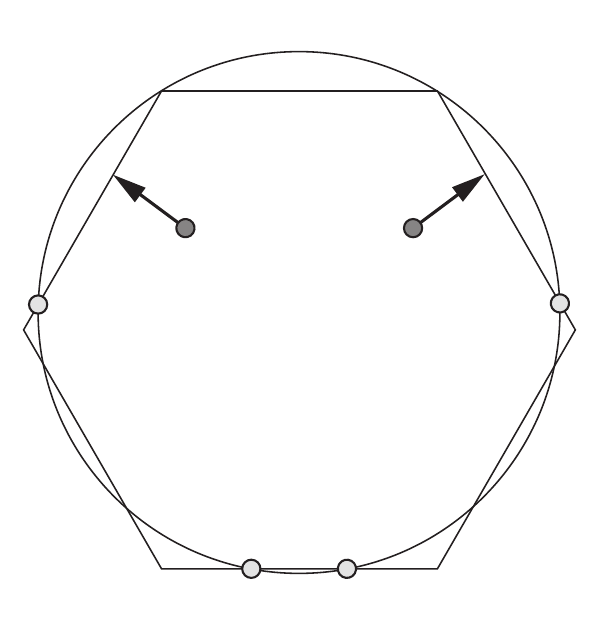} &
\includegraphics[scale=0.5]{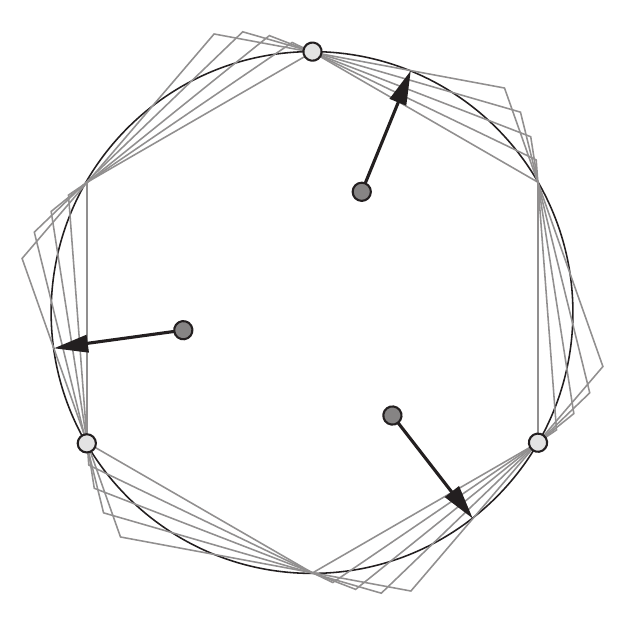} \\
(a) & (b) & (c)
\end{tabular}
\caption{(a) As the robot on the right moves to SEC, a \PR configuration is accidentally formed. The robot on the left recognizes a \PR configuration, and starts executing the corresponding protocol, which is inconsistent with the other robot's move. (b) To prevent this behavior, enough critical points are added. Now the swarm is guaranteed to stop as soon as a \PR configuration is formed. (c) A case in which infinitely many \PR configurations are formable. Still, only the innermost is relevant, because it can be reached before all the others.}
\label{f09}
\end{figure}

The general protocol that we use for radial moves is called \emph{cautious move}. In a cautious move, the robots compute a set of \emph{critical points}, and move in such a way as to freeze whenever they are all located at a critical point (see for instance Figure~\ref{f09}(b)). Intuitively, the robots ``wait for each other'': only the robots that are farthest from their destinations are allowed to move, while the others wait. Then, the robots make only moves that are short enough, and in addition they stop at every critical point that they find on their paths.\footnote{Roughly the same mechanism has been used in~\cite{CieFPS12}, with some technical differences.} Now, if we use the potentially formable \PR configurations to generate the critical points, we can indeed guarantee that the robots will freeze as soon as they form one. This is still not enough, because the formable \PR configurations may be infinitely many (as in Figure~\ref{f09}(c)), while the critical points must be finite, or the cautious move would never end. However, it can be shown that, in all cases, either there is a finite number of \PR configurations that will be formed before all the others, or suitable critical points can be chosen in such a way as to prevent the formation of \PR configurations altogether. Hence, it turns out that it is always possible to choose a finite set of critical points for all cautious moves, and guarantee that the swam is frozen whenever it transitions into a \PR configuration.

\subsection{Proof of Correctness: Outline} 

The proof of correctness of this algorithm is necessarily long and complex. This is partly because the algorithm itself is complicated and full of subtle details, and partly because the analysis must take into account a large number of different possible configurations and behaviors, and show that all of them are resolved correctly.

The correctness of the \PR case of the algorithm, as well as the \CE, \CO, and \HD cases is relatively straightforward, and is proven in the first lemmas of Section~\ref{sec:corr}. The difficulty here is to prove that the execution flows seamlessly from \HD to \CO, etc.

The other parts of the algorithm need a much more careful analysis. The correctness of the cautious move protocol is proven in Section~\ref{sec:cautious}. The discussion on the accidental formation of \PR configurations and on how to choose the critical points of the cautious moves is in Section~\ref{sec:analysis}. Much different strategies and ideas have to be used, depending on several properties of the configurations. In Proposition~\ref{p:locked2} we give a complete characterization of the locked configurations, showing where the non-movable and the unlocking analogy classes are.

With all these tools, we can finally tackle the \VA case, and so analyze the main ``loop'' of the algorithm. In the middle part of Section~\ref{sec:corr} we show that the different phases of the execution ``hinge together'' as intended: all the walkers reach SEC/3 and freeze there (unless a \PR configuration is formed in the process), then they all move to their finish set, freeze again, and finally they move back to SEC. As the execution continues and more iterations of this phase are made, we have to study how exactly the target set changes, and we have to make sure that a \PR configuration is eventually formed.

To this end we prove that, at each iteration, some ``progress'' is made toward a \RE or \BI configuration. The progress may be that the walkers join another analogy class (thus reducing the total number of analogy classes), or that a new axis of symmetry is acquired, or that more robots become satisfied. A precise statement and a complete proof is given in Lemma~\ref{z:next1}. Of course the configuration may also be locked, and this case is analyzed separately, in Lemma~\ref{p:unlocked}: here we prove that, after one iteration, either the configuration is no longer locked, or some analogy classes have merged, or a previously non-movable analogy class has become movable.

Also, by design, the algorithm never allows an analogy class to split (because the walkers constitute an analogy class when they are selected, and are again all analogous when they reach their finish set), and it never causes a symmetric configuration to become asymmetric from one iteration to the next. However, it is true that the targets may change, and thus the number of satisfied robots may actually decrease. But this can happen only when some analogy classes merge, or when the configuration becomes symmetric. And we know that this can happen only finitely many times.

So, either a \PR configuration is formed by accident (and we know that this case leads to a quick resolution), or eventually there will be only one analogy class left, and hence the configuration will be \RE or \BI. This will conclude the proof.

\subsection{Smaller Swarms} 

The algorithm we just outlined works if the robots in the swarm are $n>5$. If $n=3$, we have an ad-hoc algorithm described in Lemma~\ref{z:n3}. If $n=5$, the general algorithm needs some modifications, because it is no longer true that, in a locked configuration, there is a non-satisfied unlocking analogy class. The details of the extended algorithm are given in Lemma~\ref{z:n5}.\footnote{The results in~\cite{Kat05} seem to imply that the \CF problem can be solved for any odd number of robots in \ASYNC. A proof for the \SSYNC model is given, but its generalization to \ASYNC is missing some crucial parts. No extended version of the paper has been published, either. Hence, for completeness, we provide our own solutions for the special cases $n=3$ and $n=5$.} Finally, the case $n=4$ has recently been solved in~\cite{MamV16}.

\section{The Algorithm: Formal Description}\label{sec:formal}

\subsection{Geometric Definitions and Basic Properties}\label{defi:geo}

\paragraph{Smallest enclosing disks and circles.}
Given a finite set $S\subset \mathbb R^2$ of $n\geqslant 2$ points, we define the \emph{smallest enclosing disk} of $S$, or \emph{SED}$(S)$, to be the (closed) disk of smallest radius such that every point of $S$ lies in the disk. For any $S$, SED$(S)$ is easily proven to exist, to be unique, and to be computable by algebraic functions. The \emph{smallest enclosing circle} of $S$, or \emph{SEC}$(S)$, is the boundary of SED$(S)$.

Another disk will play a special role: \emph{SED/3}$(S)$. This is concentric with SED$(S)$, and its radius is $1/3$ of the radius of SED$(S)$. The boundary of SED/3$(S)$ is denoted as \emph{SEC/3}$(S)$.

If $S$ is understood, we may omit it and simply refer to SED, SEC, SED/3, and SEC/3.

\paragraph{Centrality and co-radiality.}
If one point of $S$ lies at the center of SED, then $S$ is said to form a \CE configuration. If two points lie on the same ray emanating from the center of SED, they are said to be \emph{co-radial} with each other, and each of them is a \emph{co-radial} point. If $S$ has co-radial points, it is said to form a \CO configuration. It follows that a \CE set is also \CO.

\paragraph{Antipodal points.}
Two points on SEC$(S)$ that are collinear with the center of SEC$(S)$ are said to be \emph{antipodal} to each other (with respect to SEC$(S)$).

\begin{observation}\label{l:sechull}
The center of SED$(S)$ lies in the convex hull of $S\cap \mbox{SEC}(S)$. Therefore, every half-circle of SEC$(S)$ contains at least one point of $S$. In particular, if just two points of $S$ lie on SEC$(S)$, they are antipodal.
\end{observation}

\paragraph{\PR configurations.}
$S$ is \PR if there exists a regular $n$-gon (called the \emph{supporting polygon}) such that, 
for each pair of adjacent edges, one edge contains exactly two points of $S$ (possibly on its endpoints), and the other edge's relative interior contains no point of $S$~\cite{DieLP08}.  A \PR set is shown in Figure~\ref{f01}(c). There is a natural correspondence between points of $S$ and vertices of the supporting polygon: the {\em matching vertex} $v$ of point $p\in S$ is such that $v$ belongs to the edge containing $p$, and the segment $vp$ contains no other point of $S$. If two points of $S$ lie on a same edge of the supporting polygon, then they are said to be \emph{companions}.

\paragraph{\RE configurations.}
$S$ is \RE if its points are the vertices of a regular $n$-gon. The \CF problem requires $n$ robots to reach a \RE configuration and never move from there.

\paragraph{\HD configurations.}
Suppose that there exists a line $\ell$ through the center of SED, called the \emph{principal line}, such that exactly one of the two open half-planes bounded by $\ell$ contains no points of $S$. Then, such an open half-plane is called \emph{empty half-plane}, and $S$ is said to be a \HD set. A \HD set is shown in Figure~\ref{f04}(a). The center of SED divides $\ell$ into two rays, called \emph{principal rays}. Note that there must be two points of $S$ lying at the intersections between $\ell$ and SEC.

\paragraph{Angular distance and sectors.}
Let $c$ be the center of SED$(S)$. The \emph{angular distance} between two points $a$ and $b$ (distinct from $c$) is the measure of the smallest angle between $\angle a c b$ and $\angle b c a$, and is denoted by $\theta(a,b)$. The \emph{sector} defined by two distinct points $a$ and $b$ is the locus of points $x$ such that $\theta(a,x)+\theta(x,b)=\theta(a,b)$. (In the exceptional case in which $c$ lies on the segment $ab$, the points $a$ and $b$ define two sectors, which are the two half-planes bounded by the line through $a$ and $b$.)

\paragraph{Angle sequences.}
For the rest of this section we assume $S\subset \mathbb R^2$ to be a finite set of $n>2$ points that is not \CO.

Note that the positions of the points of $S$ around the center of SED, taken clockwise, naturally induce a cyclic order on $S$. Let $p\in S$ be any point, and let $p_i\in S$ be the $(i+1)$-th point in the cyclic order, starting from $p=p_0$. Let $\alpha_i^{(p)}=\theta(p_i,p_{i+1})$, where the indices are taken modulo $n$. Then, $\alpha^{(p)}=(\alpha_i^{(p)})_{0\leqslant i< n}$ is called the \emph{clockwise angle sequence} induced by $p$. Of course, depending on the choice of $p\in S$, there may be at most $n$ different clockwise angle sequences.

Letting $\beta_i^{(p)}=\alpha_{n-i}^{(p)}$, for $0\leqslant i< n$, we call $\beta^{(p)}=(\beta_i^{(p)})_{0\leqslant i< n}$ the \emph{counterclockwise angle sequence} induced by $p\in S$. We let $\alpha$ and $\beta$ be, respectively, the lexicographically smallest clockwise angle sequence and the lexicographically smallest counterclockwise angle sequence of $S$.

Finally, we denote by $\mu^{(p)}$ the lexicographically smallest between $\alpha^{(p)}$ and $\beta^{(p)}$, and by $\mu$ the lexicographically smallest between $\alpha$ and $\beta$. We call $\mu^{(p)}$ the \emph{angle sequence} induced by point $p$. (Since $\mu$ is a sequence, we denote its $i$-th element by $\mu_i$, and the same goes for $\mu^{(p)}$.)

\paragraph{Periods.}
The number of distinct clockwise angle sequences of $S$ is called the \emph{period} of $S$. It is easy to verify that the period is always a divisor of $n$. $S$ is said to be \EQ if its period is $1$, \BI if its period is $2$, \PER if its period is greater than $2$ and smaller than $n$, and \AP if its period is $n$. In a \BI set, any two points at angular distance $\mu_0$ are called \emph{neighbors}, and any two points at angular distance $\mu_1$ are called \emph{quasi-neighbors}. A \PER set is \UP if $\alpha\neq\beta$, and \BP if $\alpha=\beta$. Similarly, an \AP set is \UA if $\alpha\neq\beta$, and \BA if $\alpha=\beta$.

\begin{figure}[h!]
\centering
\begin{tabular}{c@{\qquad}c}
\includegraphics[scale=0.5]{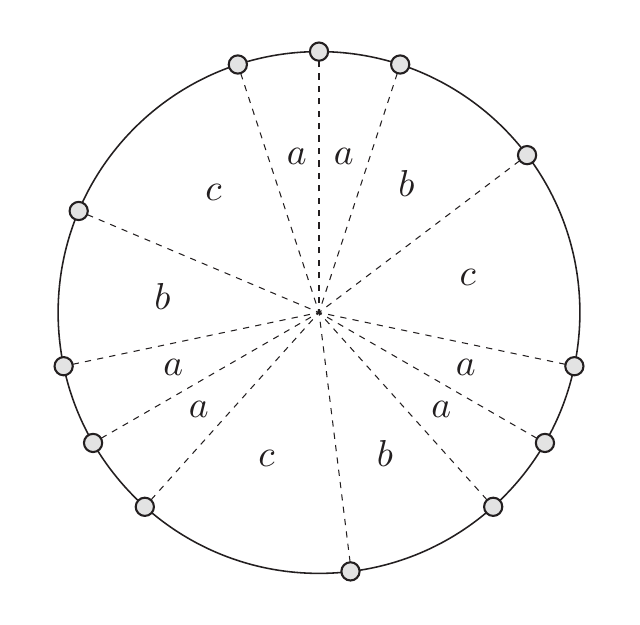} &
\includegraphics[scale=0.5]{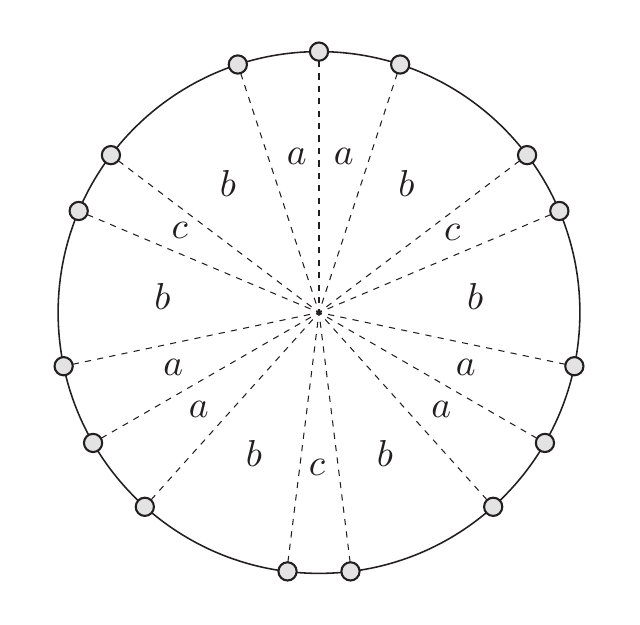} \\
(a) & (b)\\
\\
\includegraphics[scale=0.5]{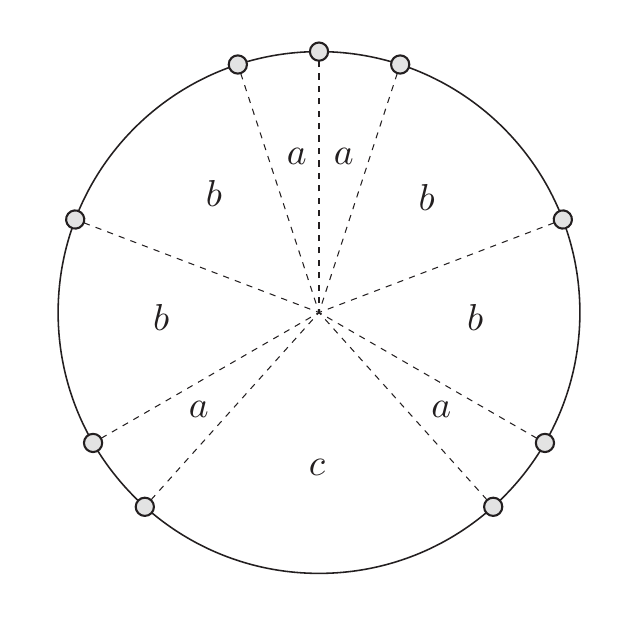} &
\includegraphics[scale=0.5]{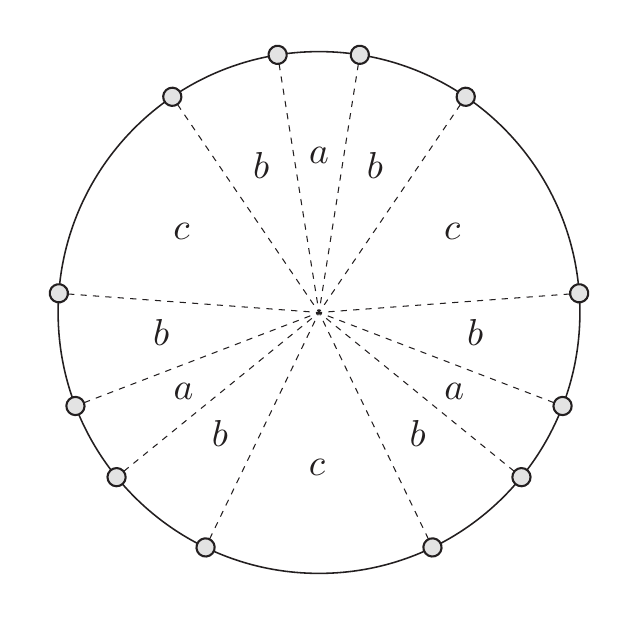} \\
(c) & (d)
\end{tabular}
\caption{(a) A \UP set. (b) A \BP set. (c) A \BA set. (d) A \DBI set.}
\label{g02}
\end{figure}

\paragraph{Analogy and strong analogy.}
We say that $p\in S$ is \emph{analogous} to $q\in S$ if $\mu^{(p)}=\mu^{(q)}$. In particular, if $\alpha^{(p)}=\alpha^{(q)}$, $p$ and $q$ are said to be \emph{strongly analogous}. Analogy and strong analogy are equivalence relations on $S$, and the equivalence classes that they induce on $S$ are called \emph{analogy classes} and \emph{strong analogy classes}, respectively.

\begin{observation}\label{o:analogy2}
Let $S$ be a set whose points all lie on SEC$(S)$.
\begin{itemize}
\item If $S$ is \EQ, all points are strongly analogous.
\item If $S$ is \BI, all points are analogous, and there are exactly two strong analogy classes.
\item If $S$ is \UP with period $k\geqslant 3$, each analogy class is an \EQ subset of size $n/k$.
\item If $S$ is \BP with period $k\geqslant 3$, each analogy class is either a \BI set of size $2n/k$, or an \EQ set of size $n/k$ or $2n/k$.
\item If $S$ is \UA, each analogy class consists of exactly one point.
\item If $S$ is \BA, each analogy class consists of either one or two points.
\end{itemize}
\end{observation}

\begin{observation}\label{o:analogy1}
The following statements are equivalent.
\begin{itemize}
\item $S$ has a unique analogy class.
\item $S$ has period $1$ or $2$.
\item $S$ is \EQ or \BI.
\end{itemize}
\end{observation}

\begin{proposition}\label{p:symm}
Let $S$ be a set of at least two points, and let $C$ be an analogy class of $S$. If $\ell$ is an axis of symmetry of $S$, then $\ell$ is an axis of symmetry of $C$. Also, if $S$ has a $k$-fold rotational symmetry around the center of SED$(S)$, then $C$ has a $k$-fold rotational symmetry with the same center.
\end{proposition}
\begin{proof}
Suppose that $\ell$ is an axis of symmetry of $S$. Let $p\in C$, and let $p'$ be the symmetric of $p$ with respect to $\ell$. Since $p\in S$ and $\ell$ is an axis of symmetry of $S$, it follows that $p'\in S$. Also, the clockwise angle sequence induced by $p$ (respectively, $p'$) is the same as the counterclockwise angle sequence induced by $p'$ (respectively, $p$). Hence $\mu^{(p)}=\mu^{(p')}$, which means that $p$ and $p'$ are analogous, and therefore $p'\in C$.

Suppose that $S$ has a $k$-fold rotational symmetry with respect to the center of SED$(S)$. Let $p\in C$, and let $p'$ be any point such that $\theta(p,p')=2\pi/k$, and $p$ and $p'$ are equidistant from the center of SED$(S)$. Since $p\in S$, it follows that $p'\in S$. Also, the clockwise (respectively, counterclockwise) angle sequence induced by $p$ is the same as the clockwise (respectively, counterclockwise) angle sequence induced by $p'$. Hence $\mu^{(p)}=\mu^{(p')}$, which means that $p$ and $p'$ are analogous, and therefore $p'\in C$.
\end{proof}

\paragraph{\DBI configurations.}
$S$ is said to be \DBI if it is \BP with period $4$ and has exactly two analogy classes.

\paragraph{Concordance.}
Two points $p,q\in S$ are \emph{concordant} if there exists an integer $k$ such that the angular distance between $p$ and $q$ is $2k\pi/n$, and there are exactly $k+1$ points of $S$ in the sector defined by $p$ and $q$ (including $p$ and $q$ themselves). Concordance is an equivalence relation on $S$, and its equivalence classes are called \emph{concordance classes}.

\begin{observation}\label{o:concordance2}
In a \UP or \UA set, any two analogous points are also concordant. Hence, in such a set, each analogy class is a subset of some concordance class.
\end{observation}

\begin{proposition}\label{p:axiscenter}
Let $S$ be a set of at least two points, all of which are on SEC$(S)$. Then, each axis of symmetry of $S$ passes through the center of SED$(S)$.
\end{proposition}
\begin{proof}
If $S$ consists of exactly two points, then such two points must be antipodal, by Observation~\ref{l:sechull}. In this case, $S$ has exactly two axes of symmetry, both of which pass through the center of SEC. Suppose now that $S$ consists of at least three points, and it has an axis of symmetry $\ell$. In this case, there must be a point $p\in S$ that does not lie on $\ell$, whose symmetric point $p'\in S$ does not lie on $\ell$, either. Both $p$ an $p'$ lie on SEC by assumption, and the axis of the (non-degenerate) chord $pp'$ must be $\ell$. But the axis of a circle's chord passes through the center of the circle, and therefore $\ell$ passes through the center of SED.
\end{proof}

\paragraph{Footprints and anti-footprints.}
We define the \emph{footprint} (respectively, \emph{anti-footprint}) of $p\in S$ as the point on SEC$(S)$ (respectively, SEC/3$(S)$) that is co-radial with $p$, and we denote it by $\mathcal F(p)$ (respectively, $\mathcal F'(p)$). We also define the footprint (respectively, anti-footprint) of a subset $A\subseteq S$, denoted by $\mathcal F(A)$ (respectively, $\mathcal F'(A)$), as the set of the footprints (respectively, anti-footprints) of all the points of $A$.

\paragraph{External and internal points.}
We let $\mathcal E(S)=S\cap\mbox{SEC}(S)$ be the set of \emph{external} points of $S$. Similarly, we let $\mathcal I(S)=S\setminus \mathcal E(S)$ be the set of \emph{internal} points of $S$.

\paragraph{Main sectors, occupied sectors, and consecutive points.}
Each sector defined by pairs of distinct points of $S$ whose interior does not contain any point of $S$ is called \emph{main sector} of $S$. It follows that $S$ has exactly $|S|$ main sectors (recall that we are assuming $S$ not to be \CO). A main sector of $\mathcal E(S)$ is an \emph{occupied} sector of $S$ if it contains some points of $\mathcal I(S)$. If two points of $S$ define a main sector, they are said to be \emph{consecutive} points of $S$.

\paragraph{Midpoints.}
We say that $p\in S$ is a \emph{midpoint} in $S$ if $\alpha_0^{(p)}=\beta_0^{(p)}$.

\paragraph{Relocations and well-occupied configurations.}
If $\mathcal I(S)$ is not empty, a \emph{relocation} of $\mathcal I(S)$ (with respect to $S$) is the image of an injective function $f\colon \mathcal I(S)\to\mbox{SEC}(S)$ that maps every internal point of $S$ to some point in the interior of the same occupied sector of $S$. The \emph{principal relocation} is the (unique) relocation $R\subset\mbox{SEC}(S)$ every point of which is a midpoint in $\mathcal E(S)\cup R$. If there exists a relocation $R$ of $\mathcal I(S)$ that is an analogy class of $\mathcal E(S)\cup R$, then $S$ is said to be \emph{well occupied}.

\paragraph{\VA configurations (\RD or \WA).}
$S$ is a \VA set if it consists of at least five points, it is not \CO, not \HD, and one of the following conditions holds.
\begin{itemize}
\item All the points of $S$ are either on SEC or in SED/3, and $S$ is well occupied (as in Figure~\ref{f07}(b)). In this case, $S$ is said to be \RD.
\item No point of $S$ is in the interior of SED/3, and all the internal points of $S$ are analogous (as in Figure~\ref{f07}(a)). In this case, $S$ is said to be \WA.
\end{itemize}

\begin{remark}
If $S$ has no internal points, it is \VA and \WA. Also, if $S$ is \EQ or \BI and none of its points lies in the interior of SED/3, it is \VA and \WA.
\end{remark}

\begin{remark}
There exist \VA sets that are both \RD and \WA. For instance, if the internal points of a \VA set constitute an analogy class and they all lie on SEC/3, then the set is both \RD and \WA.
\end{remark}

\begin{proposition}\label{p:readyinternal}
In a \VA and \RD set, the occupied sectors either contain exactly one point each, or they contain exactly two points each.
\end{proposition}
\begin{proof}
Let $S$ be a \VA and \RD set. Then $\mathcal I(S)$ has a relocation $R$ that is an analogy class of $S'=\mathcal E(S)\cup R$. If $S'$ has period $1$ or $2$, by Observation~\ref{o:analogy1} it has a unique analogy class, and therefore $R=S'$, meaning that all points of $S$ are internal, which is impossible. Hence $S'$ has period at least $2$, and is therefore \PER or \AP.

Recall that a relocation remaps the internal points within the same occupied sector. If $S'$ is \UP or \UA, then no two analogous points are consecutive in $S'$, and hence each occupied sector of $S$ contains exactly one point. If $S'$ is \BP or \BA, then there can be no three consecutive analogous points in $S'$ (i.e., there cannot be three analogous points $a,b,c\in S$ such that $b$ is consecutive to both $a$ and $c$). Hence, either all occupied sectors of $S$ contain exactly one point, or all contain exactly two points.
\end{proof}

\paragraph{\IN configurations.}
If $S$ consists of at least five points, it is not \CO, not \HD, and not \VA, it is said to be \IN.

\paragraph{Movable analogy classes.}
An analogy class $C$ of a \VA and \WA set $S$ is \emph{movable} if $C\neq S$ and $\mbox{SED}(S)=\mbox{SED}(S\setminus C)$. For instance, in Figure~\ref{f06}(a), every analogy class is movable, except the bottom one.

\begin{observation}\label{o:unmovable}
A set $C\subseteq S$ is a non-movable analogy class of a \VA and \WA set $S$ if and only if there exists a line through the center of SED$(S)$ bounding a (closed) half-plane containing no points of $(S\cap \mbox{SEC}(S))\setminus C$.
\end{observation}

\begin{proposition}\label{p:unmovable}
Let $S$ be a \VA and \WA set. If $S$ has a non-movable analogy class, then $S$ is not \PER.
\end{proposition}
\begin{proof}
Without loss of generality, we assume that all points of $S$ lie on SEC. If this is not the case, we may equivalently consider $\mathcal F(S)$ instead of $S$.

Suppose for a contradiction that $S$ is \PER with period $3\leqslant k\leqslant n/2$, and some analogy class $C\subseteq S$ is not movable. Due to Observation~\ref{o:analogy1}, hence $S$ has another analogy class $C'\subseteq S\setminus C$. By Observation~\ref{o:analogy2}, $C'$ is either an \EQ or a \BI set of size either $n/k$ or $2n/k$, hence $|C'|\geqslant 2$. Also, $C'$ is rotationally symmetric with respect to the center of SED$(S)$. Since all points of $C'$ lie on SEC$(S)$, by Observation~\ref{o:unmovable} there exists a closed half-plane bounded by a line through the center of SED$(S)$ that contains no points of $C'$. But this is impossible, due to the rotational symmetry of $C'$.
\end{proof}

\begin{proposition}\label{p:targetcompatible}
Let $S$ be a \VA set whose points all lie on SEC, and suppose that $S$ has at least one axis of symmetry. If $p,q\in S$ are two points that lie on an axis of symmetry of $S$ (not necessarily on the same axis), then $p$ and $q$ are concordant. If no points of $S$ lie on any axis of symmetry of $S$, then the union of the axes of symmetry partitions the plane into sectors, all of which contain the same number of points of $S$.
\end{proposition}
\begin{proof}
By Proposition~\ref{p:axiscenter}, all axes of symmetry of $S$ pass through the center of SED.

Suppose first that the set $Y$ of the points of $S$ that lie on an axis of symmetry of $S$ is not empty. If $|Y|=1$ there is nothing to prove, so let us assume that $|Y|\geqslant 2$. Let $p,q\in Y$ be two points at minimum angular distance (with respect to the center of SED$(S)$), and let $\gamma$ be their angular distance. If $p$ and $q$ lie on the same axis of symmetry, then $\gamma=\pi$. In this case, $p$ and $q$ define two sectors, each containing exactly $n/2+1$ points, implying that $p$ and $q$ are concordant. Assume now that $p$ and $q$ do not lie on the same axis of symmetry, and that therefore $\gamma<\pi$. Since $q$ lies on an axis of symmetry of $S$, there is a point $p'\in S\setminus \{p\}$, lying on an axis of symmetry of $S$, at angular distance $\gamma$ from $q$. Proceeding in this fashion, we construct a sequence of points around SEC, each of which has angular distance $\gamma$ from the next, and each of which lies on an axis of symmetry of $S$. The set of points in this sequence has to coincide with $Y$, or else it would contain a point at distance smaller than $\gamma$ from $p$, contradicting the minimality of $\gamma$. It follows that $2\pi/\gamma$ is an integer $k$, and the (closed) sector defined by two consecutive points in the sequence contains exactly $n/k+1$ points. This implies that all the points that are consecutive in $Y$ are concordant. But concordance is an equivalence relation, and therefore all points of $Y$ are concordant.

Suppose now that no points of $S$ lie on any axis of symmetry, and let $\ell$ and $\ell'$ be two axes of symmetry at minimum angular distance (i.e., whose intersections with SEC$(S)$ include two points whose angular distance is minimum among all pairs of axes of $S$). Let such a minimum angular distance be $\gamma$. Reasoning as above, we construct a sequence of axes of symmetry of $S$, each at angular distance $\gamma$ from the next. Again, $2\pi/\gamma$ must be an integer $k$, or else $\gamma$ would not be minimum. The union of the axes in this sequence partitions the plane into $k$ sectors, each of which contains exactly $n/k$ points of $S$ (because each sector is a symmetric copy of the next).
\end{proof}

\paragraph{Target sets and point-target correspondence.}
If $S$ is a \VA set, we can define a \emph{target set} on $S$, which consists of a \RE set of $n$ points lying on SEC$(S)$ (refer to Figure~\ref{f05}). Each of the $n$ points of the target set is a \emph{target}. Furthermore, there is a bijection, called \emph{correspondence}, mapping each element of $S$ into its \emph{corresponding} target in the target set. Such a bijection preserves the cyclic ordering around the center of SED, that is, if $t$ is the target corresponding to point $p\in S$, then the next point $p'\in S$ in the clockwise order around the center or SED is mapped to the target $t'$ that follows $t$ in the clockwise order around the center of SED. Therefore, in order to fully define a correspondence between points of $S$ and targets, it is sufficient to define it on one point.

The targets and the point-target correspondence are identified as follows. We first define a set $S'$: if $S$ is \RD, then $S'=\mathcal E(S)\cup R$, where $R$ is the principal relocation of $\mathcal I(S)$; otherwise, $S'=\mathcal F(S)$.
\begin{itemize}
\item Suppose that $S'$ has no axes of symmetry (i.e., it is \UP or \UA) and $S$ is not \RD. We let $\mathcal T$ be the set of all concordance classes of $S'$ that have the greatest number of points. Let $\widetilde{\mathcal T}$ be the subset of $\mathcal T$ containing the concordance classes $C\in\mathcal T$ for which there exists a movable analogy class $A$ of $S'$, with $A\cap C=\varnothing$, and a relocation $R_{C,A}$ of $\mathcal F'(A)$ (with respect to $(S'\setminus A)\cup \mathcal F'(A)$) such that $C\cup R_{C,A}$ is a concordance class of $(S'\setminus A)\cup R_{C,A}$. If $\widetilde{\mathcal T}$ is empty (respectively, not empty), we let $T$ be the concordance class of $\mathcal T$ (respectively, $\widetilde{\mathcal T}$) containing the points that induce the lexicographically smallest angle sequence with respect to $S'$. By definition, $T$ is a subset of the target set. Furthermore, each point $p\in S$ such that $\mathcal F(p)\in T$ corresponds to $\mathcal F(p)$. The rest of the target set and the other correspondences are determined accordingly.
\item Suppose that $S'$ has no axes of symmetry (i.e., it is \UP or \UA) and $S$ is \RD. We let $\mathcal T$ be the set of all concordance classes of $S'$ that have the greatest number of points in $\mathcal E(S)$. Let $\widetilde{\mathcal T}$ be the subset of $\mathcal T$ containing the concordance classes $C\in\mathcal T$ for which there exists a relocation $R_C$ of $\mathcal I(S)$ (with respect to $S$) such that $(\mathcal E(S)\cap C)\cup R_C$ is a concordance class of $\mathcal E(S)\cup R_C$. If $\widetilde{\mathcal T}$ is empty (respectively, not empty), we let $T$ be the concordance class of $\mathcal T$ (respectively, $\widetilde{\mathcal T}$) containing the points that induce the lexicographically smallest angle sequence with respect to $S'$. By definition, $T$ is a subset of the target set. Furthermore, each point of $\mathcal E(S)$ that coincides with a point of $T$ corresponds to that target. The rest of the target set and the other correspondences are determined accordingly.

\item If $S'$ has some axes of symmetry and a point $p\in S'$ lies on one of them, then $p$ coincides with a target $t$, by definition. Also, if $p\in S$, then $t$ corresponds to $p$. Otherwise, $t$ corresponds to the unique point $p'\in S$ that lies in the occupied sector containing $t$. The other targets and correspondences are determined accordingly (this definition is sound, due to Remark~\ref{rem:targets} below).
\item Finally, suppose that $S'$ has some axes of symmetry, but no point of $S'$ lies on any of them. Then, if $\ell$ is an axis of symmetry of $S'$, the target set is chosen in such a way that it has $\ell$ as an axis of symmetry as well, and no target lies on $\ell$. Also, each point $p\in S'$ at minimum distance from $\ell$ corresponds to the closest to $p$ among the targets that have minimum distance from $\ell$. The other targets and correspondences are determined accordingly (this definition is sound, due to Remark~\ref{rem:targets} below).
\end{itemize}

\begin{remark}\label{rem:targets}
From Proposition~\ref{p:targetcompatible} it follows that, even if $S$ has several axes of symmetry, it has a unique target set, and a unique point-target correspondence. (If $S$ has no axes of symmetry, this is true by construction.) Also, if $S$ is the set of locations of the robots in a swarm, the target set of $S$ is correctly computable by all robots, regardless of their position and handedness, because so are angle sequences, principal relocations, and footprints.
\end{remark}

\begin{proposition}\label{p:ontarget}
Let $S$ be a \VA set such that each point of $S$ lies on SEC$(S)$ and no point of $S$ is on its corresponding target. Then $S$ has an axis of symmetry on which no point of $S$ lies.
\end{proposition}
\begin{proof}
Since $S$ has no internal points, it is \WA and not \RD. If $S$ had no axes of symmetry, the points from one concordance class would lie on their corresponding targets. Hence $S$ has at least one axis of symmetry $\ell$. If a point of $S$ lay on $\ell$, it would coincide with its target. Hence no point of $S$ lies on $\ell$.
\end{proof}

\paragraph{Reachable points and sets.}
A point $q\in\mathbb R^2$ is \emph{reachable} by point $p\in S$ if $q$ and $p$ lie in the interior of the same main sector of $S\setminus \{p\}$. Equivalently, $p$ \emph{can reach} $q$.

\paragraph{Satisfied and improvable analogy classes.}
A point $p$ of a \VA and \WA set $S$ is \emph{satisfied} if $\mathcal F(p)$ coincides with the target of $p$. An analogy class of $S$ is \emph{satisfied} if all its points are satisfied. An analogy class of $S$ is \emph{improvable} if it is movable, not satisfied, and each of its points can reach its corresponding target.

\begin{observation}\label{o:concordance}
In a \VA and \WA set, all the points that lie at their respective targets belong to the same concordance class. Hence, any two points that belong to some satisfied analogy class are concordant.
\end{observation}

\paragraph{Locked configurations and unlocking analogy classes.}
A \VA and \WA set is said to be \emph{locked} if it has more than one analogy class, and no analogy class is improvable (see Figure~\ref{f06}). If $S$ is locked, then any movable analogy class of $S$ that contains points that are consecutive to some point in a non-movable analogy class of $S$ is said to be an \emph{unlocking} analogy class.

\begin{proposition}\label{p:locked1}
Let $S$ be a locked \VA and \WA set. Then, $S$ has at least one non-movable analogy class.
\end{proposition}
\begin{proof}
Without loss of generality, we assume that all points of $S$ lie on SEC. If this is not the case, we may equivalently consider $\mathcal F(S)$ instead of $S$.

Assume for a contradiction that $S$ is locked and all its analogy classes are movable. By definition of locked, $S$ is neither \EQ nor \BI, and every point of $S$ is either on its own target, or it cannot reach its target. Suppose first that there is a point $p\in S$ located on its own target, and label every point of $S$ that coincides with its own target as ``on''. Then imagine walking around SEC clockwise starting from $p$, and label every unlabeled point $q\in S$ that is encountered as ``before'' (respectively, ``after'') if the target of $q$ has not been encountered yet (respectively, has already been encountered). The walk starts and ends at $p$, hence the sequence of labels starts with an ``on'' and ends with an ``on''. Also, there must be labels other than ``on'', otherwise $S$ would coincide with its target set and it would be \EQ. If the sequence of labels has at least one ``before'', then the last ``before'' in the sequence must be followed by an ``on'' or an ``after''. But this means that the last point labeled ``before'' is not on its target and it can reach it, which is a contradiction. Otherwise, there are just ``on''s and ``after''s in the label sequence. But in this case the first point in the sequence that is labeled ``after'' is not on its target and it can reach it, because it is preceded by a point labeled ``on''. Hence we have a contradiction in both cases.

Suppose now that no point of $S$ is on its target. Then $S$ has an axis of symmetry $\ell$ on which no point of $S$ lies, by Proposition~\ref{p:ontarget}. Moreover, $\ell$ is an axis of symmetry of the target set of $S$, as well. Let us walk around SEC clockwise starting from $\ell$, and label the points of $S$ as described in the previous paragraph. By assumption no point is labeled ``on'', hence all points are labeled either ``before'' or ``after''. Also, a point is labeled ``before'' if and only if its symmetric point with respect to $\ell$ is labeled ``after''. It follows that there must be a point labeled ``before'' followed by a point labeled ``after'' (wich may be the last and the first point in the sequence, respectively). These two points are not on their targets but they can reach their targets, which is once again a contradiction.
\end{proof}

\begin{proposition}\label{p:consunmovable}
If a \VA and \WA and \UA set $S$ has two non-movable analogy classes $\{p\}$ and $\{q\}$, then $p$ and $q$ are consecutive points of $S$.
\end{proposition}
\begin{proof}
Since $S$ is \UA, every analogy class of $S$ consists of a single point, due to Observation~\ref{o:analogy2}. By Observation~\ref{o:unmovable}, there exists a closed half-plane bounded by a line through the center of SEC that contains $p$ and no other points of $S$, and there exists a similar half-plane for $q$. These two half-planes must have a non-empty intersection, so suppose that point $v\in \mbox{SEC}$ lies in the intersection. This means that the (shortest) arc ${\overset\frown{vp}}\subset \mbox{SEC}$ and the (shortest) arc ${\overset\frown{vq}}\subset \mbox{SEC}$ are devoid of points of $S\setminus\{p,q\}$. Therefore $p$ and $q$ are consecutive in $S$.
\end{proof}

\begin{proposition}\label{p:locked2}
Let $S$ be a locked \VA and \WA set whose points all lie on SEC. Then, $S$ is \AP. Moreover, if $S$ is \UA, then
\begin{itemize}
\item $S$ has either one or two non-movable analogy classes, each consisting of a single point;
\item if $S$ has two non-movable analogy classes $\{p\}$ and $\{q\}$, then $p$ and $q$ are consecutive points of $S$;
\item $S$ has exactly two unlocking analogy classes, each consisting of a single point.
\end{itemize}
Otherwise $S$ is \BA, and
\begin{itemize}
\item $S$ has a unique non-movable analogy class, which consists of two consecutive points of $S$;
\item $S$ has a unique unlocking analogy class consisting of two points.
\end{itemize}
Also, if $n>5$, at least one unlocking analogy class of $S$ is not satisfied.
\end{proposition}
\begin{proof}
Without loss of generality, we assume that all points of $S$ lie on SEC. If this is not the case, we may equivalently consider $\mathcal F(S)$ instead of $S$.

By Proposition~\ref{p:locked1}, $S$ has at least one non-movable analogy class. Also, by Proposition~\ref{p:unmovable}, $S$ is not \PER. Since, by definition of locked, $S$ is neither \EQ nor \BI, it must be \AP.

Suppose that $S$ is \UA. Then, every analogy class of $S$ consists of a single point, due to Observation~\ref{o:analogy2}. If, by contradiction, $S$ had three non-movable analogy classes, the three points they involve would have to be mutually consecutive, due to Proposition~\ref{p:consunmovable}. Equivalently, $S$ would consist of only three points, contradicting the definition of \VA set, stating that $n>4$. Hence $S$ has either one or two non-movable analogy classes, whose points are consecutive.

Suppose now that $S$ is \BA, and hence it has a (unique) axis of symmetry $\ell$. As already noted, $S$ has at least one non-movable analogy class. Suppose for a contradiction that $S$ has two analogy classes $C$ and $C'$, each of which, by Observation~\ref{o:analogy2}, consists of either one or two points, and is symmetric with respect to $\ell$. By Observation~\ref{o:unmovable}, there exists a line $\ell'$ through the center of SED bounding a closed half-plane that contains no points of $S$ other than those of $C$. Without loss of generality, due to the symmetry of $S$, we may assume that $\ell'$ is perpendicular to $\ell$. By a similar reasoning, the other closed half-plane bounded by $\ell'$ contains no points of $S$ other than the points of $C'$. We conclude that $S=C\cup C'$, and therefore $|S|\leqslant 4$, contradicting the assumption that $n>4$. Hence $S$ has exactly one non-movable analogy class $C$, which may consist of either one or two points. Suppose for a contradiction that $C$ consists of a single point $p$. Then $p$ must lie on the axis of symmetry $\ell$, and the closed half-plane $\Gamma$ bounded by $\ell'$ that contains $p$ contains no other points of $S$. Let $C''$ be the analogy class consisting of the two points that are consecutive to $p$. Since $p$ lies on an axis of symmetry of $S$, by definition $C$ is satisfied. Also, since $n>4$, the targets of the two points of $C''$ lie in $\Gamma$, while the points of $C''$ do not.  It follows that $C''$ is improvable (recall that $C$ is the only non-movable analogy class), which contradicts the fact that $S$ is locked. Therefore $C$ must consist of two points, i.e., $C=\{p,q\}$. The fact that $p$ and $q$ must be consecutive follows from Observation~\ref{o:unmovable} and the fact that $S$ is symmetric with respect to $\ell$.

In all cases, $S$ has either one or two consecutive points that belong to some non-movable analogy class. Let $L$ be the set of such points, with $1\leqslant |L|\leqslant 2$. Hence, because $n>4$,  there are exactly two points of $S\setminus L$ that are consecutive to some point of $L$, and which belong to some unlocking analogy class. Let $U$ be the set of these points, with $|U|=2$. If $S$ is \UA, each analogy class consists of a single point, and therefore there are exactly two unlocking analogy classes. If $S$ is \BA, the two points of $U$ are symmetric with respect to the axis of symmetry of $S$, and hence they belong to the same analogy class. In this case, there is exactly one unlocking analogy class.

Observe that, in all cases, there exists a line through the center of SED that leaves all the points of $L$ in one open half-plane and all the points of $S\setminus L$ in the other open half-plane. Therefore all the points of $S\setminus L$, hence at least $n-2$ points, lie in the sector defined by the two points of $U$. However, if $n>5$, the two points of $U$ cannot be concordant, because otherwise their angular distance would be at least $2\pi(n-3)/n\geqslant \pi$, which is a contradiction. It follows that, if $n>5$, the two points of $U$ do not belong to the same concordance class, and hence at least one of them belongs to a non-satisfied analogy class, due to Observation~\ref{o:concordance}.
\end{proof}

\paragraph{Walkers.}
Suppose that $S$ is \VA and all points of $S$ are on SEC. Then we can identify a set of \emph{walkers}, denoted by $\mathcal W(S)$, as follows.
\begin{itemize}
\item If $S$ has only one analogy class, $\mathcal W(S)=\varnothing$.
\item Otherwise, if $S$ is not locked, $\mathcal W(S)$ is the improvable analogy class whose points induce the lexicographically smallest angle sequence.
\item Otherwise, if $S$ is locked and $n>5$, then $\mathcal W(S)$ is the non-satisfied unlocking analogy class whose points induce the lexicographically smallest angle sequence (by Proposition~\ref{p:locked2}, such an analogy class exists).
\item Otherwise $S$ is locked and $n=5$. In this case, the walkers are the unlocking analogy class whose points induce the lexicographically smallest angle sequence.
\end{itemize}
In general, if $S$ is \VA and \WA, we define the set of walkers of $S$ as $\mathcal W(S)=\{p\in S\mid \exists p'\in \mathcal W(\mathcal F(S)),\ \mathcal F(p)=p'\}$.

\begin{observation}\label{o:walkers}
Let $S$ be a \VA and \WA set with more than one analogy class. Then, $\mathcal W(S)$ is a movable analogy class. If $n>5$, $\mathcal W(S)$ is also a non-satisfied analogy class of $S$.
\end{observation}

\paragraph{Finish set and point-finish-line correspondence.}
Suppose that $S$ is \VA and \RD. Then we can define the \emph{finish set} of $\mathcal I(S)$, which is the union of $|\mathcal I(S)|$ \emph{finish lines}, each of which is a half-line emanating from the center of SED$(S)$.

We first define the \emph{tentative finish set} $R$ as follows. Let $P$ be the principal relocation of $\mathcal I(S)$.
\begin{itemize}
\item If $P$ is a proper subset of an analogy class of $\mathcal E(S)\cup P$ (as in Figure~\ref{f08}(b)), we let $R=P$.
\item Otherwise, if the set of targets $T$ of the internal points of $S$ is a relocation of $\mathcal I(S)$, we let $R=T$.
\item Otherwise, we let $R=P$.
\end{itemize}
Now we define the finish set as follows.
\begin{itemize}
\item Suppose that the set $S'=\mathcal E(S)\cup R$ is locked and $R$ is an unlocking analogy class of $S'$. Then, by Proposition~\ref{p:locked2}, $S'$ is \AP.
\begin{itemize}
\item If $S'$ is \UA, then $R=\{r\}$. Let $\{r'\}$ be the unique non-movable analogy class of $S'$ such that $r$ and $r'$ are consecutive in $S'$ (cf.~Proposition~\ref{p:locked2}). Let $r''\in S'$ be the other point that is consecutive to $r'$. Then, the point of SEC$(S)$ that is antipodal to $r''$ belongs by definition to the finish set of $\mathcal I(S)$ (note that this implicitly defines the whole finish set).
\item If $S'$ is \BA, then $|R|=2$ (cf.~Proposition~\ref{p:locked2}). Let $R'$ be the relocation of $\mathcal I(S)$ consisting of two antipodal points on SEC$(S)$ such that $R'$ is an analogy class of $\mathcal E(S)\cup R'$, as shown in Figure~\ref{f06}(b) (see Proposition~\ref{p:finishsamesector} below for a proof that this definition is sound). Then, $R'$ is a subset of the finish set of $\mathcal I(S)$ (this implicitly defines the whole finish set).
\end{itemize}
\item Otherwise, $R$ is a subset of the finish set of $\mathcal I(S)$ (again, this implicitly defines the whole finish set).
\end{itemize}

\begin{proposition}\label{p:finishsamesector}
Let $S$ be a \VA and \RD set. Then there is a unique bijective function that maps each point $p\in\mathcal I(S)$ to a finish line $\ell$ lying in the same occupied sector of $S$ as $p$, and that preserves the relative clockwise ordering around the center of SED$(S)$.
\end{proposition}
\begin{proof}
It suffices to show that there is a relocation of $\mathcal I(S)$ with one point on each finish line. Then we can construct our bijective function by simply mapping internal points within each occupied sector in the right order. But if the tentative finish set $R$ is a subset of the finish set, then our claim is obvious, because the tentative finish set is a relocation of $\mathcal I(S)$, by construction. Otherwise, it means that $S'=\mathcal E(S)\cup R$ is locked and $R$ is an unlocking analogy class of $S'$, by definition of finish set.

Suppose that $S'$ is \UA, and let $r$, $r'$, and $r''$ be as in the definition of finish set. By Observation~\ref{l:sechull}, the antipodal point of $r''$ must lie on the arc $\overset\frown{rr'}$, or there would be an empty half-circle between $r'$ and $r''$. Moreover, the antipodal point of $r''$ cannot coincide with $r'$, or $S'$ would be \HD, implying that also $S$ is \HD (because $R=\{r\}$ is a relocation of $\mathcal I(S)$), which contradicts the fact that $S$ is \VA. It follows that $r$ can reach the antipodal point of $r''$ and therefore the unique point of $\mathcal I(S)$ can reach the unique finish line.

Suppose now that $S'$ is \BA, and therefore has an axis of symmetry $\ell$. By Proposition~\ref{p:locked2}, $S'$ has a unique non-movable analogy class $C$, which also has $\ell$ as an axis of symmetry. Moreoever, by Observation~\ref{o:unmovable}, there is a line $\ell'$ through the center of SED$(S')$ bounding a half-plane whose intersection with $S'$ is precisely $C$. Without loss of generality, we may take $\ell'$ to be perpendicular to $\ell$. Let $R'=\ell'\cap \mbox{SEC}(S')$. As $R$ is the unlocking analogy class of $S'$, its two elements are closest to $\ell'$ among all the points of $S'\setminus C$. It follows that $R'$ is a relocation of $\mathcal I(S)$, unless $R'=C$. But $R'=C$ implies that $S'$ is \HD, which makes $S$ \HD as well, contradicting the fact that $S$ is \VA. Hence $R'$ is a relocation of $\mathcal I(S)$, and is also a subset of the finish set of $\mathcal I(S')$, by definition. This concludes the proof, and incidentally also proves that the definition of finish set in this case is sound.
\end{proof}

The function whose existence and uniqueness is established by Proposition~\ref{p:finishsamesector} is called \emph{correspondence}. If correspondence maps point $p\in\mathcal I(S)$ to the finish line $\ell$, then $\ell$ is said to \emph{correspond} to $p$.

\begin{proposition}\label{p:finishreachable}
Let $S$ be a \VA and \RD set. Then, at least one internal point of $S$ can reach any point on its corresponding finish line.
\end{proposition}
\begin{proof}
By Proposition~\ref{p:finishsamesector}, the finish line corresponding to each point $p\in\mathcal I(S)$ lies in the same occupied sector as $p$. Moreover, Proposition~\ref{p:readyinternal} states that each occupied sector contains either one or two internal points. So, if an occupied sector contains exactly one internal point, it can certainly reach its corresponding finish line. If an occupied sector contains two internal points, and since correspondence preserves the relative clockwise ordering around the center of SED, it is easy to see that at least one of the two internal points can reach its corresponding finish line. Indeed, if a segment joining one of these two internal points to its corresponding finish line contains a point that is co-radial with the other internal point, it means that the other internal point can reach its corresponding finish line.
\end{proof}

\begin{proposition}\label{p:internalclass}
Let $S$ be a \VA and \RD set. Then, all the points of the principal relocation $P$ of $\mathcal I(S)$ are analogous in $S'=\mathcal E(S)\cup P$. Also, if $L$ is the relocation of $\mathcal I(S)$ (with respect to $S$) having one point on each finish line of $S$, then all the points of $L$ are analogous in $S''=\mathcal E(S)\cup L$.
\end{proposition}
\begin{proof}
By definition of \RD, there exists a relocation $A$ of $\mathcal I(S)$ such that $A$ is an analogy class of $S^*=\mathcal E(S)\cup A$. It is clear that $\mbox{SED}(S)=\mbox{SED}(S')=\mbox{SED}(S'')=\mbox{SED}(S^*)$. By definition of analogy class, there exist two constants $\gamma$ and $\gamma'$ such the angular distances (with respect to the center of SED) between any point of $A$ and its two consecutive points in $S^*$ are, respectively, $\gamma$ and $\gamma'$. Recall that, by Proposition~\ref{p:readyinternal}, either all occupied sectors of $S$ contain one point, or they all contain two points. Suppose first that they all contain one point. Then, each point of $P$ has angular distance $(\gamma+\gamma')/2$ from both its consecutive points in $S'$. Since all the other angular distances between consecutive points of $S'$ involve points of $\mathcal E(S)$ only, they are the same as in $S^*$. Therefore all the points of $P$ are analogous in $S'$, as the points of $A$ are analogous in $S^*$. Now suppose that all the occupied sectors of $S$ contain two points. Without loss of generality, let $\gamma$ be the angular distance between any two consecutive points of $A$ (with respect to the center of SED). Then, each point of $P$ has angular distance $(\gamma+2\gamma')/3$ from both its consecutive points in $S'$. Again, this implies that all points of $P$ are analogous in $S'$.

Let $T$ be the set of targets of the internal points of $S$, and let $R$ be the tentative finish set of $S'$. By definition, either $R=P$ or $R=T$. If $R=P$ and $L=R$, the points of $L$ are analogous in $S''$ because they are the principal relocation of $\mathcal I(S)$. Suppose instead that $R=T$ and $L=R$. This is true only if $T$ is a relocation of $\mathcal I(S)$. If $S^*$ has an axis of symmetry $\ell$, then, by Proposition~\ref{p:symm}, $A$ does too. It is easy to see that also $S'$ and $P$ have the same axis of symmetry. But $\ell$ is also an axis of symmetry of the target set of $S^*$, by definition of target set, and also of $T$, because $T$ is a subset of the target set that is also a relocation of $\mathcal F'(A)$. Since this holds for every axis of $S^*$, it easily follows that all the points of $T$ are analogous in $S''$. Suppose now that $S^*$ is \UP with period $k\geqslant 3$. This implies that $S^*$ has an $(n/k)$-fold rotational symmetry with respect to the center of SED. Since $A$ is an analogy class of $S^*$, it also has an $(n/k)$-fold rotational symmetry, by Proposition~\ref{p:symm}. In this case, every occupied sector of $S$ contains exactly one internal point. But also the target set of $S^*$ has an $(n/k)$-fold rotational symmetry, being a \RE set of $n$ points. Moreover, since the points of $A$ are all concordant, the points of $T$ must be all concordant, as well. This implies that all points of $T$ are analogous in $S''$. Finally, suppose that $S^*$ has no axes of symmetry and it is not \UP, and hence it is \UA. In this case, $T$ consists of a single point, and therefore there is nothing to prove.

The only cases left to consider are those in which $L\neq R$. By definition of finish set, this only happens when $\mathcal E(S)\cup R$ is locked and $R$ is an unlocking analogy class. If $\mathcal E(S)\cup R$ is \UA, then $L$ consists of a single point, and there is nothing to prove. If $\mathcal E(S)\cup R$ is \BA, then $L$ consists of two antipodal points that are symmetric with respect to an axis of symmetry of $S''$. This implies that the two points of $L$ are analogous in $S''$.
\end{proof}

\subsection{Algorithm}\label{sec:code}
The \UCF algorithm consists of an ordered set of tests to determine the class of the current configuration. For each class, we have a procedure that recognizes it: procedure {\sc Is Regular?}$(S)$ determines if $S$ is a \RE configuration, and so on. The implementation of all these procedures is straightforward and is therefore omitted, with the exception of procedure {\sc Is Pre-regular?}, which will be described in Section~\ref{sec:ispreregular}, and procedure {\sc Is Valid and Ready?}, which will be described in Remark~\ref{rem:welloccupied}. After the configuration class has been determined, the executing robot takes the appropriate action.

We stress that some configurations belong to more than one class, and so the order in which such classes are tested by the algorithm matters.

\begin{algorithmic}[h]
\begin{algorithm}
\REQUIRE{$S\subset \mathbb R^2$ is a finite set with $|S|>5$, and $(0,0)\in S$. $S$ represents the set of positions of the robots, as observed by the executing robot. The executing robot's position is $(0,0)$.}
\IFi{{\sc Is Regular?}$(S)$}{{\bf Do Nothing}}
\ELSIFi{{\sc Is Pre-regular?}$(S)$}{{\bf Execute} {\sc Pre-regular}$(S)$}
\ELSIFi{{\sc Is Central?}$(S)$}{{\bf Execute} {\sc Central}$(S)$}
\ELSIFi{{\sc Is Half-disk?}$(S)$}{{\bf Execute} {\sc Half-disk}$(S)$}
\ELSIFi{{\sc Is Co-radial?}$(S)$}{{\bf Execute} {\sc Co-radial}$(S)$}
\ELSIFi{{\sc Is Valid and Ready?}$(S)$}{{\bf Execute} {\sc Valid and Ready}$(S)$}
\ELSIFi{{\sc Is Valid and Waiting?}$(S)$}{{\bf Execute} {\sc Valid and Waiting}$(S)$}
\ELSEi{{\bf Execute} {\sc Invalid}$(S)$}
\caption*{{\bf Algorithm} \UCF$(S)$}
\end{algorithm}
\end{algorithmic}

\begin{remark}\label{rem:welloccupied}
Procedure {\sc Is Valid and Ready?}$(S)$ should verify if $S$ is well occupied. To do this, it is not necessary to check every possible relocation of $\mathcal I(S)$; it is sufficient to check only two of them. First construct $S'=\mathcal E(S)\cup P$, where $P$ is the principal relocation of $\mathcal I(S)$. If some points of $P$ are not analogous in $S'$, return ``false''; if $P$ is an analogy class of $S'$, return ``true''. Otherwise, construct a second configuration $S''=\mathcal E(S)\cup P'$, where $P'$ is another relocation of $\mathcal I(S)$, obtained by moving the points of $P$ symmetrically within the same principal sectors of $\mathcal E(S)$ (in such a way as to keep them analogous) by any angle that is incommensurable with all the angular distances between pairs of points of $S'$. Then return ``true'' if $P'$ is an analogy class of $S''$. It is easy to see that, if $P'$ is not, then no other relocation of $\mathcal I(S)$ can be an analogy class, and therefore we can safely return ``false''.
\end{remark}

Before detailing the main procedures, we introduce a few auxiliary ones, and some terminology.

\subsubsection{Auxiliary Procedures}

\paragraph{Radial and lateral moves.}
We distinguish two types of moves that the robots can perform. If the destination point computed by a robot is co-radial with the current robot's position (with respect to the center of the SED of the observed robots' locations), then we say that the robot performs a \emph{radial} move, or moves \emph{radially}. If a move is not radial, it is said to be \emph{lateral}.

\paragraph{Procedure {\sc Cautious Move}.}
This procedure makes a subset of robots $\mathcal M$ execute a cautious move with a given set of \emph{critical points} $C$. All robots of $\mathcal M$ move radially, either all from SEC/3 to SEC, or all from SEC to SEC/3. The line segment connecting the center of SED with a robot in $\mathcal M$'s co-radial point on SEC is called the \emph{path} of the robot. If a robot is directed toward SEC (respectively, SEC/3), the point on SEC (respectively, SEC/3) on the robot's path is called the \emph{endpoint} of the path. The procedure first augments the set of input critical points $C$ with a set of \emph{auxiliary} critical points (which may be \emph{final}, \emph{transposed} or \emph{intermediate} critical points), and then lets a robot move toward the next critical point (auxiliary or not) along its path, provided that some conditions are met. The details are as follows.
\begin{itemize}
\item The endpoint of each robot's path is added to the set of critical points. This auxiliary critical point is called \emph{final}.
\item For every robot $r$ and every critical point $p$, a critical point is added on $r$'s path at the same distance from the center of SED as $p$. If not already present in the critical point set, such an auxiliary critical point is called \emph{transposed}.
\item For each pair of critical points on each robot's path (which may be critical points of $C$, or final, or transposed), the midpoint is added as a critical point. If not already present in the critical point set, such an auxiliary critical point is called \emph{intermediate}.
\item The robots that are not farthest from the endpoints of their respective paths are not allowed to start moving.
\item The robots that are farthest from the endpoints of their respective paths move to the next critical point (auxiliary or not) along their respective paths.
\end{itemize}

\begin{algorithmic}[h]
\begin{algorithm}
\caption*{{\bf Procedure} {\sc Cautious Move} $(S,\mathcal M, C, \mbox{dir})$}
\REQUIRE{$S$ is not \CO. $\mathcal M\subseteq S$ is the set of robots that have to perform the move. $C\subset \mathbb R^2$ is a finite set of critical points. dir is the direction in which the robots of $\mathcal M$ should move: its value is either ``SEC'' or ``SEC/3''. If $\mbox{dir}=\mbox{``SEC/3''}$, then no point of $S$ lies in the interior of SED/3$(S)$.}
\IF[I am one of the robots that should do the cautious move]{$(0,0)\in \mathcal M$}
\STATE $c\longleftarrow$ center of SED$(S)$
\STATE $P\longleftarrow$ set of points collinear with $c$ and $(0,0)$
\STATE $\mbox{proceed}\longleftarrow \mbox{\bf true}$
\IF{$\mbox{dir}=\mbox{``SEC''}$}
\STATE $d\longleftarrow P\cap \mbox{SEC}(S)$
\FORALL{$r\in \mathcal M$}
\IFi{$\|rc\|<\|c\|$}{$\mbox{proceed}\longleftarrow \mbox{\bf false}$}
\ENDFOR
\ELSE
\STATE $d\longleftarrow P\cap \mbox{SEC/3}(S)$
\FORALL{$r\in \mathcal M$}
\IFi{$\|rc\|>\|c\|$}{$\mbox{proceed}\longleftarrow \mbox{\bf false}$}
\ENDFOR
\ENDIF
\IF[I am farthest from the endpoint]{proceed}
\STATE $C'\longleftarrow C\cup\{d\}$
\FORALL{$p\in C$}
\STATE $p'\longleftarrow$ point on $P$ such that $p$ and $p'$ are equidistant to $c$
\STATE $C'\longleftarrow C'\cup\{p'\}$
\ENDFOR
\STATE $C'\longleftarrow C'\cap P$
\STATE $C''\longleftarrow C'$
\FORALL{$p,q\in C'$}
\STATE $C''\longleftarrow C''\cup\{(p+q)/2\}$
\ENDFOR
\STATE $\mbox{dest}\longleftarrow d$
\FORALL{$p\in C''$}
\IFi{$\|pd\|<\|d\|$ {\bf And} $\|p\|<\|\mbox{dest}\|$}{$\mbox{dest}\longleftarrow p$}
\ENDFOR
\STATE {\bf Move To} dest
\ENDIF
\ENDIF
\end{algorithm}
\end{algorithmic} 

\paragraph{Procedure {\sc Move Walkers to SEC/3}.}
This procedure assumes the configuration to be \VA and \WA, and it assumes all internal robots to be walkers. It makes all walkers move radially toward SEC/3, executing the {\sc Cautious Move} procedure with suitable critical points intercepting the possible \PR configurations that may be formed (the exact locations of the critical points will be discussed in Section~\ref{sec:analysis}).

\paragraph{Procedure {\sc Move All to SEC}.}
This procedure assumes the configuration to be not \CO. First all robots that lie in the interior of SED/3 move radially to SEC/3. Then, the procedure selects a subset $C$ of robots and makes them move radially toward SEC, executing procedure {\sc Cautious Move} with suitable critical points intercepting the possible \PR configurations that may be formed (the exact locations of the critical points will be discussed in Section~\ref{sec:analysis}). The set $C$ is either an analogy class or a strong analogy class, and it is selected as follows.
\begin{itemize}
\item If the robots form a \BI configuration, all the robots on SEC belong to the same strong analogy class $C'$, and there are robots of $C'$ that are not on SEC, then $C=C'$.
\item If the robots form a \DBI configuration, all the robots on SEC belong to the same analogy class $C'$, and there are robots of $C'$ that are not on SEC, then $C=C'$.
\item Otherwise, among the least numerous analogy classes that are not entirely on SEC, $C$ is the one whose robots induce the lexicographically smallest angle sequence.
\end{itemize}

\begin{remark}
The reason why strong analogy classes are considered in the \BI case, as opposed to analogy classes, will be clear in the proof of Theorem~\ref{t:bsec4}. Similarly, the reason why the robots move in this fashion in the \DBI case will be apparent in the proof of Theorem~\ref{t:dbi}. The fact that, in all other cases, the least numerous analogy classes move first, will be used in the proof of Theorem~\ref{t:p1}.
\end{remark}

\paragraph{Procedure {\sc Move to Finish Line}.}
This procedure is executed when the configuration is \VA and \RD, and all robots lie either on SEC or on SEC/3.
Each internal robot $r$ makes a lateral move to the intersection $q$ between its corresponding finish line and SEC/3, provided that $q$ is reachable by $r$ (i.e., if no other robot is co-radial with any point on the segment from $r$'s location to $q$).

\subsubsection{Main Procedures}

\paragraph{Procedure {\sc Pre-regular}.}
Each robot moves to its matching vertex of the supporting polygon.

\paragraph{Procedure {\sc Central}.}
The robot at the center of SED moves toward any point on SEC/3 that is not co-radial with any other robot (any deterministic algorithm for choosing this point works).

\paragraph{Procedure {\sc Half-disk}.}
Note that this procedure is executed only if the configuration is not \CE, hence we may assume that no robot lies at the center $c$ of SED.
\begin{itemize}
\item Suppose that all robots are collinear and one of the principal rays contains fewer than three robots. Let $r$ be the robot that lies on the other principal ray and that is closest to $c$.
\begin{itemize}
\item If $r$ does not lie in SED/3, it moves radially toward SEC/3.
\item If $r$ lies in SED/3, it moves to a point on SEC/3 that has angular distance $\pi/3$ from its current position.
\end{itemize}
\item Otherwise, if the intersection between one of the principal rays $\ell$ and SED/3 contains no robots, let $s$ be the robot that lies at the intersection between $\ell$ and SEC. Then, the robot distinct from $s$ with smallest angular distance from $s$ that is closest to $c$ moves to the point of $\ell$ that lies on SEC/3.
\item Otherwise, both principal rays contain at least two robots, one of which is in SED/3. In this case, the robot on each principal ray that is closest to $c$ moves into an empty half-plane, to the point on SEC/3 that has angular distance $\pi/3$ from its current position.
\end{itemize}

\paragraph{Procedure {\sc Co-radial}.}
\begin{itemize}
\item If there are non-co-radial robots that lie in the interior of SED/3, they move radially to SEC/3.
\item Otherwise, if the co-radial robots that are closest to the center of SED do not lie in SED/3, they moves radially toward SEC/3.
\item Otherwise, each robot $r$ that is closest to the center of SED moves to a point on SEC/3 whose angular distance from $r$'s current location is $1/3$ of the smallest positive angular distance between two robots.
\end{itemize}
 
\paragraph{Procedure {\sc Valid and Ready}.}
\begin{itemize}
\item If there are robots in the interior of SED/3, they move radially to SEC/3.
\item Otherwise, if not all the internal robots lie on their corresponding finish lines, procedure {\sc Move to Finish Line} is executed.
\item Otherwise, procedure {\sc Valid and Waiting} is executed (indeed, if all the internal robots lie on their corresponding finish lines and on SEC/3, the configuration is \WA, due to Proposition~\ref{p:internalclass}).
\end{itemize}

\paragraph{Procedure {\sc Valid and Waiting}.}
\begin{itemize}
\item If all the internal robots are walkers, procedure {\sc Move Walkers to SEC/3} is executed.
\item Otherwise, procedure {\sc Move All to SEC} is executed.
\end{itemize}

\paragraph{Procedure {\sc Invalid}.}
Procedure {\sc Move All to SEC} is executed.

\section{Properties and Correctness} \label{sec:correctness}
To prove the correctness of the algorithm, we need to analyze the possible transitions between configurations.

In the following, we will closely examine all the possible flows of the algorithm in the space of robots' configurations, paying special attention to the transitions that may arise as critical points of cautious moves. In Section~\ref{sec:cautious} we prove some fundamental results on cautious moves, which show that robots executing the {\sc Cautious Move} protocol introduced in Section~\ref{sec:code} indeed behave as intended. In Section~\ref{sec:analysis} we thoroughly analyze the \PR configurations that may arise during a cautious move, and we produce critical points to intercept them. Then, in Section~\ref{sec:corr} we conclude the proof by showing that all the possible flows of the algorithm eventually reach a \RE configuration.

The diagram in Figure~\ref{f:diagram} shows the possible transitions between configurations. We will prove the correctness of this diagram in Section~\ref{sec:corr}, culminating with Theorem~\ref{main}.

\begin{figure}[h!]
\centering
\vspace{0.2cm}
\includegraphics[scale=1.0]{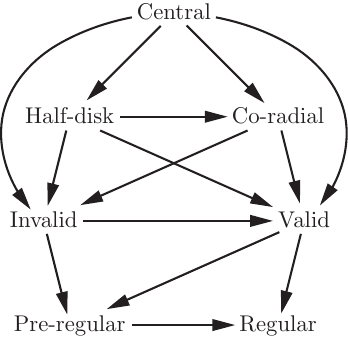}
\vspace{0.3cm}
\caption{Possible transitions between configurations of Algorithm~\UCF.}
\label{f:diagram}
\end{figure}

In this section, unless stated otherwise, $\mathcal R=\{r_1,\cdots, r_n\}$ will denote a swarm of $n>4$ robots. By $r_i(t)$ we denote the location of robot $r_i$ at time $t\geqslant 0$, and we let $\mathcal R(t)=\{r_1(t),\cdots, r_n(t)\}$.

\subsection{Correctness of Cautious Moves}\label{sec:cautious}
Let a set of robots execute the {\sc Cautious Move} protocol of Section~\ref{sec:code}, starting from a given frozen configuration $I$ and using a set of critical points $C$. We denote by $\mathcal E^\delta_{I,C}$ the set of all possible executions of such a robot system (recall the definition of execution from Section~\ref{sec:model}). Similarly to Section~\ref{sec:model}, we define $\mathcal E_{I,C}=\bigcup_{\delta>0} \mathcal E^\delta_{I,C}$, and we say that a cautious move with critical point set $C$ and initial configuration $I$ \emph{enjoys} the property $\mathcal P$ if $\mathcal E_{I,C}$ enjoys $\mathcal P$.

First we show that a cautious move always ``terminates'', that is, if every robot's path (either toward SEC or toward SEC/3) contains finitely many critical points, then after finitely many cycles the robot reaches the endpoint.

\begin{lemma}\label{l:cautious:0}
Suppose that a subset $\mathcal M$ of a swarm $\mathcal R$ of robots keeps executing the {\sc Cautious Move} protocol from a frozen initial configuration (while the robots of $\mathcal R\setminus\mathcal M$ remain still). Then, in a finite amount of time, each robot of $\mathcal M$ will be found at the endpoint of its path, and the swarm will be frozen again.
\end{lemma}
\begin{proof}
We define a \emph{round} to be a span of time in which every robot executes at least one complete cycle. Any execution can be decomposed (not necessarily in a unique way) into an infinite sequence of rounds. Let $\mathcal L(t)\subseteq\mathcal M$ be the set of robots that are farthest from the endpoints of their respective paths at time $t$, and let $d(t)$ be the distance of any robot in $\mathcal L(t)$ from the endpoint of its path at time $t$. Suppose for a contradiction that $d(t)>0$ for every $t$. Since $d(t)$ can only decrease in time, it converges to an infimum $m$.  Suppose first that the infimum is reached, i.e., $d(t)=m$ for some $t$. Then, after a round, say at time $t'$, all the robots in $\mathcal L(t)$ have moved, and hence $d(t')<m$, which is a contradiction.

Suppose that $d(t)>m$ for every $t$, and therefore the infimum is never reached. Let $t'$ be such that $d(t')-m<\delta$. Let $r\in \mathcal L(t')$ and let $p_r$ be the point on $r$'s path at distance $m$ from the endpoint. Since the critical points are finitely many, we may assume that no critical points (auxiliary or not) lie on the path of $r$ strictly between $r(t')$ and $p_r$. By our choice of $t'$, all the robots that perform a cycle at any time after $t'$ necessarily reach their destination point. Hence, after a round, each robot $r\mathcal L(t')$ has moved onto $p_r$ or past it, and therefore there exists a time $t''>t'$ such that $d(t'')\leqslant m$, which is a contradiction.

It follows that each robot eventually reaches the endpoint of its path. Since this is also a critical point and the robot is not moving in the initial configuration, it stops there. Afterwards, every time the robot performs a Look-Compute phase and some other robot has not reached the endpoint of its path yet, it waits. Eventually, when the last robots have reached the endpoints of their paths and they stop, none of the robots is moving, and therefore the configuration is frozen.
\end{proof}

Next we prove that cautious moves are sound, i.e., that if a configuration of points $C$ is taken as the input set of critical points of a cautious move, then, whenever the robots are found in configuration $C$, they freeze.

\begin{theorem}\label{t:cautious1}
Let a subset $\mathcal M$ of a swarm $\mathcal R$ of robots execute the {\sc Cautious Move} protocol with critical points $C$, with $|C|=|\mathcal R|=n$, from a frozen initial configuration. Then, during the cautious move, as soon as the swarm is found in configuration $C$, it freezes.
\end{theorem}
\begin{proof}
Because the paths of the robots of $\mathcal M$ are disjoint, $C$ can be formed only if each path contains exactly one point of $C$. Moreover, the other $n-|\mathcal M|$ points of $C$ must coincide with the locations of the robots in $\mathcal R\setminus\mathcal M$ (which remain still throughout the execution). By $c_r$ we denote the element of $C$ that lies on the path of robot $r\in\mathcal M$. Since a robot can only move toward the endpoint of its path, we may assume that each robot $r\in\mathcal M$ is initially located not past $c_r$ along its path, otherwise $C$ would never be formed during the cautious move. Let $d_r$ be the distance between $c_r$ and the endpoint of $r$'s path, and let $\mathcal H$ be the set of robots $r$ such that $d_r$ is maximum.

Suppose first that each robot $r\in \mathcal M\setminus \mathcal H$ initially lies at $c_r$. According to the {\sc Cautious Move} protocol, the only robots that are able to move in this situation are those in $\mathcal H$. By Lemma~\ref{l:cautious:0}, for every $r\in\mathcal H$ there exists a minimum time $t_r$ such that $r(t_r)=c_r$. Since this is a critical point, $r$ stops in $c_r$ at time $t_r$. Moreover, $r$ waits in $c_r$ until time $t^*=\max_{r\in \mathcal H}\{t_r\}$. Therefore, at time $t^*$, the robots form configuration $C$ for the first time, and none of them is moving. After that time, as soon as a robot $r$ moves, it passes $c_r$, and therefore $C$ cannot be formed any more.

Suppose now that some robots in $\mathcal M\setminus \mathcal H$ initially lie strictly before the element of $C$ on their respective path. For every $r\in \mathcal M\setminus \mathcal H$, let $f_r$ be the transposed critical point on the path of $r$ having distance $d_{r'}$ from the endpoint, with $r'\in \mathcal H$. Let $\mathcal H'$ be the set of robots $r\in \mathcal M\setminus \mathcal H$ such that $r$ is initially located in $f_r$ or before $f_r$. Let $\mathcal H''=\mathcal M\setminus(\mathcal H\cup \mathcal H')$. By our assumptions, $\mathcal H'\cup \mathcal H''$ is not empty. For every $r\in \mathcal H$, we define $t_r$ as in the previous paragraph. For $r'\in \mathcal H'$, we define $t_{r'}$ as the minimum time at which $r'$ is found in $f_{r'}$. Finally, we let $t^*=\max_{r\in \mathcal H\cup \mathcal H'}\{t_r\}$. By the {\sc Cautious Move} protocol, for every $r\in \mathcal H$, $r(t^*)=c_r$ and, for every $r'\in \mathcal H'$, $r'(t^*)=f_{r'}$. On the other hand, up to time $t^*$, no robot in $\mathcal H''$ has moved.

If $\mathcal H'$ is empty, then some robot $r\in \mathcal H''$ is not located in $c_r$ at time $t^*$ or before time $t^*$. Hence, the robots cannot form configuration $C$ until time $t^*$. After time $t^*$, the first robots that move are those in $\mathcal H$. When one of these robots moves, it goes past the element of $C$ that lies on its path, and therefore $C$ cannot be formed after time $t^*$, either.

Let $\mathcal H'$ be non-empty. Then, the robots cannot form configuration $C$ until time $t^*$, because each robot $r\in \mathcal H'$ is located strictly before $c_r$ at all times $t\leqslant t^*$. After time $t^*$, the first robots that are allowed to move are those in $\mathcal H\cup \mathcal H'$. For each $r\in \mathcal H\cup \mathcal H'$, let $t'_r\geqslant t^*$ be the first time at which robot $r$ performs a Look-Compute phase. Because each $r\in \mathcal H\cup \mathcal H'$ at time $t'_r$ lies at a critical point (possibly a transposed one), by the {\sc Cautious Move} protocol its destination is its next critical point, which is an intermediate one. In particular, if $r\in \mathcal H'$, its destination point is strictly before $c_r$. After such a robot $r$ has moved, it waits at least until after time $\max_{r\in \mathcal H\cup \mathcal H'}\{t'_r\}$. Indeed, before $r$ is allowed to move again, all the robots in $\mathcal H\cup \mathcal H'$ must ``catch up'' with it. However, as soon as a robot $r'\in \mathcal H$ moves after time $t^*$, it goes past $c_{r'}$, and therefore the configuration $C$ is not formable any more.
\end{proof}

Now we show that the cautious move protocol is ``robust'', in that merging two sets of critical points yields a cautious move that enjoys all the properties that are enjoyed when either set of critical points is taken individually.

\begin{lemma}\label{l:cautious2}
$\mathcal E_{I,C\cup \{p\}}\subseteq \mathcal E_{I,C}$.
\end{lemma}
\begin{proof}
By the {\sc Cautious Move} protocol, the addition of $p$ to the set of the input critical points causes the appearance on the path of each robot of at most one extra transposed critical point and at most $|C|+1$ extra intermediate critical points. However, by Lemma~\ref{l:cautious:0}, each robot still reaches the end of its path within finitely many turns in every execution. Let $E\in\mathcal E_{I,C\cup \{p\}}$ be an execution. We claim that $E\in\mathcal E^\delta_{I,C}$, for a suitable choice of a small-enough $\delta$. Let us order chronologically the (instantaneous) Look-Compute phases of all the robots in the execution $E$, resolving ties arbitrarily. We will prove by induction that, up to the $k$-th Look-Compute, $E$ coincides with some execution in $\mathcal E_{I,C}$.

Let us assume that our claim holds up to a certain $k$, and let us prove that it holds up to $k+1$. Let $r$ be the robot performing the $k$-th Look-Compute, say at time $t$, and let $q=r(t)$. If this is $r$'s first Look-Compute phase, there is nothing to prove. Otherwise, let $t'<t$ be the last time before $t$ at which $r$ performed a Look-Compute phase, according to $E$. Since $r$ must stop at every critical point on its path, there must be no critical points in the relative interior of the segment $r(t')q$. By the inductive hypothesis, $E$ coincides with some execution in $\mathcal E_{I,C}$, and therefore with some execution $E'\in\mathcal E^{\delta}_{I,C}$, for some $\delta>0$. In particular, $r$ performs a Look-Compute at time $t'$ in $E'$, as well. We may also assume that the $(k+1)$-th Look-Compute phase in $E'$ is performed by $r$ at time $t$, and that $E$ and $E'$ coincide at all times in the interval $[t',t)$. Since the critical point set of the cautious move with input $C$ is a subset of that of the cautious move with input $C\cup\{p\}$, the destination point of $r$ computed in $E'$ at time $t'$ cannot be closer to $r(t')$ than $q$. So, $r$ can be stopped in $q$ by the adversary even if the input critical point set is $C$, provided that $\delta$ is small enough. Specifically, if $d$ is the distance between $r(t')$ and $q$, such an execution can be found in $\mathcal E^{\min\{\delta,d\}}_{I,C}$, and therefore in $\mathcal E_{I,C}$.
\end{proof}

\begin{theorem}\label{t:cautious3}
Let the cautious move from a frozen initial configuration $I$ and critical point set $C_1$ (respectively, $C_2$) enjoy property $\mathcal P_1$ (respectively, $\mathcal P_2$). Then, the cautious move with initial configuration $I$ and critical point set $C_1\cup C_2$ enjoys both $\mathcal P_1$ and $\mathcal P_2$.
\end{theorem}
\begin{proof}
The theorem easily follows from Lemma~\ref{l:cautious2}: we add the critical points of $C_2$ to the set $C_1$, one by one. Each time we add a new point, by Lemma~\ref{l:cautious2} we have a set of executions that is a subset of the previous one, and therefore it still enjoys $\mathcal P_1$. Hence the cautious move with critical points $C_1\cup C_2$ enjoys property $\mathcal P_1$ and, by a symmetric argument, it also enjoys property $\mathcal P_2$.
\end{proof}

\begin{corollary}\label{c:cautious4}
Let a swarm of $n$ robots execute the {\sc Cautious Move} protocol with critical point set $\bigcup_{i=1}^k C_i$, with $|C_i|=n$ for $1\leqslant i\leqslant k$, from a frozen initial configuration. Then, during the cautious move, as soon as the robots are found in a configuration $C_i$, they freeze.
\end{corollary}
\begin{proof}
By Theorem~\ref{t:cautious1}, the cautious move with critical point set $C_i$ has the property $\mathcal P_i$ that, as soon as the robots are found in configuration $C_i$, they freeze. By repeatedly applying Theorem~\ref{t:cautious3}, we have that the cautious move with critical point set $\bigcup_{i=1}^k C_i$ enjoys all properties $\mathcal P_i$, for every $i$.
\end{proof}

\subsection{Analysis of \PR Configurations}\label{sec:analysis}
In this section, we prove several properties of \PR configurations that will be needed in the correctness proof of Section~\ref{sec:corr}. First we show that a \PR configuration cannot be \HD (Theorem~\ref{t:halfdisk}), it cannot be \CO (Theorem~\ref{t:coradial}), and it has no points in SED/3 (Theorem~\ref{t:sec3}). Then we prove that \PR configurations can effectively be taken as critical points during the execution of the algorithm, by showing that only finitely many \PR configurations are formable whenever a cautious move has to be made, or that the ``relevant'' \PR configurations that are formable are only finitely many.

In the following, we assume that $S\subset \mathbb R^2$ is a finite set of $n>4$ points, none of which lies at the center of SED. In particular, if $S$ is \PR, then $n\geqslant 6$, because in this case $n$ must be even. Since points model robots' locations, with abuse of terminology we will refer to points of $S$ that ``slide'' according to some rules. Formally, what we mean is that we consider $S$ as a function of time, so that $S(t)$ represents a set of robots' locations at time $t$; likewise a ``sliding'' point $a\in S$ will formally be a function $a(t)$ representing the trajectory of a robot.

\subsubsection{\HD Configurations, Co-Radial Points, and Points in SED/3}\label{sec:analysis1}

\begin{lemma}\label{l:prehull}
If $S$ is \PR, then $S$ is in strictly convex position, and in particular no three points of $S$ are collinear. Moreover, the convex hull of $S$ contains the center of the supporting polygon of $S$.
\end{lemma}
\begin{proof}
Let $P$ be the supporting polygon of $S$, which is regular and therefore convex. The fact that $S$ is in strictly convex position follows directly from the definition of \PR. Indeed, $S$ is a subset of the boundary of $P$, and no three points of $S$ lie on the same edge of $P$.

Let $c$ be the center of $P$, and let $a$ and $b$ be any two points of $S$ that lie at adjacent vertices of the convex hull of $S$. Since $S$ is in convex position, it is contained in a half-plane $\mathcal H$ bounded by the line $ab$. To prove that $c$ is contained in the convex hull of $S$, it is sufficient to show that it lies in $\mathcal H$ (because a convex polygon is the intersection of the half-planes determined by its own edges). If $a$ and $b$ are companions, $\mathcal H$ contains all of $P$, and therefore also its center. Otherwise, $\mathcal H$ entirely contains all edges of $P$, except at most three (i.e., the edges on which $a$ and $b$ lie, plus the edge between them). Since $P$ has at least six edges, it easily follows that $\mathcal H$ must contain its center.
\end{proof}

\begin{theorem}\label{t:halfdisk}
If $S$ is \PR, then it is not a \HD set.
\end{theorem}
\begin{proof}
Suppose by contradiction that $S$ is \PR and \HD, and let $\ell$ be the principal line. Due to Observation~\ref{l:sechull}, $\ell\cap \mbox{SEC}(S)$ consists of two antipodal points $a$ and $b$, both belonging to $S$. 

First assume that $a$ and $b$ belong to the same edge of the supporting polygon. Recall that the supporting polygon is a regular polygon, which implies that it has no other intersections with SED$(S)$ other than $a$ and $b$, as its edges are at least as long as the diameter of SEC$(S)$. This means that $n=2$, contradicting our assumption that $n>4$.

Hence $a$ and $b$ do not belong to the same edge of the supporting polygon. However, since every other edge of the supporting polygon must contain points of $S$, and the empty half-plane does not contain any point of $S$, it follows that $a$ and $b$ belong to two edges $AB$ and $CD$ of the supporting polygon, respectively, such that $BC$ is a third edge of the same polygon. Note that $AB$ does not lie on $\ell$, otherwise the companion of $a$ would be collinear with $a$ and $b$, contradicting Lemma~\ref{l:prehull}. Similarly, $CD$ does not lie on $\ell$. Since $n>4$, the supporting polygon is at least a hexagon, and therefore the extensions of $AB$ and $CD$ meet in the empty half-plane. On the other hand, let $\mathcal H$ be the part of SED$(S)$ that does not lie in the empty half-plane. Observe that the companion of $a$ lies in $AB\cap \mathcal H\setminus\{a\}$, and the companion of $b$ lies in $CD\cap \mathcal H\setminus\{b\}$. This implies that the extensions of $AB$ and $CD$ meet in the non-empty half-plane, which is a contradiction.
\end{proof}

\begin{lemma}\label{l:coradial}
If $S$ is \PR, then any ray from the center of SED intersects the perimeter of the supporting polygon in exactly one point.
\end{lemma}
\begin{proof}
Let a ray from the center of SED intersect the perimeter of the supporting polygon in exactly two points $a$ and $b$, none of which coincides with the center of SED. Then, by Lemma~\ref{l:prehull}, the intersection of the line through $a$ and $b$ with the convex hull of $S$ is exactly the segment $ab$, and therefore the center of SED does not belong to the convex hull of $S$. This contradicts Observation~\ref{l:sechull}.

Suppose now that an edge of the supporting polygon, belonging to a line $\ell$, is collinear with the center of SED$(S)$. Due to Lemma~\ref{l:prehull}, either $S$ lies entirely on $\ell$, in which case it cannot be \PR, or it is \HD with principal line $\ell$, which is impossible due to Theorem~\ref{t:halfdisk}.
\end{proof}

\begin{theorem}\label{t:coradial}
If $S$ is \PR, then it is not \CO.
\end{theorem}
\begin{proof}
If two points of $a,b\in S$ were co-radial, then the ray from the center of SED through $a$ and $b$ would intersect the perimeter of the supporting polygon in at least $a$ and $b$, contradicting Lemma~\ref{l:coradial}.
\end{proof}

\begin{theorem}\label{t:sec3}
If $S$ is \PR, then no points of $S$ lie in SED/3.
\end{theorem}
\begin{proof}
If $S$ is \PR, all points of $S$ lie on the perimeter of the same regular $n$-gon, with $n\geqslant 6$. Therefore, they all lie in an annulus $A$ with inner and outer radii $r'$ and $r''$ respectively, such that $r'/r''\geqslant \sqrt{3}/2$. Also, since the outer circle of $A$ encloses $S$, we have $r''\geqslant r$, where $r$ is the radius of SED, implying that $r'\geqslant (\sqrt{3}/2)r$.

Suppose for a contradiction that a point $p\in S$ lies in SED/3. Let $d$ be the distance between the center of SED and the center of $A$. Since $p$ must also lie in $A$, it follows that $d\geqslant r'-r/3\geqslant (\sqrt 3/2-1/3)r>0$. Therefore the set $\mbox{SED}\cup A$ has a unique axis of symmetry $\ell$. Let $\ell'$ be the axis of $A$ that is orthogonal to $\ell$, and let $a$ and $b$ be the two points of intersection between $\ell'$ and the inner circle of $A$. The distance between the center of SED and $a$ (or $b$) is
$$\sqrt{d^2+r'^2}\geqslant \sqrt{\left(\frac{\sqrt 3}2-\frac 13\right)^2+\left(\frac{\sqrt 3}2\right)^2}\cdot r>r,$$
which means that $a$ and $b$ lie outside of SED. Since $a$ and $b$ are antipodal points of the inner circle of $A$, it follows that at least a half-annulus of $A$ lies outside SED: precisely, the part of $A$ that lies on one side of $\ell'$. Since this half of $A$ lies outside SED, it is devoid of points of $S$. But this is a contradiction, because every other edge of the supporting polygon of $S$ must contain points of $S$ and, since $n\geqslant 6$, every half-annulus of $A$ contains at least two whole adjacent edges of the supporting polygon.
\end{proof}

\subsubsection{Cautious Moves for \EQ Configurations}\label{sec:analysis1b}

\begin{observation}\label{l:dist}
If $S$ is \PR and $x,y$ are companions, then $xy\leqslant xz$ for every $z\in S\setminus\{x\}$. In particular, if some $z\in S\setminus\{x\}$ is such that $xy=xz$, then $xy$ and $xz$ are adjacent edges of the supporting polygon. Moreover, if $c$ is the center of the supporting polygon, then $\angle xcz \geqslant \angle xcy$ for every $z\in S\setminus\{x\}$.
\end{observation}

\begin{lemma}\label{l:cyclic}
If $S$ is \PR, the cyclic order of $S$ around the center of SED is the same as the cyclic order of $S$ around the center of the supporting polygon.
\end{lemma}
\begin{proof}
By Lemma~\ref{l:prehull}, $S$ is in convex position, hence any two points in the convex hull of $S$ induce the same cyclic order on $S$. By Observation~\ref{l:sechull}, the center of SED lies in the convex hull of a subset of $S$, hence it lies in the convex hull of $S$. But due to Lemma~\ref{l:prehull}, the center of the supporting polygon is contained in the convex hull of $S$ as well, and the claim follows.
\end{proof}

\begin{lemma}\label{l:bound}
If $S$ is \PR, then every internal angle of the convex hull of $S$ is greater than $\pi(n-3)/n$.
\end{lemma}
\begin{proof}
Let $x,y,z,w$ be four consecutive vertices of the convex hull of $S$, such that $x$ is the companion of $y$, and $z$ is the companion of $w$. Let $ab$ be the edge of the supporting polygon containing $x$ and $y$, such that $x$ is closer to $a$. Similarly, $cd$ is the edge containing $z$ and $w$, and $z$ is closer to $c$. The infimum of $\angle xyz$ is reached (in the limit) when $y$ coincides with $b$, $w$ coincides with $d$, and $z$ tends to $w$. As the limit angle contains exactly $n-3$ edges of the supporting polygon, its size is $\pi(n-3)/n$.
\end{proof}

\begin{lemma}\label{l:quad}
Let $abcd$ be a convex quadrilateral with $ab\leqslant bc$ and $cd<da$. If $\angle adb\geqslant\angle bdc$, then $\angle abc+\angle cda\leqslant\pi$.
\end{lemma}
\begin{proof}
Let $C$ be the circumcircle of $abc$. We will prove that $d$ does not lie in the interior of $C$. This will imply that $\angle abc+\angle cda\leqslant \pi$, since $b$ and $d$ lie on opposite sides of $ac$ (because $abcd$ is convex).

Suppose by contradiction that $d$ lies in the interior of $C$. Let $\ell$ be the axis of $ac$, and let $b'$ be the intersection point between $\ell$ and the perimeter of $C$ such that $bb'$ does not intersect $ac$. Let $A$ be the circumcircle of $cb'd$. Since $d$ lies inside $C$, the center of $A$ lies between the center of $C$ and the midpoint of $b'c$. Therefore the center of $A$ lies on the same side of $\ell$ as $c$. If $B$ is the symmetric of $A$ with respect to $\ell$, the center of $B$ lies on the same side of $\ell$ as $a$. Since $cd<da$, $d$ lies on the arc of $A$ that is external to $B$. Because $A$ and $B$ have the same radius, and $ab'=b'c$, it follows that $\angle adb'<\angle b'dc$. But $ab\leqslant bc$, hence $\angle adb\leqslant\angle adb'$ and $\angle b'dc\leqslant \angle bdc$, implying that $\angle adb<\angle bdc$. This contradicts the hypothesis that $\angle adb\geqslant\angle bdc$.
\end{proof}

\begin{lemma}\label{l:equi}
If $S$ is both \PR and \EQ, then it is \RE.
\end{lemma}
\begin{proof}
Let $a\in S$ be a point on SEC, and let $b\in S$ be its companion, which, by Lemma~\ref{l:cyclic}, has angular distance $2\pi/n$ from $a$. Let $c\in S\setminus\{a\}$ be the other point of $S$ at angular distance $2\pi/n$ from $b$. If $c$ lies on SEC as well, then $ab=bc$ and, by Observation~\ref{l:dist}, $ab$ and $bc$ are adjacent edges of the supporting polygon. Because the supporting polygon is a regular $n$-gon, $\angle abc=\pi(n-2)/n$, and hence $b$ lies on SEC, too. It follows that the supporting polygon is inscribed in SEC, so all points of $S$ lie on SEC, and the configuration is \RE.

Suppose now that $c$ does not lie on SEC. If $d$ is the center of SED, then $cd<da$, and $\angle adb=\angle bdc=2\pi/n$. Also, by Observation~\ref{l:dist}, since $a$ and $b$ are companions, $ab\leqslant bc$. Therefore Lemma~\ref{l:quad} applies to $abcd$, and we get $\angle abc+\angle cda\leqslant\pi$. But $\angle cda = 4\pi/n$, implying that $\angle abc \leqslant \pi(n-4)/n < \pi(n-3)/n$. This contradicts Lemma~\ref{l:bound}.
\end{proof}

\begin{theorem}\label{t:equi}
Let $\mathcal R$ be frozen at time $t_0$, let $\mathcal R(t_0)$ be an \EQ configuration with no points in the interior of SED/3, and let the robots execute procedure {\sc Move All to SEC}. Then, the robots eventually freeze in a \RE configuration.
\end{theorem}
\begin{proof}
The procedure makes the robots move radially toward SEC, hence the configuration remains \EQ. The robots execute a cautious move with critical points only on SEC, because no \PR configuration can be formed until all the robots reach SEC, due to Lemma~\ref{l:equi}. By Lemma~\ref{l:cautious:0}, the robots eventually reach SEC, forming a \RE configuration, and they freeze as soon as the reach it.
\end{proof}

\subsubsection{Cautious Moves for \BI Configurations}\label{sec:analysis2}

\begin{lemma}\label{l:3slide1}
If some points of $S$ are allowed to ``slide'' radially in such a way that SED never changes and there are at least three consecutive points $a,b,c\in S$ (in this order) that do not slide, with $ab=bc$, then there is at most one configuration of the points that could be \PR.
\end{lemma}
\begin{proof}
If some configuration is \PR, then by Lemma~\ref{l:cyclic} either $a$ and $b$ are companions, or $b$ and $c$ are. Since $ab=bc$, by Observation~\ref{l:dist} $ab$ and $bc$ are adjacent edges of the supporting polygon, and therefore the whole supporting polygon is fixed, no matter how the points slide. Then, there is only one possible position in which each sliding point may lie on the supporting polygon, due to Lemma~\ref{l:coradial}. Hence, if a \PR configuration is formable, it is unique.
\end{proof}

\begin{observation}\label{o:3sides}
For every $n\geqslant 3$, if three straight lines are given in the plane, there is at most one regular $n$-gon with three edges lying on the three lines.
\end{observation}

\begin{lemma}\label{l:3slide2}
If some points of $S$ are allowed to ``slide'' radially in such a way that SED never changes, and there are at least three consecutive points $a,b,c\in S$ (in this order) that do not slide, plus at least another non-sliding point $d$, not adjacent to $a$ nor $c$, then there is at most one configuration of the points that could be \PR.
\end{lemma}
\begin{proof}
If some configuration is \PR, then by Lemma~\ref{l:cyclic} either $a$ and $b$ are companions, or $b$ and $c$ are. If $ab=bc$, Lemma~\ref{l:3slide1} applies. Otherwise, without loss of generality, assume that $ab<bc$, and therefore $a$ and $b$ are companions, due to Observation~\ref{l:dist}. Then all the companionships are fixed, again by Lemma~\ref{l:cyclic}. The slope of the edge of the supporting polygon through $a$ and $b$ is fixed, hence all the slopes of the other edges are fixed, because the supporting polygon is regular. In particular, the slopes of the edges through $c$ and $d$ are fixed, and these are two distinct edges because $c$ and $d$ are not adjacent. Therefore, by Observation~\ref{o:3sides}, the whole supporting polygon is fixed. It follows that there is at most one position of the sliding points that could be \PR, due to Lemma~\ref{l:coradial}.
\end{proof}

For the rest of this section, we will assume that $S$ is not a \CO set. Recall that, in a \BI configuration, two points at angular distance $\mu_0$ are called neighbors, and two points at angular distance $\mu_1$ are called quasi-neighbors.

\begin{lemma}\label{l:companion}
If $S$ is both \BI and \PR, then two points are neighbors if and only if they are companions.
\end{lemma}
\begin{proof}
Let $a\in S$ be a point on SEC, let $b\in S$ be the point at angular distance $\mu_1$ from $a$, and let $c\in S$ be the point at angular distance $\mu_0$ from $b$. If $d$ is the center of SED, it follows that $\angle adb>\angle bdc$. By Lemma~\ref{l:cyclic}, the companion of $b$ is either $a$ or $c$. Assuming by contradiction that $b$'s companion is $a$, Observation~\ref{l:dist} implies that $ab\leqslant bc$. Hence $c$ does not lie on SEC, otherwise $ab>bc$ (recall that $a$ lies on SEC, as well). It follows that $cd<da$, and Lemma~\ref{l:quad} applies to $abcd$, yielding $\angle abc+\angle cda\leqslant\pi$. But, since $S$ is \BI, $\angle cda = \mu_0+\mu_1 = 4\pi/n$, implying that $\angle abc \leqslant \pi(n-4)/n < \pi(n-3)/n$, which contradicts Lemma~\ref{l:bound}.
\end{proof}

\begin{lemma}\label{l:bsec1}
If $S$ is both \BI and \PR, and two companions lie on SEC, then every point of $S$ lies on SEC.
\end{lemma}
\begin{proof}
Let $a,a'\in S$ be two companion points that lie on SEC, which are also neighbors by Lemma~\ref{l:companion}. Then $a$ and $a'$ are not antipodal, and therefore by Observation~\ref{l:sechull} there must be another point $b\in S$ on SEC which, without loss of generality, we may assume to be strongly analogous to $a$. Let $p$ be the center of SED, and let $b'$ be the neighbor of $b$, which is also its companion. Because the configuration is \BI and the supporting polygon must be regular, it follows that the slope of the line $bb'$ is equal to the slope of $aa'$ increased or decreased by $\angle apb$. Hence also $b'$ lies on SEC.

If the edges of the supporting polygon on which $a$ and $b$ lie are not opposite, then it is easy to see that no two points among $a$, $a'$, $b$, $b'$ are antipodal (otherwise $S$ would be \EQ), and they belong to the same half of SEC. By Observation~\ref{l:sechull}, there must be another point $c\in S$ on SEC. By the same reasoning, the companion of $c$ also belongs to SEC. Hence three lines containing edges of the supporting polygon are given, which means that the whole polygon is fixed (by Observation~\ref{o:3sides}), and therefore all the points of $S$ lie on SEC.

Otherwise, if the edges of the supporting polygon on which $a$ and $b$ lie are opposite, the slopes of all other edges are fixed, and the size of the supporting polygon is also fixed. If the center of the polygon is not $p$, then some points of $S$ must lie outside SED. Hence the center of the supporting polygon is $p$, and all the points of $S$ lie on SEC.
\end{proof}

\begin{lemma}\label{t:bsec2}
If $S$ is both \BI and \PR, and there are two points on SEC that are not strongly analogous, then every point of $S$ lies on SEC.
\end{lemma}
\begin{proof}
If two points on SEC are neighbors, by Lemma~\ref{l:companion} they are also companions, and then Lemma~\ref{l:bsec1} applies. Otherwise, if no two neighbors lie on SEC, by assumption there exist two non-neighboring points $a,b\in S$ that are not strongly analogous and lie on SEC (and belong to different edges of the supporting polygon, by Lemma~\ref{l:companion}). Let $p$ be the center of SED. Then, since the supporting polygon is regular, the slope of the edge through $b$ equals the slope of the edge through $a$ plus or minus $\angle apb$. As a consequence, if the companion of $a$ lay in the interior of SED, then the companion of $b$ would lie outside, which would be a contradiction. Therefore, the companion of $a$ lies on SEC as well, and Lemma~\ref{l:bsec1} applies.
\end{proof}

\begin{lemma}\label{t:bsec3}
Let $S$ be \BI, and suppose that all the points of $S$ that lie on SEC are strongly analogous. If the points of $S$ that are strongly analogous to those on SEC are allowed to ``slide'' radially toward SEC (while the other points of $S$ do not move), then there is at most one configuration of the points that could be \PR.
\end{lemma}
\begin{proof}
By assumption, at least $n/2$ strongly analogous points do not slide, hence no two adjacent points are allowed to slide. Moreover, there is a point $a\in S$ already on SEC that does not slide and, by assumption, neither of its adjacent points is allowed to slide, because they are not strongly analogous to $a$. Hence Lemma~\ref{l:3slide2} applies.
\end{proof}

\begin{theorem}\label{t:bsec4}
Let $\mathcal R$ be frozen at time $t_0$, let $\mathcal R(t_0)$ be a \BI (and not \CO) configuration with no points in the interior of SED/3, let $n>4$, and let the robots execute procedure {\sc Move All to SEC} with suitable critical points. Then, the robots eventually freeze in a \PR configuration.
\end{theorem}
\begin{proof}
If $\mathcal R(t_0)$ is already a \PR configuration, there is nothing to prove, because the swarm is already frozen at time $t_0$. If two points that are not strongly analogous lie on SEC at time $t_0$, then no \PR configuration can be formed, unless all robots lie on SEC, due to Lemma~\ref{t:bsec2}. Hence, in this case, no critical points are needed. On the other hand, if all the robots that lie on SEC at time $t_0$ belong to the same strong analogy class, procedure {\sc Move All to SEC} makes the robots of the same strong analogy class move first toward SEC. By Lemma~\ref{t:bsec3}, during this phase at most one configuration $C$ could be \PR. Therefore, we may take $C$ as a set of critical points for the cautious move. Note that this set does not change as the robots perform the cautious move. By Corollary~\ref{c:cautious4}, the robots freeze in configuration $C$, provided that they reach it. If they do not reach it, then by Lemma~\ref{l:cautious:0} they eventually reach SEC and freeze.

Assume now that all the robots of one strong analogy class are on SEC, forming a \RE set of $n/2$ points. Let $P$ be the regular $n$-gon inscribed in SEC that has these $n/2$ points among its vertices. Procedure {\sc Move All to SEC} makes the robots of the other strong analogy class move toward SEC, and the possible \PR configurations in which the robots can be found are precisely those in which none of the robots lies strictly in the interior of the area enclosed by $P$, and every two strongly analogous robots are equidistant from the center of SED.

If all the robots at time $t_0$ lie in the interior or on the boundary of $P$, then we let $C$ be the configuration obtained from $\mathcal R(t_0)$ by sliding all the robots radially away from the center of SED, until they reach the boundary of $P$. In this case, $C$ will be the input critical point set of the cautious move. Otherwise, let $d$ be the maximum distance of an internal point of $\mathcal R(t_0)$ from the center of SED. Let $C'$ be the configuration obtained from $\mathcal R(t_0)$ by sliding the internal robots radially away from the center of SED, until they reach distance $d$ from it. In this case, $C'$ will be the input critical point set of the cautious move. In both cases, the cautious move will make the swarm freeze in configuration $C$, which is the first \PR configuration formable.
\end{proof}

\subsubsection{Cautious Moves for \DBI Configurations}\label{sec:analysis2b}

\begin{lemma}\label{l:p2}
If $S$ is \DBI and not \CO, and the points of one analogy class stay still on SEC, while the other points are allowed to ``slide'' radially within SED, then at most one configuration of the points can be a \PR in which sliding points are not companions.
\end{lemma}
\begin{proof}
Let $p_0\in S$ be a point belonging to the analogy class that stays still on SEC, and let $p_i\in S$ be the $(i+1)$-th point in the cyclic order around the center of SED, $c$. We may assume that $p_1$ is analogous to $p_0$, and therefore that the clockwise angle sequence induced by $p_0$ is of the form $(\alpha, \beta, \gamma, \beta, \alpha, \beta, \gamma, \beta, \alpha, \beta, \gamma, \beta, \cdots)$. It follows that the points analogous to $p_0$ are those of the form $p_{4i}$ and $p_{4i+1}$.

Suppose that $S$ reaches a \PR configuration in which no two sliding points are companions. Hence every other edge of the supporting polygon contains a point of $S$ of the analogy class that stays still on SEC (cf.~Lemma~\ref{l:cyclic}). Let $q_{2i}$ (respectively, $q_{2i+1}$) be the point at which the extensions of the edges containing $p_{4i}$ and $p_{4i+1}$ (respectively, $p_{4i+1}$ and $p_{4i+4}$) meet, where indices are taken modulo $n$. Since the supporting polygon is regular, then clearly the $q_i$'s form a \RE configuration with $n/2$ elements, and in particular $q_0q_1=q_1q_2$ and $\angle p_0q_0p_1=\angle p_1q_1p_4=\angle p_4q_2p_5=\pi(n-4)/n$. On the other hand, the analogy class of $p_0$ is a \BI or \EQ set of size $n/2$ lying on SEC, hence it forms a polygon with equal internal angles, and in particular $\angle p_0p_1p_4=\angle p_1p_4p_5=\pi(n-4)/n$.

Let $\theta=\angle p_1p_0q_0$ and $\theta'=\angle q_0p_1p_0$. Then
$$\pi-\theta-\theta'=\angle p_0q_0p_1=\pi(n-4)/n=\angle p_0p_1p_4=\pi-\theta'-\angle p_4p_1q_1,$$
implying that $\angle p_4p_1q_1=\theta$. Similarly $\angle p_5p_4q_2=\theta$, and therefore $p_0p_1q_0$ and $p_1p_4q_1$ are similar triangles, and $p_0p_1q_0$ and $p_4p_5q_2$ are congruent (because $p_0p_1=p_4p_5$).

We have $q_0p_1+p_1q_1=q_0q_1=q_1q_2=q_1p_4+p_4q_2$. Also, $p_0q_0/p_1q_1=q_0p_1/q_1p_4$ and $p_0q_0=p_4q_2$. Hence we may substitute $q_1p_4$ with $q_0p_1\cdot p_1q_1/p_0q_0$ and $p_4q_2$ with $p_0q_0$, obtaining $$q_0p_1+p_1q_1 = \frac{q_0p_1\cdot p_1q_1}{p_0q_0}+p_0q_0.$$ After rearranging terms and factoring, we get $$(p_0q_0-p_1q_1)(p_0q_0-q_0p_1)=0,$$ which implies that either $p_0q_0=p_1q_1$ or $p_0q_0=q_0p_1$.

Assume first that $p_0q_0=p_1q_1$ and $p_0q_0\neq q_0p_1$. This implies that $\alpha=2\beta+\gamma=4\pi/n$ and therefore, by observing that the sum of the internal angles of the quadrilateral $cp_1q_1p_4$ is $2\pi$, we have $\angle cp_1q_1=\pi-\angle q_1p_4c$. This means that the segment $p_1q_1$ has some points in the interior of SED if and only if $q_1p_4$ has none. However, $p_2$ is the companion of $p_1$ and hence it lies on $p_1q_1$, and $p_3$ is the companion of $p_4$ and hence it lies on $q_1p_4$, which yields a contradiction. It follows that in this case no \PR configuration is formable.

Assume now that $p_0q_0=q_0p_1$, hence $\angle cp_0q_0=\angle q_0p_1c=\pi(n+4)/2n-\alpha/2$. Therefore the slopes of the two edges of the supporting polygon to which $p_0$ and $p_1$ belong are fixed. This also fixes the slope of the edge of the supporting polygon through $p_4$, and hence the whole supporting polygon is fixed, by Observation~\ref{o:3sides}. Due to Lemma~\ref{l:coradial}, the trajectory of each sliding point intersects the supporting polygon in at most one point, and therefore in this case at most one \PR configuration can be formed.
\end{proof}

\begin{lemma}\label{l:p3}
Let $\mathcal R$ be frozen at time $t_0$, let $\mathcal R(t_0)$ be a \DBI (and not \CO) configuration with no points in the interior of SED/3, and let the robots in $\mathcal A\subset\mathcal R$, forming one analogy class of $\mathcal R(t_0)$, stay still on SEC, while the robots in $\mathcal A'=\mathcal R\setminus\mathcal A$ execute procedure {\sc Move All to SEC} or procedure {\sc Move Walkers to SEC/3} with suitable critical points. Then, if a \PR configuration in which analogous robots are companions is ever formed, the robots freeze as soon as they form one.
\end{lemma}
\begin{proof}
Suppose first that $n\geqslant 12$. If $\mathcal R(t)$ is \PR at some time $t\geqslant t_0$, there are at least three pairs of companions that stay still on SEC (cf.\ Lemma~\ref{l:cyclic}). These three pairs determine the slopes of three edges of the supporting polygon, which, due to Observation~\ref{o:3sides}, is fixed. By Lemma~\ref{l:coradial}, the trajectory of each robot intersects the supporting polygon in at most one point, and hence there is at most one formable \PR configuration, which can be chosen as a set of critical points for the cautious move, due to Theorem~\ref{t:cautious1}.

Let $n<12$, and hence $n=8$. Let $\mathcal R=\{a,b,c,d,e,f,g,h\}$, where $\mathcal A=\{c,d,g,h\}$ is the set of robots that stay still on SEC. We seek to characterize the formable \PR configurations in which $c$ and $d$ are companions. Let $\ell$ be the line through $c$ and $d$, let $\ell'$ be the line through $g$ and $h$, and let $\lambda$ be the distance between $\ell$ and $\ell'$. Then, the two edges of the supporting polygon to which $a$ and $b$ belong must be orthogonal to both $\ell$ and $\ell'$, and similarly for the edge to which $e$ and $f$ belong. Moreover, the distance between these two edges must be $\lambda$. Let $x$ be the center of SED$(S)$, and let $a'$ (respectively, $b'$, $e'$, $f'$) be the point on SEC$(S)$ that is co-radial with $a$ (respectively, $b$, $e$, $f$). It is easy to see that the positions of $a$ that could give rise to a \PR configuration belong to a (possibly empty) closed segment $A$, which is a subset of the segment $a'x$. Similarly, the positions of $b$, $e$, and $f$ that could give rise to \PR configurations belong to closed segments $B$, $E$, and $F$, which, together with $A$, form a set that is mirror symmetric and centrally symmetric with respect to $x$. If $A$ is empty, then no \PR configuration in which moving robots are companions can be formed. Therefore, let us assume that $A$ is not empty.

Assume now that $a$, $b$, $e$, and $f$ move toward SEC executing procedure {\sc Move All to SEC}. The case in which they execute procedure {\sc Move Walkers to SEC/3} is symmetric, and therefore it is omitted. Let $a''$ and $a'''$ be the endpoints of $A$, with $a''$ closest to $a'$, and let $a^*$ be the midpoint of $A$. Similar names are given to the endpoints and midpoints of $B$, $E$, and $F$. Note that, by construction, $\{a'',b'',c(t),d(t),e''',f''',g(t),h(t)\}$, $\{a''',b''',c(t),d(t),e'',f'',g(t),h(t)\}$, and $\{a^*,b^*,c(t),d(t),e^*,f^*,g(t),h(t)\}$ are \PR sets at any time $t\geqslant t_0$.

Without loss of generality, let $a(t_0)$ be such that the segment $a(t_0)a'$ is not longer than $b(t_0)b'$, $e(t_0)e'$, and $f(t_0)f'$. If $a(t_0)$ belongs to the segment $a''a'$, open at $a''$ and closed at $a'$, then no \PR configuration can be formed, regardless of how the robots move toward SEC. Hence in this case no critical points are needed. If $a(t_0)$ belongs to the (closed) segment $xa^*$, then we take $\{a^*,b^*,c(t),d(t),e^*,f^*,g(t),h(t)\}$ as a set of critical points at any time $t\geqslant t_0$. Since $b(t_0)\in xb^*$, $e(t_0)\in xe^*$, and $f(t_0)\in xf^*$, procedure {\sc Cautious Move} will make $a$, $b$, $e$, and $f$ stop at $a^*$, $b^*$, $e^*$, and $f^*$, respectively, and wait for each other. When all of them have reached such critical points, a \PR configuration is reached, and the swarm is frozen. Also, this is the first \PR configuration that is reached by the robots.

Finally, let $a(t_0)$ belong to the segment $a^*a''$, open at $a^*$ and closed at $a''$. Let $b_1$ and $b_2$ be the two points on $xb'$ whose distance from $b^*$ is the same as the distance between $a(t_0)$ and $a^*$, with $b_1$ closest to $x$. Similarly, we define $e_1$ and $e_2$ on $xe'$, and $f_1$ and $f_2$ on $xf'$. Then, the set $\{a(t),b_2,c(t),d(t),e_1,f_1,g(t),h(t)\}$ is \PR at any time $t\geqslant t_0$, and we may take it as a set of critical points. If $e(t_0)$ is past $e_1$, or $f(t_0)$ is past $f_1$, then no \PR set can be formed, regardless of how the robots move. Otherwise, procedure {\sc Cautious Move} will make $e$ and $f$ reach $e_1$ and $f_1$, stop there, and wait for each other (note that the position of $a$ does not change while this happens, hence $a(t)=a(t_0)$).

If $b(t_0)=b_2$, then a \PR configuration is reached for the first time, and none of the robots is moving. Otherwise, suppose that $b(t_0)$ is in the (closed) segment $xb_1$. Then, eventually, $b$ will stop in $b_1$ while $e$ and $f$ are in $e_1$ and $f_1$. Note that $e$ and $f$ acquire $e_2$ and $f_2$ as transposed critical points (because $b_2$ is a critical point of $b$), and also $e^*$ and $f^*$ as intermediate critical points (because they are the midpoints of $e_1e_2$ and $f_1f_2$). Similarly, $b$ acquires $b^*$ as a new critical point. When all three of them have moved once, they will be found somewhere in the \emph{open} segments $b_1b_2$, $e_1e_2$, and $f_1f_2$. While they reach this configuration, no \PR configuration is ever formed. Moreover, no \PR configuration can be formed afterwards. Finally, let $b(t_0)$ be in the open segment $b_1b_2$. Then, $b$ will stay still and wait for $e$ and $f$, which will eventually move and stop somewhere in the \emph{open} segments $e_1e_2$ and $f_1f_2$. As in the previous case, no \PR configuration can ever be reached.
\end{proof}

\begin{theorem}\label{t:dbi}
Let $\mathcal R$ be frozen at time $t_0$, let $\mathcal R(t_0)$ be a \DBI (and not \CO) configuration with no points in the interior of SED/3, and let the robots execute procedure {\sc Move All to SEC} or procedure {\sc Move Walkers to SEC/3} with suitable critical points. Then, if a \PR configuration is ever formed, the robots freeze as soon as they form one.
\end{theorem}
\begin{proof}
Recall that in a \DBI set there are exactly two analogy classes of equal size. According to both procedures, only one analogy class of robots is allowed to move at each time. Indeed, even procedure  {\sc Move All to SEC} lets the second class move only when the first class has completely reached SEC, and therefore no robot in that class is moving. Let $\mathcal A\subset \mathcal R$ be the analogy class that is allowed to move at a given time, and let $\mathcal A'$ be the other class.

Suppose first that not all the robots of $\mathcal A'$ are on SEC. This means that the procedure being executed is {\sc Move All to SEC}, because procedure {\sc Move Walkers to SEC/3} assumes the robots of $\mathcal A'$ to be all on SEC (recall that the walkers are all analogous, due to Observation~\ref{o:walkers}). But procedure {\sc Move All to SEC} allows the robots of $\mathcal A$ to move only if some of them are already on SEC (by Observation~\ref{l:sechull}, some robots must indeed be on SEC). Because all the robots of $\mathcal A'$ stay still, and at least one robot of $\mathcal A$ stays still because it is already on SEC, this implies the presence of three consecutive robots that do not move, and enables the application of Lemma~\ref{l:3slide2}. Hence at most one \PR configuration is formable, which can be taken as a set of critical points, due to Theorem~\ref{t:cautious1}.

Suppose now that all the robots of $\mathcal A'$ are on SEC. By Lemma~\ref{l:p2}, at most one \PR configuration $C_1$ is formable in which no two robots in the same analogy class are companions. Theorem~\ref{t:cautious1} guarantees that the cautious move with critical point set $C_1$ enjoys property $\mathcal P_1$ that the robots freeze as soon as they reach configuration $C_1$. On the other hand, by Lemma~\ref{l:p3}, there exists a set of critical points $C_2$ ensuring property $\mathcal P_2$ that the robots will freeze as soon as they reach a \PR configuration in which robots in the same analogy class are companions. Hence, due to Theorem~\ref{t:cautious3}, the cautious move with critical point set $C_1\cup C_2$ enjoys both properties $\mathcal P_1$ and $\mathcal P_2$, and therefore it correctly handles all formable \PR configurations.
\end{proof}

\subsubsection{Cautious Moves for \PER Configurations}\label{sec:analysis3}

If $S$ is not \CO and $n$ is even, we will say that two points of $S$ have \emph{the same parity} (respectively, \emph{opposite parity}) if there are an odd (respectively, even) number of other points between them in the cyclic order around the center of SED.

\begin{lemma}\label{l:4slide}
If some points of $S$ are allowed to ``slide'' radially in such a way that SED never changes, and there are at least four points $a,b,c,d\in S$ that do not slide, appearing in this order around the center of SEC, such that $a$ and $b$ are consecutive, $c$ and $d$ are consecutive, and $b$ and $c$ have the same parity, then there are at most two configurations of the points that could be \PR.
\end{lemma}
\begin{proof}
If some configuration is \PR, then by Lemma~\ref{l:cyclic} either $a$ and $b$ are companions and $c$ and $d$ are not, or vice versa. Assume that $a$ and $b$ are companions, and hence the line containing the edge of the supporting polygon through them is fixed. Then the slopes of the two edges through $c$ and $d$ are fixed as well, and this determines a unique supporting polygon, by Observation~\ref{o:3sides}. In turn, this may give rise to at most one possible \PR configuration, by Lemma~\ref{l:coradial}. Otherwise, if $c$ and $d$ are companions, by a symmetric argument at most one other \PR configuration is formable.
\end{proof}

\begin{lemma}\label{l:samecenter}
Suppose that $S$ is \PR, and there is a concordance class $C\subset S$ that lies on SEC$(S)$ and forms a \RE configuration. If the size of $C$ is even and greater than $2$, then the center of the supporting polygon of $S$ coincides with the center of SED$(S)$.
\end{lemma}
\begin{proof}
Because $C$ is a \RE set of even size, there exist two antipodal points points $a, a'\in C$, both lying on SEC$(S)$. Since $S$ is a \PR set and $C$ is a concordance class, Lemma~\ref{l:cyclic} implies that $a$ and $a'$ belong to opposite and parallel edges $\ell$ and $\ell'$ of the supporting polygon. Therefore, the center of the supporting polygon belongs to the line parallel to $\ell$ and $\ell'$ that is equidistant to them. Let $r$ be this line. Since $a$ and $a'$ are antipodal points, it follows that $r$ passes through the center of SED.

Because $C$ has at least four elements, there exist two antipodal points $b, b'\in C$, distinct from $a$ and $a'$. By the same reasoning, the center of the supporting polygon belongs to a line $r'$ that is parallel to the edges of the supporting polygon through $b$ and $b'$. Also $r'$ passes through the center of SED and, since $r$ and $r'$ are not parallel and they are incident at the center of SED, it follows that the center of the supporting polygon coincides with the center of SED.
\end{proof}

\begin{observation}\label{o:p1}
If $S$ is \BP with period $3$ and not \CO, then it has exactly two analogy classes: one \EQ with $n/3$ elements and the other \BI with $2n/3$ elements (where angles are always measured with respect to the center of SED$(S)$).
\end{observation}

\begin{lemma}\label{l:p1}
If $S$ is \BP with period $3$ and not \CO, and the points of the analogy class of size $n/3$ are on SEC$(S)$, then $S$ is not \PR.
\end{lemma}
\begin{proof}
If, by contradiction, $S$ is a \PR set, then $n$ must be even, and hence it must be a multiple of $6$.

Suppose that $n=6$. Let $S=\{a,b,c,d,e,f\}$, where the points appear in this order around the center of SED. Without loss of generality, the clockwise angle sequence induced by $a$ is $\{\alpha, \alpha, \beta, \alpha, \alpha, \beta\}$, with $\alpha\neq \beta$. Assume by contradiction that $S$ is \PR, let $ABCDEF$ be the supporting polygon, such that $a$ and $b$ lie on the edge $AB$. Let $x$ be the center of SED$(S)$ and let $X$ be the center of the supporting polygon. Note that $e$ and $f$ must belong to the edge $EF$ (by definition of \PR), and $x$ lies on the segment $be$ (because $be$ is an axis of symmetry of $S$). Therefore, $x$ and $A$ must lie on the same side of the line through $B$ and $E$. Suppose that $\alpha<60^\circ<\beta$. Observe that $c$ and $d$ lie on $CD$ and $\angle cxd>60^\circ$, implying that $x$ lies strictly inside the circle through $X$, $C$, and $D$. However, this circle and $A$ lie on the opposite side of the line though $B$ and $E$, which yields a contradiction. Assume now that $\alpha>60^\circ>\beta$. Since $\angle axb>60^\circ$ and $a$ and $b$ belong to $AB$, $x$ must lie strictly inside the circle through $X$, $A$, and $B$. Similarly, $x$ must lie strictly inside the circle through $X$, $E$, and $F$. But then $x$ also lies strictly inside the circle through $X$, $F$, and $A$, which contradicts the fact that $a\in AB$, $f\in EF$, and $\angle axf<60^\circ$.

Suppose now that $n\geqslant 12$. Then, the analogy class of size $n/3\geqslant 4$ is a \RE set of an even number of points located on SEC, forming a concordance class. Hence, by Lemma~\ref{l:samecenter}, the center of SEC coincides with the center of the supporting polygon. One of the angle sequences of $S$ is of the form $(\alpha,\alpha,\beta,\alpha,\alpha,\beta,\alpha,\alpha,\beta,\cdots)$, with $\alpha\neq\beta$. Let $C$ be the analogy class of size $2n/3$. Observe that, because the period of $S$ is odd, at least two points of $C$ must be companions, due to Lemma~\ref{l:cyclic}. Hence, Observation~\ref{l:dist} implies that $\alpha>\beta$, because the center of the supporting polygon is the center of SED. It follows that the companion of each point of $C$ must be another point of $C$, which contradicts the fact that the period is odd.
\end{proof}

\begin{observation}\label{o:pp1}
If $S$ is \BP with period $4$, $S$ is not \CO, and no analogy class contains consecutive points, then $S$ has exactly three analogy classes: two \EQ with $n/4$ elements each, and the other \BI with $n/2$ elements (where angles are always measured with respect to the center of SED$(S)$). Moreover, the two analogy classes of size $n/4$ collectively form a \RE set of size $n/2$ that is also a concordance class of $S$.
\end{observation}

\begin{lemma}\label{l:pp1}
Suppose that $S$ is \BP with period $4$, not \CO, and that no analogy class contains consecutive points. If the points of both analogy classes of size $n/4$ are on SEC$(S)$, then $S$ is not \PR.
\end{lemma}
\begin{proof}
Suppose by contradiction that $S$ is a \PR set. By Observation~\ref{o:pp1}, the points of the two analogy classes of size $n/4$ collectively form a \RE set of an even number of points located on SEC, forming a concordance class. Such a set has size $n/2\geqslant 4$, because the period of $S$ is $4$, and hence $n\geqslant 8$. Therefore Lemma~\ref{l:samecenter} applies, and the center of SEC coincides with the center of the supporting polygon.

One of the angle sequences of $S$ is of the form $(\alpha,\alpha,\beta,\beta,\alpha,\alpha,\beta,\beta,\cdots)$, with $\alpha<\beta$. Then, due to Lemma~\ref{l:cyclic}, Observation~\ref{l:dist}, and the fact that the center of the supporting polygon is the center of SED, it follows that each point of $S$ has a companion at angular distance $\alpha$. But causes all the points of one analogy class of size $n/4$ to have two companions, and leaves the points of the other analogy class of size $n/4$ with no companions, which contradicts the fact that each point of $S$ must have exactly one companion.
\end{proof}

\begin{theorem}\label{t:p1}
Let $\mathcal R$ be frozen at time $t_0$, let $\mathcal R(t_0)$ be a \PER (and not \CO) configuration with no points in the interior of SED/3, and let the robots execute procedure {\sc Move All to SEC} or procedure {\sc Move Walkers to SEC/3} with suitable critical points. Then, if a \PR configuration is ever formed, the robots freeze as soon as they form one.
\end{theorem}
\begin{proof}
By definition of \PER, the period of $\mathcal R(t_0)$ is $k>2$, with $n\geqslant 2k>4$. Recall that, in both procedures, only analogous robots are allowed to move. In particular, in procedure {\sc Move All to SEC}, a new analogy class starts moving only when the robots of the previous analogy class have reached SEC and have stopped. As a consequence, at any time, only one analogy class $\mathcal A\subset \mathcal R$ of robots is moving. Let $r_0\in \mathcal A$, and let $r_i\in \mathcal R$ be the $(i+1)$-th robot in the cyclic order around the center of SED. By definition of analogy class, in every set of $k$ consecutive robots (in their cyclic order around the center of SED), at most two of them belong to $\mathcal A$.

Suppose first that the size of $\mathcal A$ is $n/k$. If $k>3$, then $r_1$, $r_2$, $r_3$, and $r_{k+1}$ do not move, and they satisfy the hypotheses of Lemma~\ref{l:3slide2}, implying that no \PR configuration can be formed. If $k=3$, then Lemma~\ref{l:4slide} applies to $r_1$, $r_2$, $r_4$, and $r_5$, and at most one configuration $C$ can be \PR. By Theorem~\ref{t:cautious1}, taking $C$ as a set of critical points suffices.

Otherwise the size of $\mathcal A$ is $2n/k$, and therefore the configuration is \BP. Observe that, according to both procedure {\sc Move All to SEC} and procedure {\sc Move Walkers to SEC/3}, if an analogy class of size $2n/k$ is allowed to move, it means that all classes of size $n/k$ are located on SEC (recall that the walkers are all analogous, due to Observation~\ref{o:walkers}). Let $r_a$ be a moving robot such that $0<a<k$. Without loss of generality, we may assume that $a\leqslant\lfloor k/2\rfloor$. There are several cases to consider.
\begin{itemize}
\item Let $k=3$. By Observation~\ref{o:p1} there are only two analogy classes, with $n/3$ and $2n/3$ robots, respectively. Since the analogy class of size $n/3$ is on SEC, Lemma~\ref{l:p1} applies, and no \PR configuration can be formed.
\item Let $k=4$ and $a=1$. Then the configuration is \DBI, and Theorem~\ref{t:dbi} applies.
\item Let $k=4$ and $a=2$. Then, no analogy class contains consecutive points in the cyclic order around the center of SED. Since the analogy classes of size $n/4$ are on SEC, Lemma~\ref{l:pp1} applies, and no \PR configuration can be formed.
\item Let $k=5$. Then Lemma~\ref{l:4slide} applies, because $r_3$, $r_4$, $r_8$, and $r_9$ do not move. Hence no \PR configuration can be formed.
\item Let $k=6$ and $a=1$ or $a=2$. Then Lemma~\ref{l:3slide2} applies, because $r_3$, $r_4$, $r_5$, and $r_{11}$ do not move. Hence no \PR configuration can be formed.
\item Let $k=6$ and $a=3$. Then Lemma~\ref{l:4slide} applies, because $r_1$, $r_2$, $r_4$, and $r_5$ do not move. Hence at most two configurations can be \PR. By Corollary~\ref{c:cautious4}, taking the union of these configurations as critical points suffices.
\item Let $k>6$. Then Lemma~\ref{l:3slide2} applies, because $r_{k-3}$, $r_{k-2}$, $r_{k-1}$, and $r_{2k-1}$ do not move. Hence no \PR configuration can be formed.
\end{itemize}
\end{proof}

\subsubsection{Cautious Moves for \AP Configurations}\label{sec:analysis4}

\begin{lemma}\label{l:homot}
Let $S$ be not \CO with $|S|=6$, and suppose that two consecutive points $a,b\in S$ are allowed to ``slide'' radially without causing SED to change, while the other points of $S$ stay still. Let $L$ be the locus of positions of $a$ (within its radius of SED) for which there is a position of $b$ (within its radius of SED) giving rise to a \PR configuration in which $a$ and $b$ are companions. Then $L$ is either the empty set or a topologically closed line segment (contained in $a$'s radius of SED).
\end{lemma}
\begin{proof}
Let $a$, $b$, $c$, $d$, $e$, $f$ be the points of $S$, appearing around the center of SED in this order. Since we want $a$ and $b$ to be be companions, and since the order of the points of $S$ around the center of SED is preserved as $a$ and $b$ move radially, by Lemma~\ref{l:cyclic} $c$ and $d$ have to be companions, as well as $e$ and $f$.

For a \PR configuration to be formed, the lines $cd$ and $ef$ must meet at some point $p$ (at an angle of $60^\circ$), and the supporting polygons of such \PR configurations must all be contained in the angle $\angle dpe$. More precisely, such supporting polygons are regular hexagons with two non-adjacent edges lying on the lines $cd$ and $ef$, and having homothetic center $p$. Since all such hexagons are homothetic, their vertices must lie on four distinct lines through $p$: two such lines are $cd$ and $ef$ (and they contain four vertices of each of the hexagons), and let $\ell$ and $\ell'$ be the other two lines (each of which contains one vertex of each of the hexagons).

Among these ``candidate'' supporting polygons, we discard the ones that do not contain all of $c$, $d$, $e$, and $f$. What is left is a ``closed interval'' $\mathcal L$ of supporting polygons: the smallest one has either $d$ or $e$ as a vertex (whichever is closest to $p$) and the largest one has either $c$ or $f$ as a vertex (whichever is farthest from $p$).

Of course, we must also discard the ``candidate'' supporting polygons that cannot contain both $a$ and $b$ on the edge opposite to $p$ (even if $a$ and $b$ slide radially), to which we refer as the \emph{far edge}. The slope of the far edge is fixed (it is perpendicular to the bisector of $\ell$ and $\ell'$), and its endpoints must lie on $\ell$ and $\ell'$. Let $\rho_a$ (respectively, $\rho_b$) be the radius of SED on which $a$ (respectively, $b$) is allowed to slide. Determining the far edges that can contain both $a$ and $b$ boils down to determining the intersections among $\ell$, $\ell'$, $\rho_a$, and $\rho_b$, and comparing the distances from $p$ of such intersections with those of the endpoints of $\rho_a$ and $\rho_b$. Since the elements involved are straight lines and line segments, this leaves us with a ``closed interval'' $\mathcal L'$ of eligible supporting polygons.

Intersecting $\mathcal L$ and $\mathcal L'$, we obtain a (possibly empty) ``closed interval'' of supporting polygons, each of which effectively determines a \PR configuration obtained by suitably sliding $a$ and $b$. It follows that $L$ must also be a closed interval of $\rho_a$ (or the empty set), because $\rho_a$ is a straight line segment.
\end{proof}

\begin{lemma}\label{l:ap1}
Let $\mathcal R$ be frozen at time $t_0$, let $\mathcal R(t_0)$ be an \AP (and not \CO) configuration with no points in the interior of SED/3, let $n=6$, and let the robots execute procedure {\sc Move All to SEC} or procedure {\sc Move Walkers to SEC/3} with suitable critical points. Then, if a \PR configuration is ever formed, the robots freeze as soon as they form one.
\end{lemma}
\begin{proof}
Let $\mathcal A\subset \mathcal R$ be the analogy class of robots that is allowed to move initially. As the robots of $\mathcal A$ are required to reach their destination and stop before any other class can possibly move, it is sufficient to prove the lemma just for $\mathcal A$. Recall that, in an \AP configuration, the analogy classes have size either one or two. If $|\mathcal A|=1$, then Lemma~\ref{l:3slide2} applies, and at most one \PR configuration $C$ can be formed. Taking $C$ as a set of critical points suffices, due to Theorem~\ref{t:cautious1}.

Suppose now that $|\mathcal A|=2$, let $r_0\in \mathcal A$, and let $r_i$ be the $i$-th robot after $r_0$ in clockwise order around the center of SEC, with $1\leqslant i\leqslant 5$. Without loss of generality, either $r_1\in \mathcal A$ or $r_2\in \mathcal A$ or $r_3\in \mathcal A$. If $r_2\in \mathcal A$, then at most one \PR configuration $C$ is formable, due to Lemma~\ref{l:3slide2}. $C$ can be taken as a set of critical points, due to Theorem~\ref{t:cautious1}. On the other hand, if $r_3\in \mathcal A$, Lemma~\ref{l:4slide} applies, and at most two \PR configurations $C_1$ and $C_2$ can be formed. Therefore, by Corollary~\ref{c:cautious4}, taking $C_1\cup C_2$ as a set of critical points suffices.

Finally, assume that $\mathcal A=\{r_0,r_1\}$. If $r_0$ and $r_1$ are not companions, then by Lemma~\ref{l:cyclic} $r_3$ and $r_4$ are. Hence the slope of the edge of the supporting polygon through $r_3$ and $r_4$ is fixed, which implies that also the slopes of the edges through $r_2$ and $r_5$ are fixed. Hence, by Observation~\ref{o:3sides}, the whole supporting polygon is fixed, which means that at most one configuration $C$ of the robots can be \PR, due to Lemma~\ref{l:coradial}. Taking $C$ as a set of critical points suffices for all \PR configurations in which $r_0$ and $r_1$ are not companions, due to Theorem~\ref{t:cautious1}.

In the following, we will assume that $r_0$ and $r_1$ are companions. Suppose first that procedure {\sc Move All to SEC} is being executed, and hence $r_0$ and $r_1$ are moving toward SEC. By Lemma~\ref{l:cyclic}, $r_2$ and $r_3$ are companions, and they determine the slope of one edge of the supporting polygon. Therefore, the slope of the edge containing $r_0$ and $r_1$ is also fixed. Let $x$ be the center of SED, and let us consider the two rays from $x$ through $r_0(0)$ and $r_1(0)$, respectively. Let $f_0$ and $f_1$ be, respectively, the points at which these two rays intersect SEC. As $r_0$ and $r_1$ move radially between $x$ and SEC, they can conceivably form infinitely many \PR configurations. However, due to Lemma~\ref{l:homot}, the positions of $r_0$ on the segment $xf_0$ that could give rise to \PR configurations form a closed interval $aa'$, with $a$ closest to $x$ (we assume this interval to be non-empty, otherwise we may take $C'=\varnothing$ as a set of critical points). Similarly, the positions of $r_1$ on $xf_1$ giving rise to \PR configurations determine a closed interval $bb'$, with $b$ closest to $x$.\footnote{The proof of Lemma~\ref{l:homot} also provides a way of constructing such intervals with a compass and a straightedge, and hence by algebraic functions.} Moreover, the line $\ell$ through $a$ and $b$ and the line $\ell'$ through $a'$ and $b'$ are parallel, because the slope of the edge of the supporting polygon containing $r_0$ and $r_1$ is fixed.

Suppose first that $\ell$ is parallel to the line through $f_0$ and $f_1$. Equivalently, $xa$ and $xb$ have the same length. In this case, we take $C'=\{a,b\}$ as a set of critical points. Indeed, let us assume without loss of generality that $r_0(0)f_0$ is not longer than $r_1(0)f_1$. If $r_0(0)$ is past $a'$, no \PR configuration can ever be formed, and we may set $C'=\varnothing$. If $r_0(0)$ lies on the closed segment $aa'$, the {\sc Cautious Move} protocol will make $r_0$ stay still and wait for $r_1$ to reach the same distance from the endpoint of its respective path, and stop there. When this happens, say at time $t$, the line through $r_0(t)$ and $r_1(t)$ is parallel to $\ell$, and therefore the configuration is \PR. Moreover, no \PR configuration is reached before time $t$. Finally, let $r_0(0)$ be before $a$. Then, the {\sc Cautious Move} protocol makes $r_0$ and $r_1$ stop at $a$ and $b$ respectively, and wait for each other. When the robots reach $a$ and $b$, the configuration is the first \PR encountered.

Suppose now that $\ell$ is not parallel to the line through $f_0$ and $f_1$. Without loss of generality, suppose that $xa$ is longer than $xb$. First of all, if $r_0(0)$ is located past $a'$ or $r_1(0)$ is located past $b'$, no \PR configuration can be formed, and we set $C'=\varnothing$. Let $r_0(0)$ belong to the closed segment $aa'$, and let $c$ be the point on $bb'$ such that the line through $r_0(0)$ and $c$ is parallel to $\ell$. If $r_1(0)$ does not lie after $c$, we take $C'=\{c\}$ as a set of critical points. Indeed, the cautious move protocol makes $r_0$ stay still and wait for $r_1$ to reach $c$ and stop there. When this happens the configuration is \PR, and no other \PR configuration is reached before.

Now assume that $r_1(0)$ lies after $c$ (as defined above), or that $r_0(0)$ lies before $a$. We let $c_0=b$ and we let $c_1$ be the intersection between $bf_1$ and the line through $a$ and parallel to $f_0f_1$. Then we inductively define $c_{i+2}$, with $i\geqslant 0$, to be the point on $bf_1$ such that the length of $xc_{i+1}$ is the geometric mean between the lengths of $xc_i$ and $xc_{i+2}$. Let $k$ be the largest index such that $c_k$ is well defined, and let $c_{k+1}=f_1$. For each $0\leqslant i\leqslant k+1$, we define $\ell_i$ to be the line through $c_i$ and parallel to $f_0f_1$. Then, we let $L_i$ be the region of the plane that lies between lines $\ell_i$ and $\ell_{i+1}$, such that $\ell_i\subset L_i$ and $L_i\cap \ell_{i+1}=\varnothing$ (unless $\ell_k=\ell_{k+1}$, in which case $L_k=\ell_k$). We argue that taking $C'=\{c_0, \cdots, c_{k+1}\}$ as a set of critical points prevents the robots from reaching any \PR configuration during the cautious move. Note that a \PR configuration can be formed at time $t$ only if $r_1(t)\in L_i$ and $r_0(t)\in L_{i+1}$, for some $0\leqslant i\leqslant k-1$. This can be true at time $t=0$ but, due to our assumptions, it implies that $r_1(0)$ lies after $c$, and hence $r_1$ will reach $L_{i+1}$ while $r_0$ waits, without forming a \PR configuration. Similarly, if both robots lie initially before $L_0$, the {\sc Cautious Move} protocol will make them reach $L_0$ and wait for each other before proceeding. Moreover, if $r_0(t)\in L_i$ and $r_1(t)\in L_j$ with $j>i$, then $r_1$ waits until $r_0$ reaches $L_j$, and during this process no \PR configuration is formed.

Therefore we can assume that, at some time $t$, both $r_0(t)$ and $r_1(t)$ belong to $L_i$, for some $0\leqslant i\leqslant k$, and none of them is moving. We claim that, if $i<k$, there is a time $t'>t$ at which the two robots are in $L_{i+1}$ and none of them is moving. Moreover, between $t$ and $t'$ no \PR configuration is reached. Indeed, according to the {\sc Cautious Move} protocol, the robots stop at $\ell_{i+1}$ and wait for each other before proceeding, and hence at some point they will clearly be found both in $L_{i+1}$ and not moving. The only way they could form a \PR configuration would be if $r_0$ reached $\ell_{i+1}$ when $r_1$ was still at $\ell_i$. But this cannot happen because, according to the {\sc Cautious Move} protocol, $r_0$ stops at least once (at an intermediate critical point) after $\ell_i$ and before $\ell_{i+1}$. When this happens, $r_0$ cannot proceed any further, and hence it cannot reach $\ell_{i+1}$ if $r_1$ is still at $\ell_i$. By induction on $i$, it follows that $r_0$ and $r_1$ eventually reach $f_0$ and $f_1$, respectively, without ever forming a \PR configuration.

Finally, let us consider the case in which procedure {\sc Move Walkers to SEC/3} is being executed, and $r_0$ and $r_1$ move toward SEC/3. If one of the two robots is initially in SED/3, by Theorem~\ref{t:sec3} no \PR configuration can ever be formed, and $C'=\varnothing$. Hence we may assume that both robots move radially toward SEC/3, as this is taken as a critical point in any case. This case is symmetric to the previous one, and can be treated with a similar reasoning.

To conclude, taking $C\cup C'$ as a set of critical points yields a cautious move that makes the robots freeze at every \PR configuration that is encountered (i.e., whether $r_0$ and $r_1$ are companions or not), due to Theorem~\ref{t:cautious3}.
\end{proof}

\begin{theorem}\label{t:ap1}
Let $\mathcal R$ be frozen at time $t_0$, let $\mathcal R(t_0)$ be an \AP (and not \CO) configuration with no points in the interior of SED/3, let $n>4$, and let the robots execute procedure {\sc Move All to SEC} or procedure {\sc Move Walkers to SEC/3} with suitable critical points. Then, if a \PR configuration is ever formed, the robots freeze as soon as they form one.
\end{theorem}
\begin{proof}
Let $\mathcal A\subset \mathcal R$ be the analogy class of robots that is allowed to move at time $t_0$. As in Lemma~\ref{l:ap1}, it is sufficient to prove the theorem assuming that only $\mathcal A$ moves. Recall that $n$ must be even for a \PR configuration to be formed. If $n=6$, Lemma~\ref{l:ap1} applies. Hence, let us assume that $n\geqslant 8$. Since the analogy classes of an \AP configuration contain either one or two points, it follows that, no matter how $\mathcal A$ is chosen, the hypotheses of Lemma~\ref{l:3slide2} are satisfied, and therefore at most one \PR configuration can be formed. If such a configuration is taken as a set of critical points, our claim follows from Theorem~\ref{t:cautious1}.
\end{proof}

\subsubsection{Final Remarks}\label{sec:ispreregular}

It is straightforward to verify the following.

\begin{observation}\label{o:comp}
In all the theorems of this section, the critical points of the cautious moves are computable by performing finite sequences of algebraic operations on the positions of the robots.
\end{observation}

Also, from our initial observations it follows that a \PR configuration can easily be recognized by the robots, and the supporting polygon is always unique.

\begin{lemma}\label{l:uniquepr}
By a finite sequence of algebraic operations it is possible to decide if a given set of $n>4$ points is \PR and, if it is, to compute the vertices of the supporting polygon, which is unique.
\end{lemma}
\begin{proof}
If $n$ is odd or the points are not in a strictly convex position, then they do not form a \PR configuration, by Lemma~\ref{l:prehull}. Otherwise, the pairs of ``candidate companions'' can be uniquely identified thanks to Observation~\ref{l:dist}. Since $n>4$, the set of candidate companions determines the slopes of at least three edges of the ``candidate supporting polygon''. It is now straightforward to check if these slopes match those of a regular polygon's edges. If they do not, the set is not \PR; otherwise, by Observation~\ref{o:3sides} the candidate supporting polygon is uniquely determined and easy to compute. Now it is sufficient to verify if all the points in the set lie on the so-computed regular polygon, and if they are properly distributed among its edges.
\end{proof}

As a consequence of the previous lemma, the procedure {\sc Is Pre-regular?} used by the \UCF algorithm is well defined.

\subsection{Correctness of the Algorithm}\label{sec:corr}

\begin{lemma}\label{z:regular}
Let $\mathcal R$ be frozen at time $t_0$, let $\mathcal R(t_0)$ be a \RE configuration, and let the robots execute the \UCF algorithm. Then, the robots will never move.
\end{lemma}
\begin{proof}
By assumption, no robot is moving at time $t_0$. Then, whenever a robot performs a Look-Compute, it observes a \RE configuration and remains still, thus keeping the same configuration.
\end{proof}

\begin{lemma}\label{z:prereg}
Let $\mathcal R$ be frozen at time $t_0$, let $\mathcal R(t_0)$ be a \PR configuration with $n>4$, and let the robots execute the \UCF algorithm. Then, the robots will freeze in a \RE configuration without ever colliding.
\end{lemma}
\begin{proof}
By Lemma~\ref{l:uniquepr}, the supporting polygon $P$ of $\mathcal R(t_0)$ is unique and computable. It is straightforward to prove by induction that, every time a robot performs a Look-Compute phase, it observes a \PR configuration with the same supporting polygon $P$. This is certainly true the first time a Look-Compute phase is performed, because $\mathcal R$ is frozen at time $t_0$. Then, whenever a robot observes a \PR configuration with supporting polygon $P$, it executes procedure {\sc Pre-regular}, which makes it move toward its matching vertex of $P$. As robots asynchronously approach their respective matching vertices, the configuration remains \PR, the supporting polygon remains $P$, and no two robots collide, because their trajectories are disjoint. Moreover, each robot approaches its matching vertex by at least $\delta>0$ at each cycle, and therefore it reaches it within finitely many cycles. When a robot reaches its matching vertex, it stops moving, and hence the swarm eventually freezes in a \RE configuration that coincides with the vertex set of $P$.
\end{proof}

\begin{lemma}\label{z:invalid}
Let $\mathcal R$ be frozen at time $t_0$, let $\mathcal R(t_0)$ be an \IN configuration, and let the robots execute the \UCF algorithm. Then, the robots will freeze in a \PR or in a \VA configuration without ever colliding.
\end{lemma}
\begin{proof}
If the robots form a \PR configuration at time $t_0$, there is nothing to prove. Otherwise, we can prove by induction that the robots will always be in an \IN and not \PR configuration and therefore they will always execute procedure {\sc Move All to SEC}, until they freeze in a \PR or in a \VA configuration. This is true at time $t_0$, and moreover the swarm is frozen at that time. Subsequently, robots keep moving radially toward SEC, thus never colliding, never forming \HD or \CO configurations, never altering SEC, and never altering angle sequences and (strong) analogy classes. Hence, as long as the configuration is not \PR or \VA, the procedure that is executed is always {\sc Move All to SEC}. Moreover, if the configuration is initially \EQ (respectively, \BI, \DBI, \PER, \AP), it remains such throughout the execution.

If there are robots in the interior of SED/3 at time $t_0$, they first move onto SEC/3. No \PR configuration can be formed in this phase, due to Theorem~\ref{t:sec3}. A \VA configuration could be formed, though. However, since $\mathcal R(t_0)$ is not \VA by assumption, it follows that in this phase no \VA and \RD configuration can be formed, because at any time the configuration is well occupied if and only if it is well occupied at time $t_0$. On the other hand, a \VA and \WA configuration can be formed only when no robots lie in the interior of SED/3. But this happens only at the very end of the phase, when the configuration is frozen.

Now assume that at time $t_1\geqslant t_0$ the robots are frozen in an \IN configuration with no points in the interior of SED/3. Procedure {\sc Move All to SEC} is executed again, and the robots move toward SEC, one (possibly strong) analogy class at a time, performing a cautious move with suitable critical points. Let the first class $\mathcal C_1\subseteq \mathcal R$ cautiously move toward SEC. It is easy to see that no \VA configuration can be formed during this motion, except perhaps at the very end of the cautious move, say at time $t_2>t_1$, when the robots of $\mathcal C_1$ finally reach SEC (by Lemma~\ref{l:cautious:0}), and the swarm freezes. Indeed, if the period of the configuration is less than 3, then by Observation~\ref{o:analogy1} all robots occupy analogous positions, and therefore the configuration at time $t_1$ is \VA and \WA, which is a contradiction. Hence we may assume $\mathcal C_1$ to be an analogy class, as opposed to a strong analogy class, because the configuration is not \BI (cf.~the definition of procedure {\sc Move All to SEC}). Since the configuration is not \VA at time $t_1$ and its period is at least 3, it means that the internal points belong to at least two different analogy classes (otherwise the configuration would be \VA and \WA), one of which is $\mathcal C_1$. Therefore, as long as some points of $\mathcal C_1$ are still internal, the configuration cannot be \VA and \WA. Moreover, the configuration cannot be \VA and \RD either, because, according to procedure {\sc Cautious Move}, only the robots that are farther from SEC can move. Hence, because each robot located on SEC/3 has at least one (auxiliary) critical point on its path, after the time the first robots of $\mathcal C_1$ start moving and before time $t_2$ there will always be robots lying neither on SEC/3 nor on SEC. It follows that a cautious move with the aforementioned critical points satisfies property $\mathcal P_1$ that the swarm freezes as soon as a \VA configuration is formed. On the other hand, by Theorems~\ref{t:p1} and~\ref{t:ap1}, a cautious move with suitable critical points satisfies property $\mathcal P_2$ that the swarm freezes as soon as a \PR configuration is formed. By Theorem~\ref{t:cautious3}, there exists a cautious move satisfying both $\mathcal P_1$ and $\mathcal P_2$, whose critical points are computable by algebraic operations, by Observation~\ref{o:comp}.

Suppose now that the robots of $\mathcal C_1$ complete the cautious move (cf.~Lemma~\ref{l:cautious:0}), reaching SEC at time $t_2$ without ever forming a \PR or a \VA configuration, and freezing. Then the next class $\mathcal C_2$ moves to SEC from a frozen state, and the previous paragraph's reasoning applies again. By induction, either the robots freeze in a \PR or a \VA configuration during a cautious move, or they all finally reach SEC and freeze. Note that the configuration at this point is not \HD, because it was not at time $t_0$, and robots have only performed radial moves toward SEC. Hence the configuration must be \VA and \WA. Also, no two robots have collided, because the configuration was not \CO at time $t_0$, and radial moves toward SEC cannot create new co-radialities.
\end{proof}

\begin{lemma}\label{z:valinv}
Let $\mathcal R(t_0)$ be a \VA or \IN configuration, let all robots' trajectories at time $t_0$ be disjoint, and suppose that, if a robot $r\in \mathcal R$ is not frozen at time $t_0$, then the following conditions hold:
\begin{itemize}
\item $r(t_0)$ lies in the interior of SED/3;
\item the destination point of $r$ at time $t_0$ is on SEC/3;
\item $r(t_0)$ and $r$'s destination point at time $t_0$ lie in the interior of the same main sector of $\mathcal R(t_0)\setminus \{r(t_0)\}$.
\end{itemize}
If the robots execute the \UCF algorithm, then they will freeze in a \VA or in an \IN configuration without ever colliding.
\end{lemma}
\begin{proof}
Recall that a \VA or \IN configuration is not \CE, not \HD, and not \CO. Also, since a robot $r$ and its destination at time $t_0$ lie in the same main sector of $\mathcal R(t_0)\setminus \{r(t_0)\}$, it follows that the center of SED does not lie on the trajectory of $r$ at time $t_0$ (by definition of SED).

If there is no robot in the interior of SED/3 at time $t_0$, then by assumption all robots are frozen, and there is nothing to prove. Otherwise, the configuration is not \PR at time $t_0$, due to Theorem~\ref{t:sec3}. Moreover, it is straightforward to see that, as the non-frozen robots move toward their destination points, the configuration remains \VA or \IN and does not become \CE, \HD, \CO, or \PR (recall that the non-frozen robots' trajectories at time $t_0$ are within SED/3). Also, no collisions occur because the trajectories are all disjoint, and SED remains unaltered, because no robot on SEC moves. Therefore, the procedure that is executed by the first robots performing a Look-Compute phase (say, at time $t_1\geqslant t_0$) is either {\sc Valid and Ready} (indeed, there are robots in the interior of SED/3, hence the configuration cannot be \WA) or {\sc Invalid}. Both procedures make the robots that lie in the interior of SED/3 (which exist, by assumption) move radially toward SEC/3. Hence, at time $t_1$ the hypotheses of the lemma are still satisfied, and the same argument can be repeated.

Each moving robot either reaches SEC/3 or moves by at least $\delta$ at each turn; hence, in finite time, there are no robots left in the interior of SED/3. As soon as this happens, the swarm is frozen in a \VA or an \IN configuration, and no collisions have occurred.
\end{proof}

\begin{lemma}\label{z:coradial}
Let $\mathcal R(t_0)$ be a \CO, not \CE and not \HD configuration with $n>4$, and suppose that, if a robot $r\in \mathcal R$ is not frozen at time $t_0$ and $r(t_0)$ does not lie in SED/3, then the following conditions hold:
\begin{itemize}
\item $r(t_0)$ is co-radial in $\mathcal R(t_0)$;
\item $r(t_0)$ is the closest to the center of SED among its co-radial robots;
\item The destination point of $r$ at time $t_0$ is on SEC/3 and co-radial with $r(t_0)$.
\end{itemize}
Also suppose that, if a non-co-radial robot $r\in \mathcal R$ is not frozen at time $t_0$ and $r(t_0)$ lies in SED/3, then the following conditions hold:
\begin{itemize}
\item $r(t_0)$ lies in the interior of SED/3;
\item the destination point of $r$ at time $t_0$ is on SEC/3;
\item $r(t_0)$ and $r$'s destination point at time $t_0$ lie in the interior of the same main sector of $\mathcal R(t_0)\setminus \{r(t_0)\}$.
\end{itemize}
Further suppose that all other robots are frozen at time $t_0$, except perhaps for one robot $s\in\mathcal R$, for which the following conditions hold:
\begin{itemize}
\item $s(t_0)$ lies in SED/3;
\item $s(t_0)$ is co-radial in $\mathcal R(t_0)$;
\item $s(t_0)$ is the closest to the center of SED among its co-radial robots;
\item the line through $s(t_0)$ and the center of SED bounds two open half-planes, one of which, $\Gamma$, contains exactly one robot $s'\in\mathcal R$ at time $t_0$;
\item $s'(t_0)$ is not co-radial in $\mathcal R(t_0)$;
\item the destination point of $s$ and $s'$ at time $t_0$ is on SEC/3 and in $\Gamma$.
\end{itemize}
Finally, suppose that no two robots' trajectories at time $t_0$ intersect. If the robots execute the \UCF algorithm, then they will freeze in a \VA or in an \IN configuration without ever colliding.
\end{lemma}
\begin{proof}
Let $\mathcal M_0\subset\mathcal R$ be the set of robots outside SED/3 that are not frozen at time $t_0$. The first robots to execute a {Look-Compute} phase execute procedure {\sc Co-radial}, because the configuration cannot be \PR, due to Theorem~\ref{t:coradial}. As long as there are non-co-radial robots in the interior of SED/3, they move radially toward SEC/3, either radially if they perform a {Look-Compute} phase if they were frozen at time $t_0$, or laterally if they were not frozen at time $t_0$. Meanwhile, some robots of $\mathcal M_0$ perhaps move radially toward SEC/3, and $s$ perhaps moves and becomes non-co-radial. In this phase no two robots that were not co-radial with each other at time $t_0$ become co-radial, and in particular no collisions can occur. Indeed, recall that all trajectories are disjoint at time $t_0$, and the destination points of non-frozen robots at time $t_0$ are always on SEC/3. Hence, even if a robot that is moving laterally at $t_0$ stops and starts moving radially, it still does not collide with other robots. Also, SED is preserved, because no robot on SEC moves. Therefore, the hypotheses of the lemma are satisfied throughout this phase, and at some point only co-radial robots lie in the interior of SED/3, and all robots are frozen except perhaps some robots of $\mathcal M_0$ and $s$.

At this point, if the co-radial robots closest to the center of SED lie outside SED/3, some of them move radially toward SEC/3. Eventually, say at time $t_1\geqslant t_0$, some co-radial robots are found in SED/3. When this happens, all the robots are frozen, except perhaps $s$ and some robots in $\mathcal M_1\subset \mathcal R$, which lie outside SED/3 and are still moving radially toward SEC/3. Now the co-radial robots that are closest to the center of SED are allowed to move laterally, and let $\mathcal C\subset \mathcal R$ be the non-empty set of robots that actually move laterally in this phase ($s$ may or may not be in $\mathcal C$). As soon as a robot in $\mathcal C$ starts moving, it becomes a non-co-radial robot lying in the interior of SED/3, and therefore it prevents other robots from making lateral moves. It follows that all the robots in $\mathcal C\setminus\{s\}$, during the {Look-Compute} phase before moving laterally, observe the same smallest positive angular distance $\alpha$ between robots, and they all move in such a way that their destination point has angular distance $\alpha/3$ from their location at time $t_1$. In particular, $\alpha$ is not greater than the angular distance of two robots of $\mathcal C$ at time $t_1$, and hence the trajectories of all these robots are disjoint. Moreover, if $s\in\mathcal C$ and therefore $s$ moves laterally, it enters the half-plane $\Gamma$, approaching $s'$, which is now frozen on SEC/3 because it is not co-radial. In addition, $s$ does not collide with another robot, because it moves into the interior of the sector determined by $s(t_1)$ and $s'(t_1)$, which contains the trajectory of no robot other than $s$.

Since $\mathcal C$ is not empty, when the robots of $\mathcal C$ start moving, the number of co-radial robots strictly decreases. Let $t_2>t_1$ be the first time at which a robot performs a {Look-Compute} phase after all the robots of $\mathcal C$ have started moving. It is easy to see that the configuration is not \CE or \HD at time $t_2$, or two non-co-radial robots would have become co-radial at some point. Suppose first that the configuration is still \CO at time $t_2$. Then, the hypotheses of the lemma are satisfied again, but there are fewer co-radial robots. Hence we can repeat the previous argument until no co-radial robots are left. Without loss of generality, suppose that, after the robots in the set $\mathcal C$ defined above have started moving, the configuration is not \CO any more, and let $t_3$ be the first time at which a robot performs a {Look-Compute} phase and does not see a \CO configuration. Then, the configuration cannot be \CE or \HD either, and hence it is \VA or \IN. Also note that all the robots outside SED/3 are frozen, because the only such robots that could be moving must be co-radial. It follows that the hypotheses of Lemma~\ref{z:valinv} are satisfied, and therefore the robots will finally freeze in a \VA or \IN configuration without colliding.
\end{proof}

\begin{lemma}\label{z:hdce}
Let $\mathcal R$ be frozen at time $t_0$, let $\mathcal R(t_0)$ be a \CE or \HD configuration with $n>4$, and let the robots execute the \UCF algorithm. Then, the robots will freeze in a \VA or in an \IN configuration without ever colliding.
\end{lemma}
\begin{proof}
The configuration at time $t_0$ is not \VA, by definition. Also, by Theorems~\ref{t:halfdisk} and~\ref{t:coradial}, a \HD or \CE (hence \CO) set cannot be \PR. Suppose that $\mathcal R(t_0)$ is \CE. Then, according to the algorithm, procedure {\sc Central} is executed, and no robot moves until the robot $r$ lying at the center of SED performs a {Look-Compute} phase and moves toward a point on SEC/3 that is not co-radial with any robot other than $r$. Let $t_1>t_0$ be the first time at which a robot performs a {Look-Compute} phase and it does not see $r$ at the center of SED. On the other hand, if $\mathcal R(t_0)$ is not \CE, we just take $t_1=t_0$. In both cases, at time $t_1$ the swarm is in a configuration that is not \CE and may be \HD, or \CO, or \VA, or \IN, and no robot is moving, except perhaps one non-co-radial robot $r$ in SED/3 that is moving radially toward SEC/3.

Suppose that $\mathcal R(t_1)$ is \HD and not \CE. Then, procedure {\sc Half-disk} is executed by the first robots that perform a {Look-Compute} phase, because the configuration cannot be \PR, due to Theorem~\ref{t:halfdisk}. Assume first that the robots are all collinear, and one principal ray contains exactly two robots. Because $n>4$, the other principal ray $\ell$ contains at least three robots. Note that the configuration is frozen at time $t_1$, because we are assuming that the only moving robot must be non-co-radial, and here all robots are co-radial. The closest to the center of SED among the robots lying on $\ell$ moves radially to SED/3, and then it moves laterally within SED/3. At this point, there is at most one robot moving (within SED/3), and the configuration is \HD with only one empty half-plane. Let $t_2\geqslant t_1$ be the first time at which a robot performs a {Look-Compute} phase and observes such a configuration.

Suppose now that the robots at time $t_1$ are all collinear, and one principal ray contains only one robot $s$ (which lies on SEC). Then the swarm is frozen, because $r$ would have to be non-co-radial and in SED/3, but the only non-co-radial robot is $s$, which is not in SED/3. From this configuration, the robot that is closest to the center of SED, $s'$, first reaches SED/3 by moving radially, and then it moves laterally within SED/3. At this point, $s''$, the robot lying on the principal line that is now closest to the center of SED, moves radially to reach SEC/3, while $s'$ moves again within SED/3 to become co-radial with $s$. If $s'$ and $s''$ stop before reaching their destinations, they move again toward them. Upon reaching their destinations, they wait for each other. Hence, eventually, the swarm freezes in a configuration in which all robots are collinear and one principal ray contains exactly two robots. From this configuration, as detailed in the previous paragraph, the swarm reaches at time $t_2$ a \HD configuration with only one empty half-plane, in which only one robot is moving (within SED/3).

Now, let the configuration at time $t_2\geqslant t_1$ be \HD and assume that, if all robots are collinear, then both principal rays contain at least three robots. Also, there may be a unique  non-frozen robot $r$, which is not co-radial and located in SED/3 at time $t_2$. The destination point of $r$ is on SEC/3, and the trajectory of $r$ at time $t_2$ lies in the interior of one main sector of $\mathcal R(t_2)\setminus\{r(t_2)\}$. Note that this could even be the situation at time $t_1=t_2$. Once again, procedure {\sc Half-disk} is executed. If a principal ray $\ell$ has no robots in SED/3, a unique robot $s$ moves to reach this area. This robot is chosen in such a way that its angular distance from $\ell$ is minimum, and it is the closest to the center of SED of such robots. In particular, $s$ could lie on $\ell$, and move radially. Note that, if $s$ is not on $\ell$ at time $t_2$, and even if $s=r$, it does not become co-radial until it actually reaches $\ell$, and even if another robot $s'$ is moving to the other principal ray, $s$ and $s'$ never collide. In particular, if at time $t_2$ the principal line contains only two robots (on SEC), and all other robots are co-radial with each other, then the robot closest to the center of SED, $s$, first moves toward one of the principal rays. When it stops being co-radial, the second closest robot $s'$ moves to the other principal ray. In all cases, while this happens, the configuration remains \HD, hence it never becomes \PR by Theorem~\ref{t:halfdisk}, and procedure {\sc Half-disk} keeps being executed. Eventually $s$ reaches $\ell$ and, if there is an $s'$ moving toward $\ell'$, $s$ waits for it (and vice versa).

At some point, say at time $t_3\geqslant t_2$, the configuration is either frozen with all robots collinear and at least three robots on each principal ray, or the robots are not all collinear and the only robot that may be not frozen is $r$ (as defined in the previous paragraph). In both cases, each principal ray has at least one robot in SED/3. Let $s$ and $s'$ be the robots on the two principal rays that are closest to the center of SED. According to procedure {\sc Half-disk}, at least one between $s$ and $s'$ moves to an empty half-plane, within SED/3, and without causing collisions. We claim that, at some point after time $t_3$, the configuration stops begin \HD. In particular, if at time $t_3$ the robots are collinear and both $s$ and $s'$ move to the same empty half-plane, another pair of robots on the principal line will move into SED/3 and then at least one one them will move into the other empty half-plane.

Let $t_4\geqslant t_3$ be the first time at which a robot performs a {Look-Compute} phase and does not observe a \HD configuration. Note that this could also happen at time $t_1=t_4$. If $\mathcal R(t_4)$ is \CO, then the hypotheses of Lemma~\ref{z:coradial} are satisfied at time $t_4$, and hence the swarm eventually freezes in a \VA or \IN configuration, and no collisions occur. If $\mathcal R(t_4)$ is \VA or \IN, then the hypotheses of Lemma~\ref{z:valinv} are satisfied at $t_4$. In particular, if the configuration has evolved from a \HD, the only robots left on the former principal line are the two lying on SEC, because otherwise the configuration would be \CO. Therefore Lemma~\ref{z:valinv} applies, the swarm freezes in a \VA or \IN configuration, and no collisions occur.
\end{proof}

\begin{lemma}\label{z:transition2}
Let $S$ be a \VA and \WA set whose points all lie on SEC, and let $W$ be the set of its walkers. Let $S'=(S\setminus W)\cup W'$, where $W'=\mathcal F'(W)$. Then, $S'$ is \VA and \RD.
\end{lemma}
\begin{proof}
Note that $\mbox{SED}(S)=\mbox{SED}(S')$, because $W$ is a movable analogy class, by Observation~\ref{o:walkers}. Also observe that $S'$ is not \CO and not \HD, because $S$ is not. Moreover, $W$ is a relocation of $W'=\mathcal I(S')$, hence $S'$ is well occupied, and therefore it is \VA and \RD.
\end{proof}

\begin{lemma}\label{z:valid1a}
Let $\mathcal R$ be frozen at time $t_0$, and suppose that the following conditions hold:
\begin{itemize}
\item $\mathcal R(t_0)$ is \VA and \WA;
\item If $\mathcal R(t_0)$ is \VA and \RD, all the internal robots lie on their respective finish lines;
\item At time $t_0$, all the internal robots are walkers.
\end{itemize}
Then, if the robots execute the \UCF algorithm, they will freeze either in a \PR configuration, or in the \VA and \RD configuration $(\mathcal R(t_0)\setminus\mathcal W(\mathcal R(t_0)))\cup \mathcal F'(\mathcal W(\mathcal R(t_0)))$. During the process, no two robots collide.
\end{lemma}
\begin{proof}
If $\mathcal R(t_0)$ is \EQ or \BI, then it has no walkers. It follows that all the robots are on SEC, and hence they form a frozen \PR configuration. In this case, there is nothing to prove.

Now assume that $\mathcal R(t_0)$ is \PER or \AP. Since the robots are frozen at time $t_0$, the first robots that perform a Look-Compute phase agree on a target set, a point-target correspondence, and a set of walkers, which is an analogy class to which all the internal robots belong. According to procedures {\sc Valid and Ready} and {\sc Valid and Waiting}, in all cases procedure {\sc Move Walkers to SEC/3} is executed. Indeed, recall that in a \VA and \WA configuration there are no robots in the interior of SED/3, and therefore {\sc Valid and Waiting} is executed in any case. Then, as soon as the walkers are activated, they start moving radially toward SEC/3 following the {\sc Cautious Move} protocol, and while this happens the footprint of the configuration remains the same, and so does the set of walkers (indeed, the walkers form a movable analogy class, due to Observation~\ref{o:walkers}, hence SED is preserved). In this phase the configuration remains \VA and \WA, and it may become \PR, in which case the robots freeze, due to Theorems~\ref{t:p1} and~\ref{t:ap1} (note that, as the robots move radially, the period of the configuration does not change).

Also, as soon as some walkers start moving after time $t_0$, and as long as some walkers are not on SEC/3, the configuration cannot be \VA and \RD. This is because, due to the {\sc Cautious Move} protocol, only the walkers that are farthest from SEC/3 are allowed to move at any given time. Moreover, the walkers that are initially on SEC have to stop at least at one (auxiliary) critical point before reaching SEC/3. It follows that, unless the walkers are all on SEC/3, and except perhaps when the configuration is $\mathcal R(t_0)$, there are always walkers located in the annulus strictly between SEC and SEC/3. While this is true, the configuration is never recognized as \VA and \RD.

Therefore, if the configuration does not become \PR, procedure {\sc Valid and Waiting} keeps being executed, and the walkers keep moving toward SEC/3. Eventually all the walkers freeze on SEC/3, due to Lemma~\ref{l:cautious:0}. When this happens, the configuration finally becomes \VA and \RD, due to Lemma~\ref{z:transition2}. Note that, in the above, all robots either stay still or move radially between SEC and SEC/3, and the configuration is not \CO. Hence no two robots collide.
\end{proof}

\begin{lemma}\label{z:valid1b}
Let $\mathcal R$ be frozen at time $t_0$, and suppose that the following conditions hold:
\begin{itemize}
\item $\mathcal R(t_0)$ is \VA and \WA;
\item If $\mathcal R(t_0)$ is \VA and \RD, all the internal robots lie on their respective finish lines;
\item At time $t_0$, some internal robots are not walkers.
\end{itemize}
Then, if the robots execute the \UCF algorithm, they will freeze either in a \PR configuration, or in the \VA and \WA configuration $\mathcal F(\mathcal R(t_0))$. During the process, no two robots collide.
\end{lemma}
\begin{proof}
According to procedures {\sc Valid and Ready} and {\sc Valid and Waiting}, in all cases procedure {\sc Move All to SEC} is executed. Indeed, recall that in a \VA and \WA configuration there are no robots in the interior of SED/3, and therefore procedure {\sc Valid and Waiting} is executed in any case. By definition of \WA, the internal robots are all analogous. Hence, whenever an internal robot is activated, performs a cautious move toward SEC that, due to Theorems~\ref{t:equi}--\ref{t:ap1}, makes all robots freeze as soon as a \PR configuration is reached. As the robots move radially, the configuration remains \VA and \WA. Moreover, as soon as a robot starts moving, the configuration ceases to be \RD, and cannot become \RD throughout the cautious move. Indeed, due to the {\sc Cautious Move} protocol, only the walkers that are farthest from SEC are allowed to move at any given time. Moreover, the robots that are initially on SEC/3 have to stop at least at one (auxiliary) critical point before reaching SEC. It follows that, except perhaps when the configuration is $\mathcal R(t_0)$, there are always robots located in the annulus strictly between SEC and SEC/3. While this is true, the configuration is never recognized as \VA and \RD. Moreover, the set of walkers remains the same throughout the motion, and so the same procedure keeps being executed as the robots move. Due to Lemma~\ref{l:cautious:0}, the robots either reach SEC and freeze on it (forming a \VA and \WA configuration) or they freeze in a \PR configuration. Note that, in the above, all robots either stay still or move radially toward SEC, and the configuration is not \CO. Hence no two robots collide.
\end{proof}

\begin{lemma}\label{z:valid2}
Let $\mathcal R$ be frozen at time $t_0$, let $\mathcal R(t_0)$ be a \VA and \RD configuration, and let the robots execute the \UCF algorithm. Then, the robots will freeze in a \VA and \WA and \RD configuration in which all internal robots lie on their respective finish lines. During the process, the configuration remains \VA and \RD, the finish set does not change, and no two robots collide.
\end{lemma}
\begin{proof}
According to procedure {\sc Valid and Ready}, if initially there are internal robots that lie strictly inside SED/3, they move radially toward SEC/3. During this phase, the configuration remains \VA and \RD, and therefore the same procedure keeps being executed by all robots that perform a {Look-Compute} after time $t_0$. Hence, at some time $t_1\geqslant t_0$, all internal robots are frozen on SEC/3, and the finish lines and correspondences at time $t_0$ and at time $t_1$ are the same.

Now, as soon as an internal robot performs a {Look-Compute} phase, it executes procedure {\sc Move to Finish Line}, which makes it move toward its corresponding finish line, provided that no other robot is co-radial with some point on the trajectory. By Proposition~\ref{p:finishreachable}, at least one robot can reach its corresponding finish line, and so eventually at least one robot moves laterally. As soon as a robot starts moving, it stops being on SEC/3, and therefore any robot that performs a {Look-Compute} afterwards and lies on SEC/3 does not move. Whenever a moving robot stops because it is interrupted by the adversary, it moves radially to SEC/3 during its next cycles.

Therefore the internal robots alternate between moving all to SEC/3 and toward their finish lines. Observe that no robot's angular distance to a point on its corresponding finish line is $\pi$, and recall that no robot can be stopped by the scheduler before moving by $\delta$ at each turn. Hence, for each robot $r$, there is an angle $\bar\theta(\theta_0,\delta)>0$, depending only on $\delta$ and on $r$'s angular distance to (a point on) its corresponding finish line at time $t_0$, such that,  whenever $r$ moves toward its finish line, it either reaches it or its angular distance to it decreases by at least $\bar\theta(\theta_0,\delta)$. By Proposition~\ref{p:finishreachable}, at any time there is always a robot whose corresponding finish line is reachable, and therefore eventually all robots get on their finish lines. At this point, the internal robots move radially to SEC/3, and they freeze. When this happens, the configuration is still \VA and \RD, but it is also \WA, due to Proposition~\ref{p:internalclass}.

Note that, in the above paragraphs, we assumed that the internal robots keep executing procedure {\sc Valid and Ready}. To prove that this is indeed the case, we show that the configuration remains \VA and \RD, and it never becomes \PR, \CE, \CO, or \HD. Indeed, note that the internal robots can never get out of SED/3 or of the occupied sectors in which they lie initially. Hence the configuration cannot become \PR, due to Theorem~\ref{t:sec3}, because there are robots in SED/3 at all times. It is easy to see that the configuration cannot become \CE either, because no robot's angular distance to (a point on) its corresponding finish line is $\pi$, and therefore the robot never has to cross the center of SED to reach it. Also, recall that a robot moves toward its corresponding finish line only if it can reach it; all other moves are radial, and therefore they do not affect angular distances between robots. Moreover, the correspondence between robots and finish lines preserves their cyclic order around the center of SED. It follows that, if a robot can reach its corresponding finish line at some point and starts moving toward it, no other robot can get between them and cause the formation of a co-radiality. Hence the configuration never becomes \CO, and in particular no collisions occur. Finally, the formation of a \HD configuration is prevented by the fact that the configuration is not \HD at time $t_0$ and the internal robots remain in the interior of their initial occupied sectors at all times.
\end{proof}

\begin{lemma}\label{z:next1}
Let $S$ be a \VA set of $n>5$ points, all of which lie on SEC. Suppose that $W=\mathcal W(S)$ is not empty, and let $S'=(S\setminus W)\cup\mathcal F'(W)$. Let $L$ be the relocation of $\mathcal I(S')$ (with respect to $S'$) having one point on each finish line of $S'$, and let $S''=\mathcal E(S')\cup L$. Then, the following statements hold.
\begin{itemize}
\item $S$ does not have fewer analogy classes than $S''$.
\item If $S$ has an axis of symmetry, then $S''$ has the same axis of symmetry and the same target set. Also, each point of $\mathcal E(S')$ has the same target in both $S$ and $S''$.
\end{itemize}
Moreover, at least one of the following statements holds.
\begin{itemize}
\item $S$ has strictly more analogy classes than $S''$, or
\item $S$ has no axes of symmetry and $S''$ has some axes of symmetry, or
\item $S$ is locked and it does not have more satisfied points than $S''$, or
\item $S$ is not locked and it has strictly fewer satisfied points than $S''$.
\end{itemize}
\end{lemma}
\begin{proof}
Note that $\mathcal W(S)$ is a movable analogy class of $S$, by Observation~\ref{o:walkers}. Hence $\mbox{SED}(S)=\mbox{SED}(S')=\mbox{SED}(S'')$. Also, $S'$ is \VA and \WA and \RD, and so $W=\mathcal W(S')$, and $L$ is well defined. Moreover, the points of $L$ are all analogous in $S''$, due to Proposition~\ref{p:internalclass}. Therefore, if two points of $S\setminus W$ are analogous in $S$, then they are analogous also in $S''$. As a consequence, the number of analogy classes in $S''$ does not exceed the number of analogy classes in $S$. Specifically, $S''$ has strictly fewer analogy classes than $S$ if and only if $L$ is a proper subset of an analogy class of $S''$. In the following we denote by $P$ the principal relocation of $\mathcal I(S')$ with respect to $S'$, and we let $S^*=\mathcal E(S')\cup P$.

Suppose that $S$ has an axis of symmetry $\ell$. Since $W$ is an analogy class of $S$, it has $\ell$ as an axis of symmetry as well, due to Proposition~\ref{p:symm}. But then $\ell$ is also an axis of symmetry of $\mathcal F'(W)$, and therefore of $S'$. Moreover, $\ell$ is an axis of symmetry of $P$ (cf.~Proposition~\ref{p:internalclass}), and hence of $S^*$, and of the target set of $S^*$. Similarly, $\ell$ is an axis of symmetry of $L$, and therefore of $S''$.

If a point $p\in S$ lies on $\ell$, then $p$ is satisfied in $S$ (by definition of target set), and hence $p\notin W$, by Observation~\ref{o:walkers}. Then $p$ belongs also to $S'$, $S^*$, and $S''$. Moreover, $p$ is satisfied in $S'$, $S^*$, and $S''$, and therefore the target sets of $S$, $S'$, $S^*$, and $S''$ are the same. Similarly, if no point of $S$ lies on $\ell$, then no point of $S'$, $S^*$, and $S''$ does. Indeed, even if $S$ and $S^*$ are locked, $L$ consists of two antipodal points of SEC that are symmetric with respect to $\ell$, and none of them lies on $\ell$ (cf.\ the definition of finish line and Proposition~\ref{p:locked2}). It follows that, in all cases, the target sets of $S$, $S'$, $S^*$, and $S''$ are the same, and the points of $S\setminus W$ that are satisfied in $S$ are also satisfied in $S'$, $S^*$, and $S''$ (cf.~Proposition~\ref{p:targetcompatible}).

If $S$ does not have an axis of symmetry and $S''$ does, there is nothing to prove. So, in the following we will assume that $S$ and $S''$ are either both asymmetric or both symmetric. We will also assume that $S$ and $S''$ have the same number of analogy classes, and hence that $L$ is an analogy class of $S''$.

Let $S$ be symmetric and locked. By Observation~\ref{o:walkers}, $W$ is a movable and non-satisfied analogy class of $S$. Moreover, since $S$ is symmetric, we already proved that $S''$ has the same target set of $S$, and that all the points of $S\setminus W$ that are satisfied in $S$ are also satisfied in $S''$. Therefore, $S''$ has at least as many satisfied points as $S$.

Let $S$ be symmetric and not locked. Then the points of $W$ are not satisfied and can reach their targets in $S$. Recall that targets and correspondences are preserved from $S$ to $S'$ to $S^*$ to $S''$, because $S$ is symmetric. Therefore $S^*$ is not locked, because $P$ is improvable in $S^*$, as $W$ is improvable in $S$. Let $R$ be the tentative finish set of $S'$. By definition, $R$ is the set of targets $T$ of the points of $\mathcal I(S')$, unless $P$ is a proper subset of an analogy class of $\mathcal E(S')\cup P$. However, in this case $R=P$ and, by definition of finish set, $R=L$. This implies that $L$ is a proper subset of an analogy class of $S''$, which contradicts our previous assumptions. Hence $R=T$ and, since $S^*$ is not locked, $R=L$. It follows that the points of $L$ are satisfied in $S''$. Recalling that the points of $S\setminus W$ that are satisfied in $S$ are also satisfied in $S''$, we conclude that $S''$ has strictly more satisfied points than $S$.

Suppose that neither $S$ nor $S''$ have an axis of symmetry. Let $C$ be the set of satisfied points of $S$. By definition of target, $C$ is a concordance class of maximum size. By Observation~\ref{o:walkers}, $W$ is a non-satisfied analogy class of $S$, and therefore no point of $W$ is in $C$. Since $L$ is a relocation of $\mathcal I(S)$, there are some concordance classes in $S''$ with $|C|$ points on SEC: indeed, $C$ must be a subset of one of such classes. Considering that $S''$ is not symmetric by assumption, this implies that it has at least $|C|$ satisfied points, as well. Hence, if $S$ is locked, there is nothing to prove, because it does not have more satisfied points than $S''$.

Let therefore $S$ be not locked. We claim that $S^*$ cannot be symmetric. Assume for a contradiction that $\ell$ is an axis of symmetry of $S^*$. Suppose that $P$ is an analogy class of $S^*$. Then, by Proposition~\ref{p:symm}, $\ell$ is an axis of symmetry of $P$, as well. Hence, as argued above, $\ell$ is an axis of symmetry of $S''$, contradicting our assumptions. If $P$ is not an analogy class of $S^*$, then it must be a proper subset of an analogy class, because all the points of $P$ are analogous (by Proposition~\ref{p:internalclass}). Then, by definition of tentative finish set, $R=P$. Hence $R$ is not an analogy class of $S^*$, and in particular it cannot possibly be an unlocking analogy class of $S^*$, implying that $R=L$, by definition of finish set. As a consequence, $S''=S^*$, meaning that $S''$ is symmetric, which contradicts our assumptions. Hence $S^*$ is not symmetric.

As a consequence, by definition of target, $S'$ has at least $|C|$ satisfied points. Moreover, since the points of $W$ can reach their targets in $S$, it follows that there is a concordance class in $C'$ in $S^*$ with $|C|$ points in $\mathcal E(S')$ such that some relocation $R'$ of $\mathcal F'(W)$ with respect to $S'$ belongs to the same concordance class as $C'$ in $\mathcal E(S')\cup R'$. In particular, one of such concordance classes $C'$ defines the set of targets in $S'$, and therefore the tentative finish set $R$ coincides with the set of targets $T$ of the points of $\mathcal I(S')$ with respect to $S'$. Indeed, if this was not true, then $P$ would be a proper subset of some analogy class of $S^*$, and $R=P$. Hence $R$ is not an analogy class in $S^*$, and $R=L$. Moreover, by Proposition~\ref{p:internalclass}, the points of $L$ are all analogous in $S''$. Hence $L$ is a proper subset of an analogy class of $\mathcal S''$, contradicting our previous assumption. We conclude that $R$ must coincide with $T$. Hence $\mathcal E(S')\cup R$ has a unique concordance class with strictly more than $|C|$ points, which therefore define the targets, and are all satisfied. Such a concordance class contains $R$, and hence $R$ cannot possibly be a non-satisfied unlocking analogy class. Then, by definition of finish line, $R=L$. It follows that $S''$ has a unique concordance class with more than $|C|$ points, which are satisfied. This means that $S''$ has strictly more satisfied points than $S$.
\end{proof}

\begin{lemma}\label{p:unlocked}
Let $S$ be a locked \VA set of $n>5$ points, all of which lie on SEC. Let $W=\mathcal W(S)$, and let $S'=(S\setminus W)\cup\mathcal F'(W)$. Let $L$ be the relocation of $\mathcal I(S')$ (with respect to $S'$) having one point on each finish line of $S'$, and let $S''=\mathcal E(S')\cup L$. Then, at least one of the following statements holds.
\begin{itemize}
\item $S''$ is not locked.
\item $S''$ has fewer analogy classes than $S$.
\item $S''$ has fewer non-movable analogy classes than $S$.
\end{itemize}
\end{lemma}
\begin{proof}
Since $S$ is locked, by Proposition~\ref{p:locked2} it is \AP, and hence its period is $n$. Therefore, by Observation~\ref{o:analogy1}, $S$ has more than one analogy class, and, by definition of walker, $\mathcal W(S)$ is not empty, and it is a non-satisfied unlocking analogy class of $S$. By definition of unlocking analogy class, $W$ is movable, and hence $\mbox{SED}(S)=\mbox{SED}(S')=\mbox{SED}(S'')$. Note also that $S'$ is \VA and \RD, hence $L$ is well defined. Now, let $P$ be the principal relocation of $\mathcal I(S')$ with respect to $S'$, let $S^*=\mathcal E(S')\cup P$, and let $R$ be the tentative finish set of $S'$.

Suppose first that $P$ is a proper subset of an analogy class of $S^*$. Then, by definition of tentative finish set, $R=P$. Also, since $R$ is not an analogy class of $S^*$, then, by definition of finish set, $R=L$, and therefore $S''=S^*$. Since $\mathcal W(S)$ is an analogy class of $S$, and $L$ is a proper subset of an analogy class of $S''$, it follows that $S''$ has strictly fewer analogy classes than $S$.

Now suppose that $P$ is an analogy class of $S^*$. Suppose also that $S$ is \UA and $S^*$ is not \UA. Then, by Observation~\ref{o:analogy2}, $W$ consists of a single point, and therefore so does $P$. But $P$ is an analogy class of $S^*$, and so $S^*$ must be \BA, again by Observation~\ref{o:analogy2}. Moreover, the unique point $p\in P$ lies on the unique axis of symmetry of $S^*$ and, by definition of target, it is satisfied in $S^*$. Therefore, $p$ is the target corresponding to the unique point of $\mathcal I(S')$ and, by definition of tentative finish line, $R=P$. Note that, if $\mathcal E(S')\cup R=S^*$ is locked, then $R$ cannot be an unlocking analogy class of it. Indeed, due to Proposition~\ref{p:locked2}, since $S^*$ is \BA, its unique analogy class consists of two points. Hence, by definition of finish set, $R=L$, and therefore $S''=S^*$. This implies that $S''$ is \BA, and as such it has fewer analogy classes than $S$, as required.

So, in the following, we assume that $P$ is an analogy class of $S^*$ and that, if $S$ is \UA, then also $S^*$ is \UA. Let $T$ be the set of targets of the internal points of $S^*$. We claim that $T$ is not a relocation of $\mathcal I(S^*)$, so let us assume the opposite. We distinguish two cases.
\begin{itemize}
\item Let $S$ have an axis of symmetry. Then $S^*$ has the same axis of symmetry, and the same target set as $S$, with the same correspondences for points in $\mathcal E(S')$ (cf.~the proof of Lemma~\ref{z:next1}). So, if $T$ is a relocation of $\mathcal I(S^*)$, it means that the points of $W$ can reach their targets in $S$, contradicting the fact that $W$ is a non-satisfied analogy class of $S$, and $S$ is locked.
\item Let $S$ have no axes of symmetry. By Proposition~\ref{p:locked2}, $S$ is \UA, and $W$ consists of a single point $p$. By our assumption, $S^*$ is also \UA. Then, by definition of target in a \VA and \RD set, relocating $\mathcal F'(p)$ makes it join a concordance class of maximum size (more specifically, a concordance class whose number of points on $\mbox{SEC}(S')=S\setminus \{p\}$ is maximum). Therefore, by definition of target in a \VA and not \RD set, $p$ can reach its own target in $S$, which again contradicts the fact that $W$ is non-satisfied and $S$ is locked.
\end{itemize}
It follows that $T$ is not a relocation of $\mathcal I(S^*)$ and, by definition of tentative finish set, $R=P$.

Suppose that $S^*$ is not locked. Then, by definition of finish set, $R=L$, and hence $S''=S^*$. This implies that $S''$ is not locked, as required. Suppose now that $S^*$ is locked. Since $S$ is also locked, then, by Proposition~\ref{p:locked2}, there are two cases to consider.
\begin{itemize}
\item Let $S$ be \BA. Then $S$ has exactly one non-movable analogy class $U=\{p,q\}$, where $p$ and $q$ are consecutive. Let $\{p',q'\}=W$ be the unique unlocking analogy class of $S$, such that $p$ and $p'$ are consecutive. Let $\ell$ be the unique axis of symmetry of $S$, and let $\ell'$ be the line orthogonal to $\ell$ and passing through the center of SED$(S)$. Then, by Observation~\ref{o:unmovable} and by the symmetry of $S$, there is a half-plane bounded by $\ell'$ containing $p$ and $q$, and no other point of $S$. Recall that $S^*$ has $\ell$ as an axis of symmetry as well, and hence it is \BA. Since $S^*$ is locked too, and it is obtained from $S$ by relocating $\mathcal F'(W)$, it is easy to see that either $U$ is the non-movable analogy class of $S^*$, or $W$ and its relocation $L$ lie on opposite sides of $\ell'$, and $n=6$. However, in the latter case, the targets of $p'$ and $q'$ in $S$ lie on $\ell'$, implying that $p'$ and $q'$ can reach their targets, and therefore that either $W$ is not a movable non-satisfied analogy class, or $S$ is not locked. This is a contradiction, and hence $U$ is the non-movable analogy class of $S^*$. It follows that $P$ is the unlocking analogy class of $S^*$. Because $R=P$ (as argued above) and by definition of finish set, $L$ consists of two antipodal points lying on $\ell'$. But then $S''$ cannot be locked, because it has $\ell$ as an axis of symmetry (by Lemma~\ref{z:next1}), and no analogy class of $S''$ could be alone one side of $\ell'$, because $L\subset \ell'$.
\item Let $S$ be \UA. Then $S$ has at least one non-movable analogy class $\{p\}$ and, without loss of generality, $p$ is consecutive to $q$, where $W=\{q\}$. Let $r\in S$ be the other consecutive point of $p$ (note that $\{r\}$ is either a non-movable analogy class or an unlocking analogy class), and let $p'\in S$ be the other consecutive point of $q$. Recall that, since $S$ is \UA, then $S^*$ is \UA, as well. By Observation~\ref{o:unmovable}, an analogy class $\{c\}$ of $S$ (respectively, $S^*$) is non-movable if and only if the sum of the angular distances between $c$ and its two consecutive points in $S$ (respectively, $S^*$) is greater than $\pi$. Note that this sum, computed on $q$ with respect to $S$, is the same as the sum computed on the unique point of $P$ with respect to $S^*$. Also, since $|S\cap S^*|=n-1$, the only points of $S\cap S^*$ for which such a sum of angular distances may not be preserved in $S^*$ are $p$ and $p'$, because they are consecutive to $q$. It follows that the only possible non-movable analogy classes of $S^*$ are $\{p\}$, $\{p'\}$, and $\{r\}$ (the latter is non-movable in $S^*$ if and only if it is non-movable in $S$). Suppose that $P$ is not an unlocking analogy class of $S^*$. Therefore, by definition of unlocking analogy class, neither $\{p\}$ nor $\{p'\}$ is a non-movable analogy class of $S^*$. Also, since $R=P$, then, by definition of finish set, $R=L$, implying that $S''=S^*$. So, in this case, $S''$ is locked and it has fewer non-movable analogy classes than $S$. Suppose now that $P$ (and therefore $R$) is an unlocking analogy class of $S^*$. By definition of finish set, $L=\{r'\}$, where $r'$ is the antipodal of $r$ with respect to SEC$(S)$. So, $S''$ contains two antipodal points, $r$ and $r'$. If $S''$ is not \UA, then it has fewer analogy classes than $S$, and we are done. So, let $S''$ be \UA. Note that, by Proposition~\ref{p:finishreachable}, $r'$ is indeed reachable by $q$, and therefore the two consecutive points of $r'$ in $S''$ are $p$ and $p'$. So, by the previous argument on angular distance sums, it follows that, once again, the only analogy classes of $S''$ that could possibly be non-movable are $\{p\}$, $\{p'\}$, and $\{r\}$ (the latter if and only if it is non-movable also in $S$). But, since $r$ and $r'$ are antipodal, no analogy class of $S''$ other than $\{r\}$ can be non-movable (again, by the angular distance sum argument). Hence, $S''$ has fewer non-movable analogy classes than $S$.
\end{itemize}
\end{proof}

\begin{lemma}\label{z:valid3}
Let $\mathcal R$ be frozen at time $t_0$, let $\mathcal R(t_0)$ be a \VA configuration with $n>5$, and let the robots execute the \UCF algorithm. Then, the robots will eventually freeze in a \PR configuration without ever colliding.
\end{lemma}
\begin{proof}
Suppose for a contradiction that the robots never freeze in a \PR configuration. Then, we claim that there is a time $t_1\geqslant t_0$ at which the swarm is frozen in a \VA and \WA and \RD configuration in which all the walkers are on SEC/3 and the other robots are on SEC. Indeed, if $\mathcal R(t_0)$ is \VA and \RD, by Lemma~\ref{z:valid2} there is a time $t_0'\geqslant t_0$ at which the robots are frozen in a \VA and \WA and \RD configuration in which all the internal robots are on their finish lines. This configuration satisfies the hypotheses of either Lemma~\ref{z:valid1a} or Lemma~\ref{z:valid1b}. One of these two lemmas applies also if the configuration at time $t_0$ is not \RD. Hence, without loss of generality, we may assume that either Lemma~\ref{z:valid1a} or Lemma~\ref{z:valid1b} applies at time $t_0'$. If Lemma~\ref{z:valid1b} applies, then there is a time $t_0''\geqslant t_0'$ at which the all the robots are frozen on SEC, and therefore they satisfy the hypotheses of Lemma~\ref{z:valid1a}. Hence, without loss of generality, at time $t_0''$ Lemma~\ref{z:valid1a} applies. As a consequence, there is  a time $t_1\geqslant t_0''$ at which all the walkers are on SEC/3, and all the other robots are on SEC. This configuration is \VA and \WA and \RD, due to Observation~\ref{o:walkers}.

Subsequently, by Lemma~\ref{z:valid2}, all the internal robots of $\mathcal R(t_1)$ move to their corresponding finish lines (which remain unchanged during the movements) and freeze on SEC/3 at time $t_1'\geqslant t_1$. At this point, the configuration is \VA and \WA and \RD, due to Proposition~\ref{p:internalclass}, and either Lemma~\ref{z:valid1a} or Lemma~\ref{z:valid1b} applies, depending if the internal robots are all walkers or not. If the internal robots are walkers, then Lemma~\ref{z:valid1a} applies, and all the walkers freeze on SEC/3 at time $t_2\geqslant t_1'$. Otherwise, first the internal robots freeze on SEC at time $t_1''\geqslant t_1'$, due to Lemma~\ref{z:valid1b}. Afterwards, Lemma~\ref{z:valid1a} applies, and all the walkers of $\mathcal R(t_1'')$ move onto SEC/3, and freeze at time $t_2\geqslant t_1''$. Hence, in all cases, at time $t_2\geqslant t_1$ the new walkers are on SEC/3, and all the other robots are on SEC.

Note that at time $t_2$ the set of internal robots is not empty, because otherwise $\mathcal R(t_1'')$ would be an \EQ or \BI configuration (by definition of walker) with all robots of SEC. Hence it would be \PR, contradicting our assumptions. Also, $\mathcal R(t_1)$ and $\mathcal R(t_2)$ cannot be \EQ or \BI, otherwise they would not be \VA and \RD, due to Observation~\ref{o:analogy1}.

By repeating the previous argument, we infer that there exists a monotone sequence of time instants $(t_i)_{i>0}$ with the following properties, for all $i>0$.
\begin{itemize}
\item At time $t_i$, the configuration is \VA and \WA and \RD (hence not \EQ and not \BI), all walkers are frozen on SEC/3, and all other robots are frozen on SEC.
\item $\mathcal R(t_{i+1})$ is obtained from $\mathcal R(t_i)$ by first moving all the internal robots to their corresponding finish lines, and then sending all the non-walkers to SEC and all the walkers to SEC/3.
\end{itemize}

Let $S_i=\mathcal F(\mathcal R(t_i))$, for all $i>0$. Observe that $S_i$ and $S_{i+1}$ satisfy the hypotheses of Lemma~\ref{z:next1}, if we set $S:=S_i$ and $S'':=S_{i+1}$. Indeed, by definition of walker, $\mathcal W(\mathcal F(\mathcal R(t_i)))=\mathcal F(\mathcal W(\mathcal R(t_i)))$. Also, since $\mathcal R(t_i)$ cannot be \EQ or \BI, the set of walkers of $S_i$ is not empty. We are going to repeatedly apply Lemma~\ref{z:next1} to derive a contradiction, by arguing that either the number of analogy classes of the $S_i$'s decreases indefinitely as $i$ grows, or the number of their satisfied points grows indefinitely.

According to Lemma~\ref{z:next1}, the number of analogy classes of $S_i$ never increases as $i$ grows. Since this number cannot be smaller than $1$, there must be an index $a>0$ such that $S_i$ and $S_{i+1}$ have the same number of analogy classes, whenever $i\geqslant a$.

Let us choose an index $s$ as follows. If $S_i$ has an axis of symmetry for some $i\geqslant a$, then we let $s$ be any such $i$. Otherwise, we let $s=a$. Then, because axes of symmetry are preserved from $S_i$ to $S_{i+1}$ (by Lemma~\ref{z:next1}), it follows that, for all $i\geqslant s$, either both $S_i $ and $S_{i+1}$ are symmetric, or neither of them is.

Therefore, starting at index $s$, the $S_i$'s never go from asymmetric to symmetric, and the number of their analogy classes stays constant. As a consequence, Lemma~\ref{z:next1} implies that, for all $i\geqslant s$, $S_{i+1}$ has at least as many satisfied points as $S_i$. But the number of satisfied points of $S_i$ is bounded by the number of robots in the swarm, $n$, and so there must be an index $m\geqslant s$ such that $S_i$ and $S_{i+1}$ have the same number of satisfied points, whenever $i\geqslant m$.

We claim that there is an index $u\geqslant m$ such that $S_u$ is not locked. Assume the opposite. Then we can apply Lemma~\ref{p:unlocked}, with $S:=S_{m+i}$ and $S'':=S_{m+i+1}$, for all $i\geqslant 0$. So, either $S_{m+i+1}$ is not locked (which contradicts our assumption), or it has strictly fewer analogy classes than $S_{m+i}$ (which contradicts the fact that $m+i\geqslant a$), or it has fewer non-movable analogy classes than $S_{m+i}$. Hence there must be some $i\geqslant 0$ such that $S_{m+i}$ has no non-movable analogy classes. But, by Proposition~\ref{p:locked1}, such an $S_{m+i}$ is not locked, contradicting our assumption again. Therefore $S_u$ is not locked for some $u\geqslant m$, and Lemma~\ref{z:next1} states that $S_{u+1}$ has strictly more satisfied points than $S_u$, contradicting the definition of $m$.
\end{proof}

\begin{theorem}\label{main}
The \CF problem is solvable by $n > 5 $ robots in \ASYNC.
\end{theorem}
\begin{proof}
We apply the \UCF algorithm of Section~\ref{sec:code}. Recall that the initial configuration is frozen. If the robots are frozen in a \CO or \CE or \HD configuration, they freeze in a \VA or \IN configuration, due to Lemmas~\ref{z:coradial} and~\ref{z:hdce}. If the robots are frozen in a \VA or \IN configuration, they freeze in a \PR configuration, due to Lemmas~\ref{z:invalid} and~\ref{z:valid3}. If the robots are frozen in a \PR configuration, they freeze in a \RE configuration, due to Lemma~\ref{z:prereg}. Finally, if the robots are frozen in a \RE configuration, they remain still forever, due to Lemma~\ref{z:regular}. Therefore the \CF is solvable for $n>5$.
\end{proof}

\subsection{Small Swarms}
We have just shown how the \CF can be solved by $n>5$ robots. We now consider the cases of small swarms.

\begin{theorem}\label{z:n3}
The \CF problem is solvable by $n=3$ robots in \ASYNC.
\end{theorem}
\begin{proof}
We use the following algorithm:
\begin{itemize}
\item if the three distances between pairs of robots are all distinct and robots $r_1$ and $r_2$ are farthest apart, then robot $r_3$ moves parallel to $r_1r_2$ toward the axis of $r_1r_2$;
\item otherwise, if $r_1r_3=r_2r_3$, then $r_3$ moves to the closest point that forms an equilateral triangle with $r_1$ and $r_2$ (in case there are two such points, one is chosen arbitrarily).
\end{itemize}
In the first case, robot $r_3$ moves orthogonally to the axis of $r_1r_2$. While this happens, $r_1$ and $r_2$ remain the farthest-apart robots, and $r_3$ keeps being the robot that has to move. Eventually $r_3$ reaches the axis of $r_1r_2$, it freezes, and the configuration transitions to the second case, with $r_1r_3=r_2r_3$.

If the robots are frozen and $r_1r_3=r_2r_3$, then robot $r_3$ moves orthogonally to $r_1r_2$. While this happens, $r_3$ remains equidistant from $r_1$ and $r_2$ and keeps being the robot that has to move. When $r_3$ reaches the point that forms a \RE set with the other two robots, it freezes.
\end{proof}

\begin{lemma}\label{z:antip}
Let $S$ be a \UA set of $n=5$ points, all of which lie on SEC, and no two of which are antipodal. Then there exists a movable point of $S$ that can reach the antipodal of another point of $S$.
\end{lemma}
\begin{proof}
By Observation~\ref{o:analogy2}, every analogy class of $S$ consists of a single point, and therefore, with a slight abuse of terminology, we may refer to movable and non-movable points (as opposed to analogy classes). By Proposition~\ref{p:consunmovable}, if there are two non-movable points in $S$, they are consecutive, and hence there are at most two non-movable points. Let $p_1$, $p_2$, $p_3$, $p_4$, $p_5$ be the points of $S$, appearing in this order around the center of SED. Without loss of generality, we may assume that $p_2$ and $p_4$ are movable. Suppose for a contradiction that neither of these two points can reach the antipodal of another point of $S$. Let $p_i'$ be the antipodal of $p_i$ with respect to SEC, for $1\leqslant i\leqslant 5$, and let $S'=\{p_i'\mid 1\leqslant i\leqslant 5\}$. Since $p_2$ cannot reach any $p_i'$, the arc $\overset\frown{p_1p_3}$ is devoid of points of $S'$. Similarly, since $p_4$ cannot reach any $p_i'$, the arc $\overset\frown{p_3p_5}$ is devoid of points of $S'$. Because no two points of $S$ are antipodal, the endpoints of these arcs cannot be in $S'$, either. It follows that the whole closed arc $\overset\frown{p_1p_5}$ is devoid of points of $S'$. Note that the arc $\overset\frown{p_5p_1}$ is strictly shorter than a half-circle, due to Observation~\ref{l:sechull} (it cannot be a half-circle, otherwise $p_1$ and $p_5$ would be antipodal). Therefore the arc $\overset\frown{p_1p_5}$ is strictly longer than a half-circle, and hence it contains both $p_1'$ and $p_5'$, which is a contradiction.
\end{proof}

\begin{theorem}\label{z:n5}
The \CF problem is solvable by $n=5$ robots in \ASYNC.
\end{theorem}
\begin{proof}
We use a modified version of the general algorithm of Section~\ref{sec:code}. Note that the proof of correctness holds for the case $n=5$ as well, except for Lemmas~\ref{z:next1},~\ref{p:unlocked}, and~\ref{z:valid3}, which all assume that $n>5$. This is due essentially to the last sentence of Proposition~\ref{p:locked2}, which express a property of locked configurations of $n>5$ points. The core problem is that, if $n=5$, there are locked configurations in which all the robots that belong to unlocking analogy classes happen to be satisfied. Recall that, for $n=5$, the definition of walker allows the selection of a satisfied unlocking analogy class as the set of walkers. On one hand, this prevents us from arguing that the number of satisfied robots cannot decrease after a certain point, as we did in Lemmas~\ref{z:next1} and~\ref{z:valid3}. On the other hand, the current definition of finish set will allow such walkers to go back into their targets right away. This causes the same locked configuration to be formed infinitely many times, rendering the statement of Lemma~\ref{p:unlocked} false, and giving rise to an infinite loop in the execution.

We can fix the algorithm as follows: if $n=5$, we retain all the definitions as they are, except for the definition of walker and the definition of finish set. Assuming that the configuration $S$ is a \VA set with all $n=5$ points on SEC, the walkers are selected as usual, except in the following cases.
\begin{itemize}
\item Let $S$ be \UA with no pairs of antipodal points. Then, among the movable points of $S$ that can reach the antipodal point of another point of $S$, the walker is the one that induces the lexicographically smallest angle sequence (such a point exists due to Lemma~\ref{z:antip}).
\item Let $S$ be \UA with exactly one pair of antipodal points. Then, the walker is the unique point that is consecutive to the two antipodal points (such a point exists because a \VA set is not \HD).
\item Let $S$ be \UA with two pairs of antipodal points. Then, the walker is the unique point of $S$ that is not antipodal to any other point of $S$.
\item Let $S$ be \BA with exactly one satisfied point, and having two antipodal analogous points. Then, the walkers are the two analogous points that are not antipodal.
\end{itemize}
Note that in every case the walkers constitute a movable analogy class, in accordance with Observation~\ref{o:walkers}.

Now to the definition of finish set. Suppose that the set $S$ is \VA and \RD and has $n=5$ points. Let $S'=\mathcal E(S)\cup P$, where $P$ is the principal relocation of $\mathcal I(S)$. Then, the finish set is defined as usual, except in the following cases.
\begin{itemize}
\item Let $S'$ be \UA, and suppose that there exists at least one point of $S'$ whose antipodal point can be reached by the unique point of $P$. Then, let $p\in S'$ be the one among such points that induces the lexicographically smallest angle sequence with respect to $S'$. By definition, the finish line corresponding to the internal point of $S$ contains the antipodal point of $p$.
\item Let $S'$ be \BA, let $P$ consist of two non-consecutive points, and let the two consecutive analogous points of $S'$ be non-satisfied. Then, let $R$ be the relocation of $\mathcal I(S)$ consisting of two antipodal points on SEC$(S)$ such that $R$ is an analogy class of $\mathcal E(S)\cup R$. By definition, $R$ is a subset of the finish set of $\mathcal I(S)$.
\end{itemize}
Note that in both cases each finish line is reachable by exactly one internal point (cf.~Proposition~\ref{p:finishreachable}).

Let us prove that the above modifications to the general algorithm are sufficient to solve the \CF problem for $n=5$ robots. Note that, if the robots ever freeze in a \PR configuration, they also freeze in a \RE configuration, due to Lemma~\ref{z:prereg}, and then they remain still forever, due to Lemma~\ref{z:regular}. So, suppose for a contradiction that they never freeze in a \PR configuration. If the robots are frozen in a \CO or \CE or \HD configuration, they freeze in a \VA or \IN configuration, due to Lemmas~\ref{z:coradial} and~\ref{z:hdce}. If they are frozen in an \IN configuration, they freeze in a \VA configuration, by Lemma~\ref{z:invalid}. Hence, assume that the robots are frozen in a \VA configuration at time $t_0$, and assume for a contradiction that they never form a \RE configuration. So, as in the proof of Lemma~\ref{z:valid3}, we can construct a monotone sequence of time instants $(t_i)_{i>0}$ with the same properties (note that only Lemmas~\ref{z:transition2}--\ref{z:valid2} are used to prove this part, and they hold also for $n=5$). Again, let $S_i=\mathcal F(\mathcal R(t_i))$.

Suppose that there exists an index $s$ such that $S_s$ has an axis of symmetry. Following the proof of Lemma~\ref{z:next1}, we argue that $S_{s+i}$, for all $i\geqslant 0$, has the same axis of symmetry and the same target set. Let $m\geqslant s$ be such that the number of satisfied points in $S_m$ is maximum. Suppose first that $S_m$ is not locked. In all non-locked \BA cases, including the newly added one, the walkers are non-satisfied points that can reach their corresponding targets. Since $S_m$ is symmetric and the unique point on the axis of symmetry is satisfied, it follows that there are exactly two walkers in $S_m$. If the two walkers are non-consecutive, so are the elements of their principal relocation (of their anti-footprints). In this case, if the two other analogous points of $S_m$ are non-satisfied, the new definition of finish set applies. Therefore, in $S_{m+1}$ there is exactly one satisfied point and two antipodal analogous points. Now, according to the new algorithm, the walkers are the two analogous points that are not antipodal. In $S_{m+2}$ these two points are moved to their targets. Then the two antipodal points are selected as walkers, and are moved to their targets in $S_{m+3}$, thus forming a \RE configuration, which contradicts our assumptions. In all other non-locked \BA cases, the walkers cannot give rise to a locked configuration by moving to their targets, nor can their principal relocation be a proper subset of an analogy class, because $n=5$ and analogy classes can have at most two points each. Therefore, in all these cases, the walkers of $S_m$ choose finish lines that contain their targets. Hence the number of satisfied points in $S_{m+1}$ increases, which contradicts the definition of $m$. If, on the other hand, $S_m$ is locked, the two points of the unlocking analogy class are selected as walkers (indeed, the new \BA rule does not apply to this case, because if two points of $S_m$ are antipodal, then $S_m$ cannot be locked). These two points are non-consecutive, and perhaps are satisfied. Note that the other two analogous points of $S_m$ are not satisfied, otherwise the configuration would not be locked. Here the new definition of finish lines applies; arguing as above, we conclude that $S_{m+3}$ is \RE, which is a contradiction.

Suppose now that $S_i$ has no axis of symmetry for any $i>0$. Assume that, for some index $a$, there are two pairs of antipodal points in $S_a$. According to the new algorithm, the walker is the point that is not antipodal to any other. The principal relocation $\{p\}$ of the anti-footprint of the walker gives rise to a symmetric configuration, and the chosen finish line contains $p$. Therefore, $S_{a+1}$ has an axis of symmetry, which contradicts our assumptions. Suppose now that in $S_a$ there is exactly one pair of antipodal points. According to the new algorithm, the walker is the point that is consecutive to both antipodal points. The principal relocation $\{p\}$ of the anti-footprint of the walker gives rise to a configuration $S'$. If $S'$ is symmetric and $p$ lies on the axis of symmetry, then that is the walker's target, which is also chosen as a finish line (note that $S'$ cannot be locked, due to the two antipodal points). Hence $S_{a+1}$ is symmetric, which is a contradiction. Now let $S'$ be symmetric, and suppose that $p$ does not lie on the axis of symmetry. Then, $\{p\}$ must be a proper subset of an analogy class of $S'$, and therefore the tentative finish set of $S'$ is $\{p\}$. Also note that, if $S'$ is locked, $\{p\}$ cannot be an unlocking analogy class of $S'$, because it contains only one point (cf.~Proposition~\ref{p:locked2}). Therefore $p$ lies on the finish line, by definition. Hence $S_{a+1}$ is symmetric, which is again a contradiction. Suppose now that $S'$ is not symmetric, and therefore it is \UA. Note that $p$ can reach the antipodal of another point of $S'$, and hence it is moved to such a point, according to the new definition of finish set. Then in $S_{a+1}$ there are two pairs of antipodal points, and we already proved that this leads to a contradiction. Finally, assume that in $S_1$ there are no pairs of antipodal points. By the new algorithm, the walker is a single movable point that can reach the antipodal of another point of $S_1$. The principal relocation $\{p\}$ of the anti-footprint of the walker gives rise to configuration a $S'$. If $S'$ is \UA, the new algorithm chooses a finish line containing the antipodal of some point. Hence in $S_2$ there are exactly two antipodal points, and the previous argument applies. Suppose then that $S'$ is \BA. If $p$ lies on the axis of symmetry of $S'$, then it is satisfied, and the tentative finish set is $\{p\}$. Note that, if $S'$ is locked, then $\{p\}$ cannot be the unlocking analogy class, because it only has one point (cf.~Proposition~\ref{p:locked2}). Therefore, the finish line contains $p$, by the usual definition. Hence $S_2$ is symmetric, which is a contradiction. Suppose now that $p$ does not lie on the axis of symmetry of $S'$. So, $\{p\}$ is a proper subset of an analogy class, and hence the tentative finish set is $\{p\}$. Once again, if $S'$ is locked, $\{p\}$ cannot be an unlocking analogy class, and hence the finish line contains $p$. Then $S_2$ is symmetric, which is a contradiction.
\end{proof}

\section{Conclusions}

By Theorems~\ref{main},~\ref{z:n3},~\ref{z:n5}, and by the result in~\cite{MamV16}, which deals with the special case of $n=4$ robots, it follows that

\newtheorem{theorem2}{Theorem}[section]

\begin{theorem2}\label{final} 
The \CF problem is solvable in \ASYNC.\qed
\end{theorem2}

Recall that no pattern other than {\tt Point} and {\tt Uniform Circle} can be formed from every initial configuration, even
if the system is fully synchronous, the robots are provided with chirality, and the adversarial scheduler does not have the power of interrupting the robots' movements (rigidity). In light of the  result of~\cite{CieFPS12} for {\tt Point}, Theorem \ref{final} 
  implies that  asynchrony is not a computational handicap, and that additional powers such as chirality and rigidity are computationally irrelevant.

\ \\
\noindent {\bf Acknowledgments.} {\small 
The authors would like to thank Marc-Andr\'e Paris-Cloutier for many helpful discussions and insights, and Peter Widmayer and Vincenzo Gervasi for sharing some of the fun and  frustrations emerging from investigating this problem.
This work has been supported in part by the Natural Sciences and Engineering Research Council of Canada under the Discovery Grants program, and by Professor Flocchini's University Research Chair.
}

\end{document}